\documentclass[11pt]{article}
\usepackage[T1]{fontenc}
\usepackage{graphicx}
\usepackage{fullpage}
\usepackage{xcolor}
\usepackage{amsmath,amssymb,amsfonts,amsthm,enumerate}
\usepackage{nicefrac}
\usepackage{hyperref,cleveref}
\usepackage{authblk}
\usepackage{times}
\usepackage{xspace}
\usepackage{bbm}
\usepackage[margin=1in]{geometry}
\usepackage[procnumbered, linesnumbered,algoruled,noend]{algorithm2e}
\usepackage{multirow}
\usepackage{amsmath, amsthm, amsfonts, mathdots, thm-restate, mathrsfs, mathtools}
\usepackage{enumitem}

\def \bfx {\mathbf{x}}
\newcommand*{\email}[1]{\href{mailto:#1}{\nolinkurl{#1}} }

\newcommand{\cM}{\mathcal{M}}
\newcommand{\cI}{\mathcal{I}}
\newcommand{\cB}{\mathcal{B}}
\newcommand{\cD}{\mathcal{D}}
\newcommand{\cP}{\mathcal{P}}
\newcommand{\tU}{\tilde{U}}
\newcommand{\indicator}{\mathbf{1}}
\newcommand{\mab}{{Multi-Arm Bandits}\xspace}
\newcommand{\sw}{{Submodular Welfare}\xspace}
\newcommand{\bem}{{base exchange map}\xspace}

\newcommand{\ignore}[1]{}

\newcommand{\eps}{\varepsilon}
\crefname{conjecture}{conjecture}{conjectures}
\Crefname{conjecture}{Conjecture}{Conjectures}
\usepackage[caption=false]{subfig}
\newcommand{\ceil}[1]{ {\left\lceil #1 \right\rceil}}
\DeclareMathOperator*{\argmax}{arg\,max}

\newtheorem{theorem}{Theorem}[section]
\newtheorem{lemma}[theorem]{Lemma}
\newtheorem{corollary}[theorem]{Corollary}

\newtheorem{definition}[theorem]{Definition}
\newtheorem{claim}[theorem]{Claim}
\newtheorem{observation}[theorem]{Observation}

\newcommand{\tI}{\tilde{I}}

\newcommand{\marsum}{\textnormal{\texttt{MarRatio}}}

\newcommand{\app}{\textnormal{\texttt{approx}}}

\newcommand{\simpleF}{\textnormal{\texttt{Partition Swap $F$}}\xspace}
\newcommand{\genF}{\textnormal{\texttt{Generalized Partition Swap $F$}}}

\newcommand{\genf}{\textnormal{\texttt{Generalized Partition Swap $f$}}}

\newcommand{\BanditSwap}{\textnormal{\texttt{SimpleBanditSwap}}}

\newcommand{\ApproxKnapsack}{\textnormal{\texttt{ApproxPack}}\xspace}

\newcommand{\checkpoints}{\textnormal{\texttt{CP}}}
\newcommand{\swap}{\textnormal{\texttt{swap}}}
\newcommand{\abs}[1]{{\left| #1\right|}}
\newcommand{\E}{\mathbb{E}}
\newcommand{\cF}{\mathcal{F}}

\newcommand{\OPT}{O}

\newcommand{\simpFprob}{\textnormal{Partition-$F$}\xspace}
\newcommand{\simpfprob}{\textnormal{Partition-$f$}\xspace}
\newcommand{\partFprob}{\textnormal{General-$F$}}
\newcommand{\partfprob}{\textnormal{General-$f$}}

\newcommand{\gplong}{{Greedy Swap Poisson Process}\xspace}
\newcommand{\gp}{{GS-Poisson}\xspace}
\newcommand{\gaplong}{{Generalized Assignment Problem}\xspace}
\newcommand{\gap}{{GAP}\xspace}
\newcommand{\saplong}{{Separable Assignment Problem}\xspace}
\newcommand{\sap}{{SAP}\xspace}
\newcommand{\sm}{{Sub-Mat}\xspace}
\newcommand{\bx}{\mathbf{x}}
\newcommand{\by}{\mathbf{y}}

\def \F{\mathcal{F}}
\def \A{\mathcal{A}}

\begin{document}

\title{A Poisson Process for Submodular Maximization}
\author{Amit Ganz Rozenman}
\affil{
  {Technion},{Haifa},{Israel}.
\email{amitganz3@gmail.com}}

\author{Ariel Kulik}
\affil{
  {Ben-Gurion University of the Negev},
  {Beer-Sheva},
  {Israel}.
\email{kulik@bgu.ac.il}
}
\author{Roy Schwartz}
\affil{%
  {Technion},
  {Haifa},
{Israel}.
\email{schwartz@cs.technion.ac.il}
}
\author{Mohit Singh}
\affil{%
  {Georgia Institute of Technology},
  {Atlanta},
  {USA}.
\email{mohit.singh@isye.gatech.edu}
}
\maketitle

\begin{abstract}
We study the problem of maximizing a monotone submodular function subject to a matroid independence constraint. 
For more than a decade, a rich body of work has studied this problem.
Initially, a tight approximation of $ (1-\nicefrac{1}{e})$ was given using the continuous greedy algorithm [Calinescu-Chekuri-Pal-Vondr{\'a}k STOC`2008] and later non-oblivious local search techniques were able to match this tight approximation guarantee [Filmus-Ward FOCS`2012] and [Buchbinder-Feldman FOCS`2024].

We propose a new and remarkably simple approach to this problem that is based on a stochastic Poisson process.
Our approach matches the tight $ (1-\nicefrac{1}{e})$ approximation guarantee and it differs from the known two techniques since it does not require discretization or rounding while performing very few single element swaps.
We also present applications of our approach and obtain fast algorithms for submodular welfare maximization, and for the general and separable assignment problems.

\end{abstract}

\section{Introduction}\label{sec:intro}
Submodular maximization has been a central topic in combinatorial optimization for several decades and its study dates back to the late $70$'s \cite{NW78,nemhauser1978analysis,fisher1978analysis}.
The appeal of submodular functions lies in their diminishing returns  property: a set function $ f\colon 2^U\rightarrow \mathbb{R}_+$ over a universe $U$ is submodular if it satisfies that $ f(S+i)-f(S)\geq f(T+i) -f(T)$ for every $S\subseteq T\subseteq U$ and $i\in U\setminus T$.\footnote{We use $S+i$ and $S-i$ to denote $S\cup \{ i\}$ and $S\setminus \{i\}$, respectively, $\forall S\subseteq U$ and $\forall i\in U$.}
Submodular functions naturally arise in many disciplines, including rank in linear algebra, graph cuts in combinatorics, and entropy in probability theory.
Moreover, submodular functions play a major role in many real world applications such as sensor placement and data summarization~\cite{KLGVF08,KSG08,LB10,LB11,mirza16}, influence maximization in social networks~\cite{kempe2003maximizing,HMS08,MR10} and feature selection in machine learning~\cite{LWKSB13,KEDNG17,BHZ22} (see also, e.g., a more comprehensive survey \cite{B13})

In this work we consider the problem of maximizing a monotone\footnote{A set function $f\colon2^U\rightarrow \mathbb{R}_+$ is monotone if $ f(S)\leq f(T)$ for every $S\subseteq T\subseteq U $.} submodular function $f\colon 2^U\rightarrow \mathbb{R}_+$ subject to a matroid $ \cM=(U,\cI)$ independence constraint (denote this problem by \sm).
A notable running example is that of a partition matroid, where the universe $U$ is partitioned into $(U_1,\ldots,U_k)$ and the matroid independence constraint reduces to finding a set $S\subseteq U$ that contains at most a single element from each $U_j$ in the partition.
This special case is already sufficiently rich in structure, as it captures well-known problems such as the \gaplong~\cite{fleischer2006tight,feige2006approximation,CCPV11} and \sw~\cite{DS06,vondrak2008optimal,CCPV11,FNS11,V13}.

The main source of progress for the \sm problem has been the introduction of continuous methods via the multilinear extension~\cite{vondrak2008optimal,CCPV11}.
The multilinear extension $F\colon [0,1]^U\rightarrow \mathbb{R}_+$ is defined as:
$F(\bfx)\triangleq \sum _{S\subseteq U}f(S)\cdot \Pi _{i\in S}x_i \cdot \Pi _{i\notin S}(1-x_i) $.
It carries an intuitive probabilistic interpretation of $ F(\bfx)=\mathbb{E}_{R\sim \bfx}[f(R)]$, where $ R\sim \bfx$ denotes a random subset in which every element $i\in U$ is independently chosen to $R$ with a probability of $x_i$.
In these continuous methods, one approximately solves a fractional relaxation of the problem whose objective is the multilinear extension, that itself is NP-hard, and then converts the fractional solution to an integral one via rounding.
Besides obtaining a tight $(1-\nicefrac{1}{e})$-approximation for the \sm problem~\cite{CCPV11,NW78}, these continuous approaches have been successfully extended to accommodate various constraints, including multiple knapsacks, multiple matroids, and even more complex independence systems, e.g., ~\cite{chekuri2011submodular,kulik2013approximations,BV14,FKS21}.

Unfortunately, the continuous approach described above has several drawbacks.
First, the implementation of continuous methods as an algorithm involves discretizing a continuous time process.
Thus, in order to bound the error caused by this discretization one is required to perform sufficiently small steps, resulting usually in a large number of iterations.
Second, these methods require numerous evaluations of the gradient of the multilinear extension $F$.
Typically, every such gradient evaluation requires a large number of evaluations of the original function $f$, which can be computationally expensive (especially for large-scale problems).
Third, and of no less importance, the continuous viewpoint (multilinear extension) and the discrete viewpoint (rounding) are applied in sequence and in isolation.
While this can be achieved with no loss when considering the \sm problem, it can be time consuming.

Alongside greedy algorithms, local search methods have also been developed,  as early as~\cite{fisher1978analysis}, as an alternative way to approach constrained submodular maximization.
More recently, local search methods have been able to achieve some of the results first obtained by the continuous greedy based approach, e.g., an asymptotically tight approximation for \sm~\cite{filmus2014monotone,BF24deter}.
The improved local search methods work on a modified objective that approximates the true objective $f$. 
The main benefit of local search approaches is that they typically work directly with the function $f$ and therefore do not require gradient computation and rounding.
The challenge of these approaches is that the running times usually have large dependence on $\nicefrac{1}{\varepsilon}$~\cite{filmus2014monotone,BF24deter} to obtain approximation guarantees that are within $\varepsilon$ of the best possible.

\paragraph{Our Hybrid Approach.}
We present a hybrid approach that aims to combine the strengths of both continuous and discrete approaches. Our algorithm maintains a feasible base of the matroid at every iteration and updates this solution through a sequence of elementary exchanges, replacing one element at a time while preserving feasibility.
What sets our approach apart from local search methods is the way it implements these updates.
The timing of each potential update is governed by a non-homogeneous Poisson process.
Meanwhile, the decision of which elements to exchange is informed by dynamically computed weights based on the gradient of the multilinear extension evaluated at an appropriate time scaling of the current base.

This approach gives a new tight algorithm for \sm as well as additional applications including fast algorithms for the \gaplong and \sw.

\subsection{Our Results}\label{sec:results}
We present a remarkably simple stochastic process based on a non-homogeneous Poisson process, which we denote by \gplong (\gp), that yields a tight approximation for the \sm problem. 
The stochastic process is so simple that we can state it now.
To simplify presentation, we denote by $ \nabla _i F(\bfx)$ the coordinate that corresponds to element $i\in U$ of the gradient of the multilinear extension at point $ \bfx\in [0,1]^U$ and by
$ \indicator _A\in \{ 0,1\}^U$ the indicator vector of $ A\subseteq U$. 
The internal state of \gp consists of a single base $A \in \cB$ that is initialized to be some arbitrary base ($\cB$ denotes the collection of all bases of $\cM$).
Starting from time $t=\varepsilon$, a non-homogeneous Poisson process with rate $\lambda(t) = \nicefrac{k}{t}$, where $k$ is the rank of $\cM$, is executed.
If the next event of the Poisson process occurs at time $t$:
\begin{enumerate}
    \item A base $Z(t)\in \cB$ that maximizes $ \sum _{i\in Z(t)}\nabla _i F(t\indicator _{A})$ is computed.\label{GP-First-Step}  
    \item A bijection $h : A \to Z(t)$ is constructed such that $A-i+h(i)\in \cB$ 
    for every $i \in A$.
    \item A single, uniformly random, element $i \in A$ is selected and swapped with $h(i)$:  $A\leftarrow A-i+h(i)$. 
\end{enumerate}
The above simple steps are repeated every time an event occurs, stopping at time $t=1$.

There are three key features that set \gp apart from prior approaches:
\begin{enumerate}[label=(\roman*)]
    \item \textbf{No discretization:} Time is not discretized since the Poisson process is simulated directly by sampling the time of the next event.\label{gp-discretization}
    \item \textbf{No rounding:} The algorithm operates directly on bases; no rounding step is required.\label{gp-rounding}
    \item \textbf{Few updates:} The expected number of single element swaps is only $k \ln(1/\varepsilon)$.\label{gp-steps}
\end{enumerate}
We note that \ref{gp-discretization} and \ref{gp-rounding} above set \gp apart from the continuous greedy approach, whereas \ref{gp-steps} sets it apart from the local search approach.

The following theorem states the tight approximation guarantee of \gp.
\begin{theorem}\label{thrm:matroid}
Given a matroid $ \cM=(U,\cI)$, a monotone submodular function $ f\colon 2^U\rightarrow \mathbb{R}_+$, and a starting time $ 0<\varepsilon \leq 1$, \gp finds $ A\in \cI$ satisfying $ \mathbb{E}[f(A)]\geq (1-\varepsilon)(1-\nicefrac{1}{e})f(O)$.  
Here $O=\argmax \{ f(S):S\in \cI\}$ is an optimal base.
\end{theorem}

We make two important observations.
The first observation is the close connection of \gp to the continuous greedy algorithm, rather than local search methods such as~\cite{filmus2014monotone,BF24deter}.
One approach to modify the continuous greedy algorithm while maintaining only an integral solution is to round on the fly: only maintain a (randomized) integral solution obtained by rounding the current fractional solution.
Unfortunately, this approach does not work directly.
\gp is in fact rounding on the fly, but of an appropriately chosen scaling of the fractional solution maintained by the continuous greedy algorithm.
Thus, \gp can also be interpreted as rounding on the fly of a time-scaled version of the continuous greedy algorithm. 
From this point of view, the continuous greedy naturally becomes a Frank-Wolfe style continuous local search algorithm. We discuss this connection in Section~\ref{sec:techniques} in further detail.

The second observation is that \gp does not require the best base $Z(t)$ as defined above in step~\eqref{GP-First-Step} in order to output a good approximation.
Instead, a carefully chosen
pair of elements $i\in A$ and $j\in U$ to be swapped is enough. In particular, the construction of the whole base $Z(t)$ can be avoided. 
We refer to the procedure which finds the pair $i$ and $j$ as a swap procedure. 
This simple observation will be the main driving force behind our additional applications.

\subsubsection{Applications}
Our applications include fast algorithms, based on \gp, for partition and generalized partition matroids, as well as the \gaplong and \saplong.
In what follows, when given a submodular function, the running time of the algorithm is measured by the number of function evaluations it performs, as this is the standard measure for these types of problems.
We assume oracle access to either $f$ or $F$ and consider results for both.

\paragraph{Partition Matroid.}
When $\cM=(U,\cI)$ is a partition matroid there exists a partition $(U_1,\ldots, U_k)$ of $U$ and a set $S\in\cI$ if and only if $|S\cap U_j|\leq 1$ for each $1\leq j\leq k$.
Thus, a solution $S\subseteq U$ to \sm when $\cM$ is a partition matroid is feasible if and only if $ S$ contains at most a single element from each part $U_j$. 

The case of a partition matroid is of special interest, as it captures the classic problem of \sw.
In the \sw problem, we are given $k$ players, a set of $m$ items $I$ and a monotone non-negative submodular utility function $f_j:2^{I}\rightarrow \mathbb{R}_+$ for each player $1\leq j\leq k$.
The goal is to find disjoint subsets $S_1,\ldots, S_k\subseteq I$ that maximize the sum of the utilities: $\sum_{j=1}^k f_j(S_j)$.
Here, $S_j$ denotes the items assigned to player $j$.
A simple approximation preserving reduction (see, e.g., \cite{CCPV11}) shows that \sw can be reduced to \sm with a partition matroid of rank $m$ and a universe of size $mk$.

The following two theorems summarize the case when $\cM$ is a partition matroid, the first when the algorithm can query the multilinear extension $F$ and the second when the algorithm can query the original set function $f$.
\begin{theorem}\label{thrm:PartitionF}
There exists an algorithm that given a partition matroid $\cM=(U,\cI)$, a monotone submodular function $f\colon 2^U\rightarrow \mathbb{R}_+$, and $ 0< \varepsilon \leq 1$, performs in expectation $ O(n\ln{(\nicefrac{1}{\varepsilon})})$ evaluations of $F$ and finds $A\in \cI$ satisfying: $ \mathbb{E}[f(A)]\geq (1-\varepsilon)(1-\nicefrac{1}{e}) f(O)$.
Here $n$ is the size of the universe $U$ and $ O=\argmax \{ f(S)\colon S\in \cI\}$ is an optimal base.
\end{theorem}
\begin{theorem}\label{thrm:partition}
There exists an algorithm that given a partition matroid $ \cM=(U,\cI)$, a monotone submodular function $ f\colon 2^U\rightarrow \mathbb{R}_+$, and $ 0<\varepsilon<\nicefrac{1}{4}$, performs in expectation $O(n \log^2(1/\varepsilon)/\varepsilon^2)$ evaluations of $f$ and finds $ A\in \cI$ satisfying $ \mathbb{E}[f(A)]\geq (1-\varepsilon)(1-\nicefrac{1}{e})f(O)$.
Here $n$ is the size of the universe $U$ and $ O=\argmax \{ f(S):S\in \cI\}$ is an optimal base.
\end{theorem}
Let us compare the results of the above theorems with previously known tight $ (1-\varepsilon)(1-\nicefrac{1}{e})$ approximations.
When querying $F$, the best known algorithm was given (implicitly) by~\cite{BV14} and it requires $ O(n/\varepsilon)$ evaluations of $F$, while Theorem~\ref{thrm:PartitionF} requires only $ O(n\ln{(\nicefrac{1}{\varepsilon})})$ evaluations in expectation.
When considering $f$, the best known algorithm was given by~\cite{EN19} and it requires $ O(n\log^2{(n/\varepsilon)}/\varepsilon^5)$ evaluations of $f$ (see also~\cite{KT24,BF24deter} and the discussion in Section \ref{sec:additional} for additional incomparable bounds), where Theorem~\ref{thrm:partition} requires only $O(n \log^2(1/\varepsilon)/\varepsilon^2) $ evaluations in expectation.
It should be noted that Theorem~\ref{thrm:partition} provides the first linear-time algorithm for the problem.

It is worth noting that in order to prove Theorem~\ref{thrm:partition} we establish an interesting connection with the classic \mab problem~\cite{EMM06}.
This enables us to { obtain the first linear time algorithm for the problem.}

\paragraph{The General Assignment and Separable Assignment Problems.}
An instance of the \gaplong (\gap) consists of $m$ bins and $n$ items.
Each item $j$ and bin $i$ have two non-negative quantities associated with them: a value $v_{ij}$ and a size $s_{ij}$.
We aim to assign items to bins such that the total size of items in each bin is at most one, and the total value of all items is maximized.
Formally, the goal is to choose disjoint sets $S_1,\ldots,S_m$ of items satisfying $ \sum _{j\in S_i}s_{ij}\leq 1$ for every $ i=1,\ldots,m$, while maximizing: $\sum _{i=1}^m \sum _{j\in S_i}v_{ij}$.
We prove the following theorem.

\begin{theorem}\label{thm:gap_main}
For every $ \varepsilon >0$ there exists a $(1-\varepsilon)(1-\nicefrac{1}{e})$-approximation algorithm for the \gaplong with $n$ items and $m$ bins that runs in time $ \tilde{O}((\nicefrac{n}{\varepsilon}+\nicefrac{1}{\varepsilon^2})\cdot m)$.
\end{theorem}

Our approach in proving the above is considering a more general problem, the \saplong (\sap).
An instance of \sap consists of $n$ items and $m$ bins.
Each bin $i$ has an associated collection of feasible sets $\F_i\subseteq 2^{[n]}$ which is down-closed: $A \in \F_i$ and  $B\subseteq A$ imply that $B \in \F_i$.
Each item $j$ and bin $i$ have a value $v_{ij}$.
The goal is to choose disjoint feasible sets $S_i \in \F_i$ of items while maximizing: $\sum_{i=1}^m \sum_{j\in S_i}v_{ij}$.
Observe that \sap captures \gap when each $\F_i$ is the collection of all subsets of items whose size with respect to bin $i$ is at most one: $\F_i = \{ S\subseteq [n]\colon \sum _{j\in S}s_{ij}\leq 1\}$.

We assume that there is an efficient algorithm \ApproxKnapsack that for any bin $i$, given non-negative weights $w_j$ on items for each $j\in [n]$, returns in time $p(n,\alpha)$ an $\alpha$-approximation to the maximum weight set $T\in \F_i$ where the weight of a set $T$ is defined as: $w(T)=\sum_{j\in T} w_j$.
For \gap the above problem is exactly the knapsack problem with $n$ items and there exists an efficient $(1-\varepsilon)$-approximation that runs in time $\tilde{O}(n+\nicefrac{1}{\varepsilon^2})$~\cite{chen2024nearly}.

We reduce \sap to \sm with a partition matroid and an exponentially sized universe (refer to, e.g.,~\cite{CCPV11}, for the details of this approximation preserving reduction).
Despite the exponential size of the universe we are able to use \gp to prove the following theorem.

\begin{theorem}\label{thm:sap_main}
    Given an $\alpha$-approximation algorithm \ApproxKnapsack for finding the maximum weight set in each $\cF_i$ that runs in time $p(n,\alpha)$, for every $ \varepsilon >0$ there exists a $(1-\varepsilon)(1-e^{-\alpha})$-approximation algorithm for the \saplong with $n$ items and $m$ bins that runs in time $\tilde{O}\left(\left(\nicefrac{n}{\varepsilon} +p(n,\alpha)\right)\cdot m\right)$. 
\end{theorem}

First observe that Theorem~\ref{thm:gap_main} follows  by using the $\tilde{O}(n+\nicefrac{1}{\varepsilon^2})$-running-time $(1-\varepsilon)$-approximation algorithm of~\cite{chen2024nearly} for the knapsack problem as \ApproxKnapsack in Theorem~\ref{thm:sap_main}.
Second, let us compare the results of Theorems~\ref{thm:gap_main} and~\ref{thm:sap_main} with the previously best known algorithms.
For \gap, approximations of $(1-\nicefrac{1}{e})$ and $(1-\nicefrac{1}{e}+\delta)$ (where $ \delta \geq 10^{-120}$) were given by~\cite{fleischer2006tight} and~\cite{feige2006approximation}, respectively.
For \sap, an approximation of $ (1-\nicefrac{1}{e})\alpha$ was given by~\cite{fleischer2006tight}.
All the above mentioned algorithms are based on solving an exponential sized linear program via the ellipsoid algorithm, and thus are subsequently slow.
For \sap, an improved approximation of $(1-\varepsilon)(1-e^{-\alpha})$, for every $ \varepsilon >0$, due to~\cite{CCPV11}, reduced the problem to \sm with a partition matroid and an exponentially sized universe.
The running time of this algorithm is a large polynomial in $m$, $n$, $\nicefrac{1}{\varepsilon}$ along with a dependence on $p(n,\alpha)$.
\gap admits, for every $\varepsilon >0$, a fast $(\nicefrac{1}{(2+\varepsilon)})$-approximation in a running time of $ \tilde{O}((n+\nicefrac{1}{\varepsilon^2})\cdot m)$~\cite{CKR2006}.

There are two things we note.
First, in order to prove Theorem~\ref{thm:sap_main} we present a suitable swap procedure that determines which elements are exchanged in \gp once an event in the Poisson process occurs.
Second, the running times of Theorems~\ref{thm:gap_main} and~\ref{thm:sap_main} do not assume any oracle access but count the total number of steps taken by the algorithm, since the input is given explicitly.

\paragraph{Generalized Partition Matroid.}
When $\cM=(U,\cI)$ is a generalized partition matroid there exist a partition $U_1,\ldots, U_k$ of $U$ and bounds $\ell_1,\ldots,\ell_k\in \mathbb{N} $.
A set $S\subseteq U$ is in  $\cI$ if and only if $|S\cap U_j|\leq \ell_j$ for each $1\leq j\leq k$.
Thus, a solution $S\subseteq U$ to \sm when $\cM$ is a generalized partition matroid is feasible if and only if $ S$ contains at most $\ell_j$ elements from each part $U_j$.

The following two theorems generalize our results for partition matroids to generalized partition matroids with a slight overhead to  the running. Similarly, we provide separate theorems for the problem with an oracle to $f$ and $F$. 

\begin{theorem}\label{thrm:GeneralizedPartitionF}
There exists an algorithm that given a generalized partition matroid $\cM=(U,\cI)$, a monotone submodular function $f\colon 2^U\rightarrow \mathbb{R}_+$, and $ 0< \varepsilon \leq 1$, performs  $ O(n\ln^2{(\nicefrac{1}{\varepsilon})})$ evaluations of $F$, in expectation, and finds $A\in \cI$ satisfying: $ \mathbb{E}[f(A)]\geq (1-\varepsilon)(1-\nicefrac{1}{e}) f(O)$.
Here $n$ is the size of the universe $U$ and $ O=\argmax \{ f(S)\colon S\in \cI\}$ is an optimal base.
\end{theorem}
\begin{theorem}\label{thrm:GeneralizedPartition}
There exists an algorithm that given a generalized partition matroid $ \cM=(U,\cI)$, a monotone submodular function $ f\colon 2^U\rightarrow \mathbb{R}_+$, and $ 0<\varepsilon<\nicefrac{1}{4}$, performs {$O(n \log n \log^5(1/\varepsilon)/\varepsilon^2)$} evaluations of $f$, in expectation, and finds $ A\in \cI$ satisfying $ \mathbb{E}[f(A)]\geq (1-\varepsilon)(1-\nicefrac{1}{e})f(O)$.
Here $n$ is the size of the universe $U$ and $ O=\argmax \{ f(S):S\in \cI\}$ is an optimal base.
\end{theorem}

\subsection{Our Techniques}\label{sec:techniques}
\paragraph{\gp and the Continuous Greedy Algorithm.}
In first glance \gp seems closely related to the local search approach.
However, its origin lies with the continuous greedy algorithm~\cite{CCPV11}.

The continuous time process that defines the continuous greedy algorithm corresponds to a continuous trajectory $ \{ \bx(t)\} _{0\leq t\leq 1}$ where:
$(1)$ $\bx(0)\leftarrow \mathbf{0}$ is initialized to $\mathbf{0}$; and
$(2)$ $\frac{\partial \bx(t)}{\partial t}$ is set to $ \indicator _{Z(t)}$ where $Z(t)$ is a base that maximizes the linear objective induced by $\nabla F(\bx(t))$,  i.e., $$Z(t)\triangleq \argmax \left\{ \sum _{i\in Z}\nabla _i F(\bx(t))\colon Z\in \cB\right\} .$$
Clearly, $ \bx(t)\in t\cP_{\cB}$ where $ \cP_{\cB}$ is the base polytope of the matroid $ \cM$, i.e., $\bx(t)$ is $t$ times a convex combination of bases.
Thus, $\bx(1)$ is a fractional base that can be rounded without any loss by, e.g., pipage rounding~\cite{AS04,CCPV11}.
At the heart of the analysis of this continuous process is the following claim:
$$ \frac{\partial F(\bx(t))}{\partial t}\geq f(O)-F(\bx(t)),$$
where $O=\argmax \{ f(S)\colon S\in \cI\}$ is an optimal base.
This yields that $ F(\bx(t))\geq (1-e^{-t})f(O)$, and in particular $ F(\bx(1))\geq (1-\nicefrac{1}{e})f(O)$.

When discretizing the above continuous process to obtain an algorithm, the trajectory discretizes to $ \bx(0)={\mathbf{0}}$ and $ \bx(t+\delta) = \bx(\delta)+\delta \indicator_{Z(t)}$, where $Z(t)$ is as before and $\delta$ is the discretized step size.
Similarly to the above, it can be proved that $ F(\bx(t+\delta))-F(\bx(t))\gtrsim\delta (f(O)-F(\bx(t)))$, resulting in an approximation of $ 1-\nicefrac{1}{e}-O(\delta\cdot \text{poly}(n))$.

Intuitively, \gp is the limit of the continuous greedy algorithm that rounds $\bx(t)/t$ on the fly while the step size $\delta$ approaches $0$.
It is crucial to note that the rounding is of $\bx(t)$ scaled by $t$, and not of $ \bx(t)$.
{One can analyze the evolution of $ \bx(t)/t$, which can be proved to behave like a Frank-Wolfe style continuous local search algorithm (see Section~\ref{app:FW}).
Thus, an alternative way to understand the origin of \gp is that it is a Frank-Wolfe style algorithm with rounding on the fly and taking the limit as the step size approaches $0$.}
The focus now shifts to explaining what it means to round $ \bx(t)/t$ on the fly, and taking the limit of the algorithm as the step size $\delta$ approaches $0$.

First, consider rounding on the fly $ \bx(t)/t$.
Rather than maintaining $\bx(t)$, \gp maintains a random base $ A(t)$ such that $ \mathbb{E}[t\indicator _{A(t)}]=\bx(t)$.
Since we do not have $ \bx(t)$, but only the base $A(t)$, when computing $ Z(t)$, $ \nabla F(t\indicator _{A(t)})$ is used instead of $ \nabla F(\bx(t))$.
Hence, $ Z(t)\triangleq \argmax \{ \sum _{i\in Z}\nabla _i F(t\indicator _{A(t)})\colon Z\in \cB\}$.
As in continuous greedy we have that the trajectory at time $(t+\delta)$ is: $t\indicator _{A(t)}+\delta \indicator _{Z(t)}$. Additionally, from the guarantee as in the continuous greedy algorithm, we have:
\begin{align*}
&F\left(t\indicator _{A(t)}+\delta \indicator _{Z(t)}\right)-F\left(t\indicator _{A(t)}\right)  \gtrsim  \delta \left(f(O)-F\left(t\indicator_{A(t)}\right)\right). 
\end{align*}
Observe that $t\indicator _{A(t)}+\delta \indicator _{Z(t)}\in (t+\delta )\cP_{\cM}$ as desired but it is not a single base scaled by $ (t+\delta)$ as our approach demands.
Thus, we round the fractional base $(t\indicator _{A(t)}+\delta \indicator _{Z(t)})/(t+\delta)$ using (randomized) pipage rounding and obtain a base $A(t+\delta)$ satisfying:
\begin{align}
   \mathbb{E}\left[F\left(\left(t+\delta\right) \indicator_{A(t+\delta)}\right)\right]\geq F\left(t\indicator _{A(t)}+\delta \indicator _{Z(t)}\right), \label{our-approach1} 
\end{align}
where the expectation is over the randomness of the pipage rounding procedure.

Therefore, the discretized trajectory at time $ (t+\delta)$ is set to: $ (t+\delta)\indicator _{A(t+\delta)}$.
Since~\eqref{our-approach1} holds, rounding on the fly of the discretized trajectory scaled by the time gives:
\begin{align*}
&\mathbb{E}\left[F\left(\left(t+\delta\right)\indicator_{A(t+\delta)}\right)\right]-F\left(t\indicator _{A(t)}\right)  \gtrsim \delta \left(f(O)-F\left(t\indicator_{A(t)}\right)\right). \label{our-approach2}
\end{align*}

Second, consider taking the limit as the step size $\delta$ approaches $0$.
Rounding on the fly as described above provides an approximation of $ (1-\nicefrac{1}{e}-O(\delta\cdot \text{poly}(n)))$, hence it also suffers from a loss introduced by discretization of the continuous process.
Previous works (\cite{CCPV11}) as well as~\cite{BV14} and subsequent follow up works) aim to take the step size $\delta$ as large as possible to reduce the number of iterations, while bounding the loss $\delta$ incurs in the approximation.
We adopt the opposite approach and consider the limit as $\delta$ approaches $0$.

In the limit as $\delta$ approaches $0$, to the first order, applying the random variant of pipage rounding on $ (t\indicator _{A(t)}+\delta \indicator_{Z(t)})/(t+\delta)$ boils down to one of the following two:
$(1)$ setting $A(t+\delta)$ to be $A(t)$ w.p. $~\approx1-\delta k/t$; and
$(2)$ swapping a single uniform random element $i\in A(t)$ with $h(i)$ w.p. $\approx \delta k/t$.
The crucial, yet trivial, observation is that if the former happens, i.e., $A(t+\delta)=A(t)$, then from an algorithmic perspective nothing needs to be done!
Since the random variant of pipage rounding is oblivious to the objective, all that is needed is to understand the time of the next swap in this continuous process as $ \delta$ approaches $0$.
Calculating these limits, one can show that this is exactly given by a Poisson process with rate $ \lambda(t)=\nicefrac{k}{t}$ (as in \gp) that we analyze directly.

\paragraph{The Swap Procedure.}
Following the above, we observe that we do not need to compute $Z(t)$ when executing \gp.
It is enough, once an event in the Poisson process occurs, to find a suitable swap of a single element in $A(t)$.
This simple yet important observation is the driving force behind our applications.
Refer to Section~\ref{sec:poisson} for more details. 

{
We establish an interesting connection between the swap procedure and  the classic \mab problem~\cite{EMM06}: given a collection of random variables accessible by samples alone, find the one whose expectation is the highest.
The error introduced by the algorithm of~\cite{EMM06} depends on the support of the distribution, and in order to bound it we use a preprocessing algorithm that ensures that the effective support is sufficiently small.
}

\subsection{Related Work}\label{sec:additional}

Algorithms for submodular maximization have been studied as early as the late $70$'s, with the seminal work of~\cite{nemhauser1978analysis} who proved that greedy achieves an approximation of $(1-\nicefrac{1}{e})$ for monotone submodular maximization subject to a cardinality constraint, i.e., a uniform matroid independence constraint.
A matching hardness of $ (1-\nicefrac{1}{e})$ in the oracle model was given by~\cite{NW78}, i.e., any algorithm that achieves an approximation better than $ (1-\nicefrac{1}{e})$ is required to perform exponentially many queries to $f$.
Later, Feige~\cite{F98} proved that assuming $ \text{P}\neq \text{NP}$ the bound of $ (1-\nicefrac{1}{e})$ is also tight for the special case of a coverage function.
When considering a general matroid independence constraint, i.e., the \sm problem, it was shown by~\cite{fisher1978analysis} that both the discrete local search and the greedy algorithms achieve an approximation of $\nicefrac{1}{2}$.
This problem was subsequently settled with the work of~\cite{CCPV11} who presented a tight approximation of $(1-\nicefrac{1}{e})$ by introducing the continuous greedy algorithm and the use of the multilinear extension for submodular maximization.
Later algorithms that are based on non-oblivious local search~\cite{filmus2014monotone,BF24deter} are able to match the above tight $(1-\nicefrac{1}{e})$ approximation for a general matroid.
The continuous greedy approach has been successfully extended to accommodate various constraints, including multiple knapsacks, multiple matroids, and even more complex independence systems, e.g., ~\cite{chekuri2011submodular,kulik2013approximations,BV14,FKS21}.

When considering general submodular functions, which are not necessarily monotone, an approximation of $\nicefrac{1}{2}$ was given by~\cite{BFNS15} for the unconstrained case by introducing the double greedy algorithm.
A matching $(\nicefrac{1}{2})$-hardness for the unconstrained case was given by~\cite{feige2011maximizing}.
When considering a general matroid independence constraint, a long sequence of works~\cite{LMNS10,V13,gharan2011submodular,FNS11,EN16,BF19,BF24DR} presents ever improving guarantees, culminating in the current best-known approximation of $\approx 0.401 $ \cite{BF24DR}.
All these works are based on continuous greedy and local search approaches.
A hardness of $ \approx 0.478$ for a partition matroid was given by \cite{gharan2011submodular}, which was subsequently extended to a cardinality constraint \cite{Q23}.
Additional constraints such as a single knapsack, multiple knapsack constraints, and multiple matroid constraints, were also considered when maximizing non-monotone submodular functions~\cite{kulik2013approximations,LMNS10,LSV10}.

It is worth mentioning that some very fast algorithms for submodular maximization with a provable tight approximation guarantee are implemented in practice~\cite{CTI15,FD15,JXLWZ20,KLHK17,MKY16,PBK24,RAGT14,SSPB14}.
Two notable examples include stochastic greedy (a.k.a. sample greedy) for \sm with a uniform matroid~\cite{MBKV15,BFS17} and double greedy for unconstrained non-monotone submodular maximization~\cite{BFNS15}.

Starting with the work of~\cite{BV14} on the \sm problem, a significant body of work studies fast algorithms with a provable tight approximation guarantee~\cite{BFS17,KT24,EN19, EN19knapsack,henzinger2023faster} (where many of these results build on~\cite{BV14}).
In Table \ref{tab:example} we summarize the fastest known deterministic and randomized algorithms that are known to achieve a tight approximation of $ (1-\varepsilon)(1-\nicefrac{1}{e})$ for the \sm problem, depending on the type of matroid: uniform, partition, and general. 
{One should note that for a general matroid and in particular a partition matroid, each of the randomized algorithms does not dominate the other for all $n$, $k$, and $\varepsilon$, e.g., $O(n\log^2(n/\varepsilon)/\varepsilon^5) $\cite{EN19}, $ O((k^{3/2}+\sqrt{k} n \log^2{(n/\varepsilon)})/\varepsilon^{2.5})$~\cite{KT24}, and $ O(n\log{k}\cdot 2^{O(1/\varepsilon^4)})$~\cite{BF24ext}.
It is worth mentioning that our use and analysis of the residual random greedy algorithm of~\cite{BFNS14} for a faster preprocessing is inspired by~\cite{BFS17,EN19}, who also use the same algorithm.}

The \sw problem with $k$ players, as well other several closely related combinatorial auctions problems, exhibit a rich history \cite{CG10,F06,FV10,KLMM08,LLN01}.
An asymptotically tight approximation of $ (1-\nicefrac{1}{e})$ was given by \cite{CCPV11}, improving the previous known approximation of $ k/(2k-1)$ \cite{DS06}.
A hardness of $ (1-(1-\nicefrac{1}{k})^k)$  is known \cite{V13}, and a matching tight $ (1-(1-\nicefrac{1}{k})^k)$-approximation, for any $k$, was given by \cite{FNS11} (thus improving the tight asymptotic approximation of~\cite{CCPV11}).
\subsection{Organization}
Section~\ref{sec:prelim} contains some necessary preliminaries.
Section~\ref{sec:poisson} presents the complete analysis of \gplong, where Section~\ref{sec:Applications} focuses on the applications.
Finally, Section~\ref{app:FW} discusses the relation between \gplong, the continuous greedy algorithm and the Frank-Wolfe algorithm.

\begin{table*}[t]
    \centering
        \caption{Number of values queries used by the state of art algorithms}
    \label{tab:example}
    \begin{tabular}{|c|c|c|}
        \hline
         &  Deterministic & Random \\
        \hline
        Uniform & $O\left(n/\eps \cdot \log\left( n/\eps\right) \right)$ \cite{BV14} & $O(n\log( 1/\varepsilon))$ \cite{BFS17,MBKV15} \\
        \hline
        Partition & -  & $O\left(n \log^2\left(n/\eps \right) /\varepsilon^5\right) $ \cite{EN19}\\
        \hline
        \multirow{2}{*}{General} & $O(n^2 \cdot 2^{O(1/\varepsilon^4)})$ \cite{BF24ext} & $O((
        k^{3/2} + \sqrt{k} n   \log^2(n/\eps) )/\eps^{2.5})$ \cite{KT24} \\ 
        & $\tilde{O}(nk \cdot 2^{O(1/\varepsilon)})$ \cite{BF24deter} & $O(n\log{k} \cdot 2^{O(1/\varepsilon^4)})$ \cite{BF24ext}
        \\
        \hline
    \end{tabular}

\end{table*}

\section{Preliminaries}\label{sec:prelim}

A matroid $ \cM=(U,\cI)$ consists of a ground set $U$, a collection $ \cI\subseteq 2^U$ of independent sets satisfying: $(1)$ $\emptyset \in \cI$; $(2)$ if $S\in \cI$ and $T\subseteq S$, then $T\in \cI$; and $(3)$ if $S,T\in \cI$ and $|S|>|T|$, then there exists $i\in S\setminus T$ such that $ T+i\in \cI$. 
For a matroid $ \cM=(U,\cI)$ of rank $k$ we denote by $ \cB$ the collection of its bases: $ \cB\triangleq \{ S\in \cI:r(S)=k\}$.
Here $r\colon 2^U\rightarrow \mathbb{N}$ is the rank function of the matroid $\cM$: $ r(S)\triangleq \max \{ |T|:T\subseteq S,T\in \cI\}$, and $k\triangleq r(U)$ is the rank of $\cM$.
A classic result on matroids \cite{brualdi1969comments} states the existence of base exchange map for any two bases as defined below.  
\begin{definition}[Base Exchange Map]\label{def:h}
Given a matroid $ \cM=(U,\cI)$ and $ B_1,B_2\in \cB$, a bijection $ h\colon B_1\rightarrow B_2$ is a \bem if it satisfies: $(1)$ $ B_1-i+h(i) \in \cB$, $ \forall i\in B_1\setminus B_2$; and $(2)$ $ h(i)=i$, $\forall i\in B_1\cap B_2$.   
\end{definition}

Let $ \{ N(t)\}_{t\geq \varepsilon}$ be a non-homogeneous Poisson process with rate function $ \lambda(t)$.
For convenience of presentation, we assume that the process starts at time $\varepsilon >0$ (for some given $\varepsilon$) and not $0$.
Thus, $ N(\varepsilon)=0$ and $N(t)$ denotes the number of events that occurred in the time interval $ [\varepsilon,t]$.
More generally, for every time interval $I$ we denote by $ N_I$ the number of events that occurred in interval $I$, e.g., if $ I=(t,t+\delta]$ then $ N_{(t,t+\delta]}=N(t+\delta)-N(t)$.

Our analysis requires the following properties of a Poisson process (see, e.g., \cite{ross2014introduction}, Chapter 5):
\begin{enumerate}
\item \label{Poisson-independ} If $ I$ and $J$ are disjoint intervals, then $ N_I$ and $ N_J$ are independent.
\item \label{Poisson-events} For every $ t\geq \varepsilon$: $\lim _{\delta \rightarrow 0^+}\frac{1}{\delta}\Pr[N_{(t,t+\delta]}=1]=\lambda(t)$ and $\lim _{\delta \rightarrow 0^+}\frac{1}{\delta}\Pr[N_{(t,t+\delta]}\geq 2]=0$.
\end{enumerate}

We require the following three known properties of the multilinear extension (see, e.g., \cite{CCPV11}):
\begin{enumerate}
\item If $f$ is monotone then $F$ is also monotone:  $F(\mathbf{x})\leq F(\mathbf{y})$ for every $\mathbf{x}\leq \mathbf{y}$ and $\mathbf{x},\mathbf{y}\in [0,1]^U$. \label{multilinear-monotone}
\item If $f$ is submodular, then for every $\mathbf{x},\mathbf{y}\in [0,1]^U$, where $ \mathbf{x}\leq \mathbf{y}$, and every $ i\in U$:
\[ F(\mathbf{x}\vee \indicator _{\{ i\} })-F(\mathbf{x}) \geq F(\mathbf{y} \vee \indicator _{\{ i\} })-F(\mathbf{y}) .\]
A straightforward corollary of the above is that for every $\mathbf{x}\in [0,1]^U$ and $S\subseteq U$:
 \[ \sum _{i\in S} \left( F(\mathbf{x}\vee \indicator _{\{ i\} })-F(\mathbf{x})\right) \geq F(\mathbf{x}\vee \indicator _S) -F(\mathbf{x}).\] \label{multilinear-submodular}
\item For every $ S\subseteq U$, $ 0\leq t\leq 1$, and $i\in U$:
\begin{align*} \nabla _iF(t\indicator _S) & = F(t\indicator _{S\setminus\{ i\} } + \indicator _{\{ i\} }) - F(t\indicator _{S\setminus\{ i\} }) \\ &=\frac{1}{t} \cdot \left( F(t\indicator _{S\setminus \{ i\}} + t\indicator _{\{ i\} })-F(t\indicator _{S\setminus\{ i\} })\right).\end{align*}
Additionally for any $\mathbf{0}\leq \mathbf{x}\leq \mathbf{y}\leq \mathbf{1}$ and $i\in U$, we have: $ \nabla_i F(\mathbf{x})\geq \nabla_i F(\mathbf{y})$.
\label{multilinear-gradient}
\end{enumerate}

\section{The \gplong}\label{sec:poisson}
Recall that in \Cref{sec:results} a description of \gplong (\gp) was given.
First, the internal state is initialized to be an arbitrary base $A\in \cB$.
Second, a non-homogeneous Poisson process with rate $ \nicefrac{k}{t}$ is executed starting at time $\varepsilon$ and terminating at time $1$.
Third, if the next event of the Poisson process occurs at time $t$ then:
$(1)$ a base $Z(t)\in \cB$ maximizing $ \langle \indicator_{Z(t)},\nabla F(t\indicator_{A})\rangle$ is found;
$(2)$ a bijection $ h\colon A\rightarrow Z(t)$ that satisfies $ A-i+h(i)\in \cB$, for every $i\in A $, is computed; and
$(3)$ a uniform random $i\in A$ is chosen and $ A\leftarrow A-i+h(i)$ ($i$ is swapped with $h(i)$). 

Although the above follows a strict recipe, it will be useful to present a more generic version of \gp that uses a swap procedure.
The swap procedure is a (possibly) randomized algorithm that is given a base $A\in \cB$ and a time $ \varepsilon\leq t\leq 1$, and returns a pair $ (i,j)$ of elements $i\in A$ and $j\in U$ such that $A-i+j\in \cB$. 

\gp is now formally defined as follows.
First, the rate function is set to be $ \lambda(t)\triangleq\nicefrac{k}{t}$.
Second, for every $ t\geq \varepsilon$, $ \tau(t)$ denotes the random variable that equals the time of the next event after time $t$, i.e., $ \tau(t)\triangleq \min \{ r:N(r)>N(t)\}$.
Moreover, $ \tau _i$ denotes the random variable that equals the time of the $i$\textsuperscript{th} event, i.e., $ \tau _0 = \varepsilon$ and for every $i\geq 1$ we have $ \tau_i = \tau (\tau _{i-1})$.
We note that for every $ t\geq \varepsilon$ the density of $\tau(t)$ is explicitly known, and thus one can computationally sample $\tau(t)$ and simulate the Poisson process. 

Algorithm \ref{alg:Poisson} summarizes the formal description of \gp.
It receives as an input a matroid $ \cM=(U,\cI)$, a starting time $\eps\in (0,1)$, and a swap procedure $\swap$. 
Algorithm \ref{alg:Poisson} is fairly simple as it simulates a Poisson process with rate $\lambda(t)$, and at each event of the Poisson process it performs a swap between two items according to the swap procedure.

\begin{algorithm}[t]
\caption{\gplong $(\cM,\varepsilon,\swap)$}
\SetKwInOut{Input}{input}
\SetKwInOut{Output}{output}

\SetAlgoNlRelativeSize{0}
\label{alg:Poisson}
 	
 	$ t\leftarrow \varepsilon$ and set $A$ to be an arbitrary base of $\cM$.

    sample $\tau(t)$ the time of the next event after $t$ and $ t\leftarrow \tau(t)$.

    \While{t<1}{

    $(i, j)\leftarrow \swap(t,A)$. \label{poisson:item_sample}

    $A\leftarrow A-i+j$. \label{poisson:item_swap}

    sample $\tau(t)$ the time of the next event after $t$ and $ t\leftarrow \tau(t)$.
    
    }
	
 	\Return $A$.
 	
\end{algorithm}

In order to analyze Algorithm~\ref{alg:Poisson}, additional properties of the swap procedure that relate to the value of the inner state $A$ are required.
A swap procedure is {\em above-average} if for every $ A\in \cB$ and $\varepsilon \leq t\leq 1$:
\ignore{
\begin{align}
   \mathbb{E}_{(i,j)\sim \swap(t,A)} [\nabla _j F(t\indicator _A) - \nabla _i F(t\indicator _A) ] \geq \frac{1}{k}\left( \sum _{o\in O}\nabla _o F(t\indicator _A) - \sum _{a\in A}\nabla _a F(t\indicator _A) \right),\label{eq:average}
\end{align}}

\begin{equation}
\begin{aligned}
& \mathbb{E}_{(i,j)\sim \swap(t,A)}
\big[
\nabla_j F(t\indicator_A)
-
\nabla_i F(t\indicator_A)
\big]
\ge \frac{1}{k}
\Big(
\sum_{o\in O} \nabla_o F(t\indicator_A) 
-\sum_{a\in A} \nabla_a F(t\indicator_A)
\Big) ,
\end{aligned}
\label{eq:average}
\end{equation}
where $ O=\argmax \{f(S)\colon S\in \cB\}$ is an optimal solution and $(i,j)\sim \swap(t,A)$ denotes that the random pair $(i,j)$ is the random output of $\swap(t,A)$ when executed with base $A$ and time $t$.
{It should be noted that the expectation in~\eqref{eq:average} is taken over the randomness of the swap procedure and is independent of the Poisson process.}
For convenience of presentation, let $ p_{(i,j)}(t,A)$ denote the probability that the swap procedure returns the pair $ (i,j)$ for a given base $A$ and time $t$.
Hence, the following holds:
\begin{align}
&\mathbb{E}_{(i,j)\sim \swap(t,A)} [\nabla _j F(t\indicator _A) - \nabla _i F(t\indicator _A) ]=  \nonumber\sum _{i\in A,j\in U}p_{(i,j)}(t,A)\cdot \left( \nabla_j F(t\indicator _A) - \nabla _i F(t\indicator _A)\right). \nonumber
\end{align}

The above-average property \eqref{eq:average} has a simple intuitive interpretation.
Recalling the strict recipe mentioned above, if one chooses $Z(t)$ to be an optimal solution $O$, then the right hand side of~\eqref{eq:average} is the expected gain if one swaps a uniform random element $i\in A$ with $h(i)\in O$.
The left hand side of~\eqref{eq:average} equals the expected gain of swapping the pair $ (i,j)$ returned by the swap procedure.
Thus, \eqref{eq:average} implies that in expectation the swap procedure is as good as a uniform random swap with an optimal solution.

Before presenting the theorems summarizing the analysis of \gp, we mention that since the Poisson process is time-continuous it is required that the swap procedure is mathematically well-behaved.
Formally, a swap procedure is {\em right continuous} if $p_{(i,j)}(t,A)$, as a function of $t$ defined on $ [\varepsilon,1]$, is a right continuous function in $t$, for every $A\in \cB$ and every $i\in A, j\in U$.
For simplicity of presentation, the term swap procedure will be used henceforth to refer only to right continuous swap procedures.

The following theorem states that given a right continuous swap procedure that is above-average, Algorithm~\ref{alg:Poisson} returns a solution achieving a tight approximation.
\begin{theorem}
\label{thm:average_poisson}
Given a matroid $ \cM=(U,\cI)$ of rank $k$, a monotone submodular function $f\colon 2^U \rightarrow \mathbb{R}_+$, a starting time $ 0<\varepsilon \leq 1$, and a right continuous swap procedure that is above-average, Algorithm~\ref{alg:Poisson} finds $A\in \cB$ satisfying $ \mathbb{E}[f(A)]\geq (1-\varepsilon)(1-\nicefrac{1}{e})\cdot f(O)$ using in expectation $ k\ln{(1/\varepsilon)}$ swap procedure calls.
Here $O=\argmax\{f(S)\colon S\in \cI\}$ is an optimal base. 
\end{theorem}

While above-average swap procedures exist, in some of the applications considered in this work such procedures are either impossible to achieve or incur significant overhead in running time.
In order to overcome these obstacles, a relaxed notion of above-average is considered: given $\beta \in (0,1]$ and $\eta\geq 0$ a swap procedure is called $(\beta,\eta)$-approximate if for every $A\in \cB$ and $ \varepsilon\leq t\leq 1$:
\begin{align}
\begin{aligned}&\E _{(i,j)\sim \swap(t,A)}[ \nabla_j F(t\indicator_A) -\nabla_i F(t \indicator_A)] \geq \frac{\beta}{k} \sum_{o \in O} \nabla_o F(t \indicator_A)  - \frac{1}{k}\sum_{a\in A} \nabla_a F(t \indicator_A)  - \frac{\eta}{k} f(O).  
\end{aligned}\label{eq:approximate_swap}
\end{align}
Here, as in \eqref{eq:average}, $ O=\argmax \{ f(S)\colon S\in \cB\}$ is an optimal solution and $(i,j)\sim \swap(t,A)$ denotes that the random pair $(i,j)$ is the random output of $\swap(t,A)$ when executed with base $A$ and time $t$. 
Note that if $\beta=1$ and $\eta=0$ then $(\beta,\eta)$-approximate~\eqref{eq:approximate_swap} coincides with above-average~\eqref{eq:average}.

The following theorem captures the loss in the approximation given a right continuous swap procedure that is not above-average but is only $(\beta,\eta)$-approximate.
\begin{theorem}
\label{thm:approximate_poission}
Given a matroid $ \cM=(U,\cI)$ of rank $k$, a monotone submodular function $ f\colon 2^U\rightarrow \mathbb{R}_+$, a starting time $0<\varepsilon \leq 1$, and a right continuous swap procedure that is $ (\beta,\eta)$-approximate, Algorithm~\ref{alg:Poisson} finds $A\in \cB$ satisfying $\mathbb{E}[f(A)]\geq (1-\varepsilon)(1-\nicefrac{\eta}{\beta})(1-e^{-\beta})\cdot f(O)$ using in expectation $k\ln{(1/\varepsilon)}$ swap procedure calls.
Here $ O=\argmax \{ f(S)\colon S\in \cI\}$ is an optimal base.
\end{theorem}

\paragraph{\textbf{Analysis}.}
In order to analyze Algorithm \ref{alg:Poisson} and prove Theorems~\ref{thm:average_poisson} and \ref{thm:approximate_poission}, we denote by $\{A(t)\}_{\varepsilon \leq t\leq 1}$ the random bases that indicate the inner state $A$ of the algorithm at time $t$.
Formally, the algorithm maintains only a single base as an inner state that changes at times that correspond to events of the Poisson process. 
However, this naturally extends to all times $ \varepsilon \leq t \leq 1$ as follows:
$(1)$ we set $ A(\varepsilon)$ to be the arbitrary initial base chosen by Algorithm \ref{alg:Poisson}; and
$(2)$ for every $ \varepsilon \leq t \leq1$, where $\tau_{i-1}\leq t<\tau_i $ for some $i$, we set $A(t)\triangleq A(\tau_{i-1})$ and set $ A(\tau _i)$ according to the result of the random swap operation Algorithm \ref{alg:Poisson} performs at time $ \tau _i$ (step~\eqref{poisson:item_swap} in Algorithm \ref{alg:Poisson}).

Following the discussion in Section~\ref{sec:techniques}, that relates \gp to the continuous greedy algorithm with randomized rounding on the fly, our analysis tracks the expected value of $A(t)$ scaled by the time $t$.
Formally, define $ Q\colon [\varepsilon,1]\rightarrow \mathbb{R}_+$ to be the function that given $ \varepsilon \leq r\leq 1$ equals the expected value of the output of \gp if it terminates at time $r$:
\begin{align}
Q(r)\triangleq \mathbb{E}[F(r\indicator _{A(r)})] .\label{def:Q}
\end{align}
In order to provide a lower bound on the rate in which $Q(r)$ gains value, it will be convenient to condition on the inner state of \gp at time $t$, where $\varepsilon \leq t\leq r$.
To this end, define for every base $S\in \cB$ and time $\varepsilon \leq t\leq1$ the function $ V_{S,t}\colon [t,1]\rightarrow \mathbb{R}_+$ as follows:
\begin{align}
   V_{S,t}(r)\triangleq \mathbb{E}\left[ F(r\indicator _{A(r)})|A(t)=S\right]. \label{value}
\end{align}
Clearly, for a given fixed $r$ the definition of conditioned expectation implies that for every $ \varepsilon \leq t \leq r$:
\begin{align*}
Q(r) &= \sum _{S\in \cB}\Pr[A(t)=S] \cdot \mathbb{E}[f(r\indicator _{A(r)})|A(t)=S]\\
&= \sum _{S\in \cB}\Pr[A(t)=S]\cdot V_{S,t}(r),
\end{align*}
thus relating \eqref{def:Q} and~\eqref{value}.

We emphasize that from this point on, all derivatives are right derivatives, where the right derivative of a function $g(r)$ with respect to $r$ is denoted by $ \frac{\partial _+ g(r)}{\partial r}$.

The two main lemmas, one for each type of swap procedure, that are used to prove Theorems~\ref{thm:average_poisson} and \ref{thm:approximate_poission} by lower bounding the rate at which $V_{S,t}(r)$ and $Q(t)$ increase, are now presented.

\begin{lemma}\label{lem:value-average-swap}
For every $ \varepsilon \leq t <1$, $S\in \cB$, and a right continuous above-average swap procedure: $ \frac{\partial _+V_{S,t}(r)}{\partial r}\big|_{r=t} \geq f(O)-F(t\indicator_S)$ and subsequently, $ \frac{\partial _+Q(r)}{\partial r}\big|_{r=t} \geq f(O)-Q(t)$.
Here, $ O=\argmax \{ f(S):S\in \cI\}$ is an optimal base.
\end{lemma}
\begin{lemma}\label{lem:value-approximate-swap}
For every $ \varepsilon \leq t <1$, $S\in \cB$, and a right continuous $(\beta,\eta)$-approximate swap procedure:
$\frac{\partial _+V_{S,t}(r)}{\partial r}\big| _{r=t}\geq (\beta-\eta)f(O)-\beta F(t\indicator _S)$ and subsequently $ \frac{\partial _+Q(r)}{\partial r}\big|_{r=t} \geq (\beta-\eta)f(O)-\beta Q(t)$.
Here, $ O=\argmax \{ f(S):S\in \cI\}$ is an optimal base.
\end{lemma}
It should be noted that the lower bound in the above two lemmas is common in constrained submodular maximization; see, e.g., \cite{CCPV11}.
However, a crucial difference is that we show the rate of increase in objective $V_{S,t}(r)$ for \emph{every} intermediate state given by time $t$ and set $S$.

We start with providing the derivatives of the conditioned transition probabilities of the inner state of \gp, allowing us to prove Lemmas~\ref{lem:value-average-swap} and \ref{lem:value-approximate-swap} (and subsequently Theorems~\ref{thm:average_poisson} and \ref{thm:approximate_poission}).

\paragraph{Derivatives of Conditioned Transition Probabilities.}
The following lemma is the only part of the analysis in which the properties of the Poisson process are used.

\begin{lemma}\label{lem:prob-derivative}
For every $ \varepsilon \leq t<1$, $ S,T\in \cB$, and a right continuous swap procedure:
\begin{align*}
\begin{aligned}
&\frac{\partial_+ \Pr [A(r)=T|A(t)=S]}{\partial r}\bigg|_{r=t} = \\
&\begin{cases}
   \lambda(t) \cdot p_{(i,j)}(t,S)  & \text{ if  } \exists i\in S,j\in U\setminus S \text{ and } T=S-i+j\\
    -\lambda(t)(1-\sum\limits_{i\in S}p_{(i,i)}(t,S)) & T=S \\
    0 & \text{otherwise}
\end{cases}
\end{aligned} 
\end{align*}
\end{lemma}
\begin{proof}
In order to prove the lemma, we utilize the following claim that precisely provides the derivatives of the conditioned transition probabilities. 
For simplicity of presentation, we denote the random variable $ N_{(t,t+\delta]}$ by $N$, and denote by $1_{S=T}$ the indicator function for the case that $S=T$, i.e., $ 1_{S=T}=1$ if $S=T$ and $0$ otherwise.
\begin{claim}\label{claim:poisson}
For every $ \varepsilon \leq t<1$ and $ S,T\in \cB$:
\begin{align*}
\begin{aligned}
&\frac{\partial_+ \Pr [A(r)=T|A(t)=S]}{\partial r}\bigg|_{r=t} =\\
&\lambda(t) \cdot \left( \lim_{\delta\to 0^+}\Pr[A(t+\delta)=T|A(t)=S,N=1] - 1_{T=S}\right).
\end{aligned}
\end{align*}
\end{claim}
\begin{proof}
First, we note that the definition of a right derivative gives that:
\begin{align}
&\frac{\partial_+ \Pr [A(r)=T|A(t)=S]}{\partial r}\bigg|_{r=t}  = \nonumber\\&\lim_{\delta\rightarrow 0^+}\frac{1}{\delta} \Big( \Pr\left[A(t+\delta)=T|A(t)=S\right] - \Pr[A(t)=T|A(t)=S]\Big) =\nonumber \\
&  \lim _{\delta \rightarrow 0^+}\frac{1}{\delta}\Big( \Pr\left[A(t+\delta)=T|A(t)=S\right] - 1_{S=T}\Big)
\label{eq1-dervidef}
\end{align}
where \eqref{eq1-dervidef} follows from the observation that $\Pr[A(t)=T|A(t)=S]$ equals $1$ if $S=T$ and $0$ otherwise.

The law of total probability applied to $ \Pr[A(t+\delta)=T|A(t)=S]$ and conditioned on the number of events in the interval $ (t,t+\delta]$ being either $0$, $1$, or at least $2$, yields that:
\begin{align}
\Pr&[A(t+\delta)=T|A(t)=S]  \nonumber =\\ & \Pr[N=0|A(t)=S]\cdot \Pr[A(t+\delta)=T|A(t)=S,N=0] + \nonumber \\
& \Pr[N=1|A(t)=S]\cdot \Pr[A(t+\delta)=T|A(t)=S,N=1] + \nonumber \\
& \Pr[N\geq 2|A(t)=S]\cdot \Pr[A(t+\delta)=T|A(t)=S,N\geq 2] . \label{eq2-totalprob}
\end{align}
It is important to note that the random variable $A(t)$ depends only on all the random choices in the time interval $ [\varepsilon,t]$, i.e., the randomness of the Poisson process and the randomness of elements selection (step \ref{poisson:item_sample} in Algorithm \ref{alg:Poisson}) in the time interval $ [\varepsilon,t]$.
In contrast, the random variable $ N$ depends only on the randomness of the Poisson process in the time interval $ (t,t+\delta]$.
Since these two time intervals are disjoint, property \ref{Poisson-independ} of a Poisson process implies that $A(t)$ and $ N$ are independent random variables.
Hence, we can conclude from \eqref{eq2-totalprob} above that:
\begin{align}
(\Pr&[A(t+\delta)=T|A(t)=S]-1_{T=S}) =\nonumber\\ &
\Pr[N=0]\cdot \left(\Pr[A(t+\delta)=T|A(t)=S,N=0]-1_{T=S} \right) + \nonumber \\
& \Pr[N=1]\cdot \left(\Pr[A(t+\delta)=T|A(t)=S,N=1]-1_{T=S} \right) + \nonumber \\
& \Pr[N\geq 2]\cdot \left(\Pr[A(t+\delta)=T|A(t)=S,N\geq 2]-1_{T=S} \right)    =\nonumber  \\
& \Pr[N=1]\cdot \left(\Pr[A(t+\delta)=T|A(t)=S,N=1]-1_{T=S} \right) + \nonumber \\
& \Pr[N\geq 2]\cdot \left(\Pr[A(t+\delta)=T|A(t)=S,N\geq 2]-1_{T=S} \right) . \label{eq3new}
\end{align}
Equality \eqref{eq3new} follows from the simple observation that $ \Pr[A(t+\delta)=T|A(t)=S,N=0]$ equals $1$ if $T=S$ and $0$ otherwise. 

Plugging \eqref{eq3new} into \eqref{eq1-dervidef} gives:
\begin{align}
&\frac{\partial_+ \Pr [A(r)=T|A(t)=S]}{\partial r}\bigg|_{r=t}\nonumber =\\ & \lim _{\delta \rightarrow 0^+} \frac{1}{\delta}\Pr [N=1]\left( \Pr[A(t+\delta)=T|A(t)=S,N=1]-1_{T=S}\right) + \nonumber \\
& \lim _{\delta \rightarrow 0^+} \frac{1}{\delta}\Pr[N\geq 2] \left( \Pr[A(t+\delta)=T|A(t)=S,N\geq 2] - 1_{T=S}\right)  =\nonumber \\
 & \lambda(t) \left( \lim_{\delta\to 0^+}\Pr[A(t+\delta)=T|A(t)=S,N=1] - 1_{T=S}\right). \label{eq4new}
\end{align}

Equality \eqref{eq4new} follows from property \ref{Poisson-events} of a Poisson process, i.e., $\lim _{\delta \rightarrow 0^+}\frac{1}{\delta}\Pr[N_{(t,t+\delta]}=1]=\lambda(t)$ and $\lim _{\delta \rightarrow 0^+}\frac{1}{\delta}\Pr[N_{(t,t+\delta]}\geq 2]=0$.
\end{proof}

First, assume that there exists an element $ i\in S$ and $j\in U\setminus S$ such that $T=S-i+j$.
Therefore, in this case Claim \ref{claim:poisson} implies that:
\begin{align}
\begin{aligned}
&\frac{\partial _+\Pr [A(r)=T|A(t)=S]}{\partial r}\bigg|_{r=t} =\\& 
\lambda(t) \cdot \lim _{\delta \rightarrow 0^+}\Pr[A(t+\delta)=T|A(t)=S,N=1].
\end{aligned}\label{eq5-singleevent}
\end{align}
The right continuity of  $p_{(i,j)}(t,S)$ as a function of $t$ implies that:
\begin{align}
  \lim_{\delta\rightarrow 0^+}\Pr[A(t+\delta)=T|A(t)=S,N=1]=p_{(i,j)}(t,S),  \nonumber 
\end{align}
thus giving the first case.

Second, assume that $T=S$.
In this case, Claim \ref{claim:poisson} implies that:
\begin{align}
\begin{aligned}
&\frac{\partial _+\Pr [A(r)=T|A(t)=S]}{\partial r}\bigg|_{r=t} =\\ & \lambda(t)\cdot \left( \lim _{\delta \rightarrow 0^+}\Pr[A(t+\delta)=S|A(t)=S,N=1] - 1\right) .
\end{aligned}\label{eq6-new}
\end{align}
The right continuity of $ p_{(i,j)}(t,S)$ as a function on $t$ implies that:
\begin{align}
  \lim_{\delta\rightarrow 0^+}\Pr[A(t+\delta)=S|A(t)=S,N=1] = \sum _{i\in S}p_{(i,i)}(t,S),  \nonumber
\end{align}
thus giving the second case.

Finally, assume that we are not in any of the first two cases.
As before, recalling that $T\neq S$, Claim \ref{claim:poisson} implies that:
\begin{align}
\begin{aligned}
&\frac{\partial _+\Pr [A(r)=T|A(t)=S]}{\partial r}\bigg|_{r=t} =\\& 
\lambda(t) \cdot \lim _{\delta \rightarrow 0^+}\Pr[A(t+\delta)=T|A(t)=S,N=1]. 
\end{aligned}\label{eq7new}
\end{align}
Since $T$ differs from $S$ by at least two elements, it must be that $ \Pr[A(t+\delta)=T|A(t)=S,N=1]=0$ for every $ \delta >0$.
Taking the limit as $\delta\rightarrow 0^+$, the last case of the lemma is concluded.
\end{proof}

The following lemma lower bounds the rate in which the conditioned value of \gp increases.
\begin{lemma}\label{lem:value}
For every time $ \varepsilon \leq t<1$, base $S\in \cB$, and right continuous swap procedure:
\begin{align}
\begin{aligned}
  & \frac{\partial _+ V_{S,t}(r)}{\partial r} \bigg| _{r=t} \geq\\
  & \sum _{i\in S}\nabla _i F(t\indicator _S) + t\lambda(t) \cdot \mathbb{E}_{(i,j)\sim \swap(t,S)}\left[\nabla _jF(t\indicator _S) - \nabla _i F(t\indicator _S)\right] .\nonumber
   \end{aligned}
\end{align}
\end{lemma}
\begin{proof}
By the definitions of $ V_{S,t}(r)$~\eqref{value} and of conditioned expectation:
\begin{align} V_{S,t}(r)=\sum _{T\in \cB}F(r\indicator _T)\Pr[A(r)=T|A(t)=S] .\nonumber
\end{align}
Thus, using the chain rule provides:
\begin{align}
&\frac{\partial_+ V_{S,t}(r)}{\partial r}\bigg|_{r=t}
= \nonumber\\
&\sum_{T\in \cB}
\frac{\partial_+\!\left(
F(r\indicator_T)\Pr[A(r)=T\,|\,A(t)=S]
\right)}{\partial r}\bigg|_{r=t} =
\label{eq8-Vdef}
\\
& \sum_{T\in \cB}
\Bigg(
\frac{\partial_+ F(r\indicator_T)}{\partial r}\bigg|_{r=t}
\cdot
\Pr[A(t)=T\,|\,A(t)=S] + \nonumber
\\
&\qquad\quad
+F(t\indicator_T)\cdot
\frac{\partial_+ \Pr[A(r)=T\,|\,A(t)=S]}{\partial r}\bigg|_{r=t}
\Bigg).
\nonumber
\end{align}

Note that $ \Pr[A(t)=T|A(t)=S]=1$ only when $T=S$ and $0$ otherwise.
Moreover, for every $T\subseteq U$ the following holds: $ \frac{\partial _+ F(r\indicator _T)}{\partial r}\big| _{r=t}=\sum _{i\in T}\nabla _i F(t\indicator _T)$.
Therefore, 
\eqref{eq8-Vdef} implies that: 
\begin{align}
&\frac{\partial _+V_{S,t}(r)}{\partial r}\bigg| _{r=t} =\nonumber\\
&  \sum _{i\in S} \nabla _{i}F(t\indicator _S) + \sum _{T\in \cB}  F(t\indicator _T)\cdot \frac{\partial_+ \Pr[A(r)=T|A(t)=S]}{\partial r}\bigg| _{r=t} =\nonumber \\
& \begin{aligned} \sum _{i\in S} \nabla _{i}F(t\indicator _S) - F(t\indicator _S)\cdot \lambda(t)\left( 1-\sum _{i\in S}p_{(i,i)}(t,S)\right)+\\   +\sum _{i\in S, j\in U\setminus S}F(t\indicator_{S-i+j})\cdot {\lambda(t)\cdot p_{(i,j)}(t,S)} =\end{aligned}\label{eq10-probderiv} \\
&  \sum _{i\in S} \nabla _{i}F(t\indicator _S) +\lambda(t)\sum _{i\in S, j\in U\setminus S}{p_{ij}(t,S)} \cdot \left(F(t\indicator_{S-i+j})-F(t\indicator_S)\right) \label{eq10-probderiv2}.
\end{align}
Equality \eqref{eq10-probderiv} follows from Lemma \ref{lem:prob-derivative}.
Equality~\eqref{eq10-probderiv2} follows from the observation that if $i,j\in S$ and $i\neq j$ then $ p_{(i,j)}(t,S)=0$ since $S-i+j\in \cB$.
Hence, $ 1-\sum _{i\in S}p_{(i,i)}(t,S)=\sum _{i\in S,j\in U\setminus S}p_{(i,j)}(t,S)$.

Let us now focus on the second sum of \eqref{eq10-probderiv2}. From property~\eqref{multilinear-gradient} of multilinear extension, we have:
\begin{align}
    F(t\indicator_{S-i+j})-F(t\indicator_S) &= F(t\indicator_{S-i+j})- F(t\indicator_{S-i}) +F(t\indicator_{S-i}) -F(t\indicator_S) \nonumber\\
    &=t\cdot \nabla_jF(t\indicator_{S-i})-t\cdot \nabla_iF(t\indicator_S) \nonumber\\
    &\geq t\cdot \nabla_jF(t\indicator_{S})-t\cdot \nabla_iF(t\indicator_S) .\label{NEW-probderiv2}
\end{align}
Thus, using~\eqref{NEW-probderiv2} the second sum in~\eqref{eq10-probderiv2} simplifies as follows: 
\begin{align}
 &\lambda(t)\sum _{i\in S, j\in U\setminus S}{p_{ij}(t,S)} \left(F(t\indicator_{S-i+j})-F(t\indicator_S)\right) \geq\nonumber\\
 &  t \lambda(t)  \sum _{i\in S, j\in U\setminus S}p_{ij}(t,S)\left(\nabla_jF(t\indicator_{S})-\nabla_iF(t\indicator_S)  \right) . \nonumber
\end{align}
Plugging the above lower bound in~\eqref{eq10-probderiv2} concludes the proof.
\end{proof}

An additional lemma that is required for both above-average and $(\beta,\eta)$-approximate swap procedures is the following.

\begin{restatable}{lemma}{submodularstandard}\label{lem:submodular-standard}
For every $S,T\subseteq U$ and $ 0\leq t\leq 1$: $ \sum _{i\in T}\nabla _i F(t\indicator _S) \geq f(T) - F(t\indicator _S)$. 
\end{restatable}

\paragraph{Above-Average Swap.}
Using Lemmas~\ref{lem:value} and~\ref{lem:submodular-standard}, Lemma~\ref{lem:value-average-swap} can be now proved.
\begin{proof}[Proof of Lemma~\ref{lem:value-average-swap}]
Lemma~\ref{lem:value} and the definition of above-average~\eqref{eq:average} give:
\begin{align}
\begin{aligned}
&\frac{\partial _+ V_{S,t}(r)}{\partial r} \bigg| _{r=t} \geq\\
& \sum _{i\in S}\nabla _i F(t\indicator _S) + \frac{t\lambda(t)}{k} \cdot \left( \sum _{o\in O}\nabla _o F(t\indicator _S) - \sum _{i\in S}\nabla _i F(t\indicator _S)\right) .\end{aligned}\label{eq:above-average-1}
\end{align}
Recalling that $ \lambda(t)=\nicefrac{k}{t}$ implies that $ t\lambda(t)/k=1$ in~\eqref{eq:above-average-1}, thus:
\begin{align}
\frac{\partial _+ V_{S,t}(r)}{\partial r} \bigg| _{r=t}&\geq \sum _{i\in S}\nabla _i F(t\indicator _S) + \left( \sum _{o\in O}\nabla _o F(t\indicator _S) - \sum _{i\in S}\nabla _i F(t\indicator _S)\right) \nonumber\\
&= \sum _{o\in O}\nabla _o F(t\indicator _S).\label{eq:aove-average-2}
\end{align}
Focusing on the right hand side of~\eqref{eq:aove-average-2}, Lemma~\ref{lem:submodular-standard} implies that: $ \sum _{o\in O}\nabla _o F(t\indicator _S)\geq f(O) - F(t\indicator _S)$.
This proves the first part of the lemma.

Focusing on the second part of the lemma:
\begin{align}
 \frac{\partial _+Q(r)}{\partial r}\bigg| _{r=t} &= \sum _{S\in \cB} \Pr[A(t)=S]\cdot \frac{\partial_+ V_{S,t}(r)}{\partial r}\bigg| _{r=t} \nonumber\\ & \geq \sum _{S\in \cB}\Pr[A(t)=S]\cdot \left( f(O) - F(t\indicator _{S})\right) \label{inequality1} \\
 & = f(O) - \mathbb{E}[F(t\indicator _{A(t)})] = f(O) - Q(t). \nonumber
\end{align}
Inequality~\eqref{inequality1} follows from the first part of the lemma. 
\end{proof}

We are now ready to prove Theorem~\ref{thm:average_poisson}.
{We note that the fact that Lemma~\ref{lem:value-average-swap} lower bounds only the right derivative of $Q$ requires some care in the proof of the theorem.}
\begin{proof}[Proof of Theorem \ref{thm:average_poisson}]
We use the following claim, a slight variant of Theorem 1.2.1 ~\cite{lakshmikantham1969differential}.

\begin{restatable}{claim}{differntialEq}
\label{claim:differentialEQ}
Let $ q\colon [\varepsilon,1]\rightarrow \mathbb{R}$ be a function satisfying:
$(1)$ $q(t)$ is continuous for every $t\in (\varepsilon,1)$;
$(2)$ $q(t)$ is left continuous for $t=1$;
$(3)$ $q(t)$ is right continuous for $ t=\varepsilon$;
$(4)$ $q(\varepsilon)\geq 0$; and
$(5)$ for every $ t\in [\varepsilon,1)$ the right derivative of $q(t)$ is defined and satisfies: $ \frac{\partial _+ q(r)}{\partial r}\big|_{r=t} \geq f(O)-q(t)$.
Then: $ q(1)\geq \left( 1-e^{-(1-\varepsilon)}\right)f(O)$.
\end{restatable}

We note that $Q$ satisfies all the requirements of Claim \ref{claim:differentialEQ}:
$(1)$ requirements $(1)$, $(2)$, and $(3)$ (the continuity requirements) follow directly from the definition of $Q$\eqref{def:Q};
$(2)$ requirement $(4)$ is true since $Q$ is non-negative by definition; and
$(3)$ requirement $(5)$ follows from Lemma \ref{lem:value-average-swap}.
Thus, Claim \ref{claim:differentialEQ} gives:
\[ Q(1) \geq \left( 1-e^{-(1-\varepsilon)}\right)f(O)\geq (1-\varepsilon)(1-\nicefrac{1}{e})f(O).\]
The expected value of the output of \gp equals $Q(1)$.

The expected number of swap calls equals the expected number of events in the Poisson process:
\begin{align}
  \int _{t=\varepsilon}^1 \lambda (t)dt = \int _{t=\varepsilon}^1 \frac{k}{t}dt = k\cdot \ln{t}\bigg| _{t=\varepsilon}^1 = k\ln{(1/\varepsilon)} .  \nonumber
\end{align}
This concludes the proof.
\end{proof}

\paragraph{$(\beta,\eta)$-Approximate Swap.}
Using Lemmas~\ref{lem:value} and~\ref{lem:submodular-standard}, Lemma~\ref{lem:value-approximate-swap} can be now proved.
\begin{proof}[Proof of Lemma~\ref{lem:value-approximate-swap}]
Lemma~\ref{lem:value} and the definition of $(\beta,\eta)$-approximate~\eqref{eq:approximate_swap} give:
\begin{align}
\begin{aligned}
&\frac{\partial _+ V_{S,t}(r)}{\partial r} \bigg| _{r=t} \geq\\
& \sum _{i\in S}\nabla _i F(t\indicator _S) + \frac{t\lambda(t)}{k} \cdot \left( \beta\sum _{o\in O}\nabla _o F(t\indicator _S) - \sum _{i\in S}\nabla _i F(t\indicator _S)-\eta f(O)\right) .
\end{aligned}\label{eq:approximate-1}
\end{align}
Recalling that $ \lambda(t)=\nicefrac{k}{t}$ implies that $ t\lambda(t)/k=1$ in~\eqref{eq:approximate-1}, thus:
\begin{align}
\frac{\partial _+ V_{S,t}(r)}{\partial r} \bigg| _{r=t}\geq \beta \sum _{o\in O}\nabla _o F(t\indicator _S) - \eta f(O).\label{eq:approximate-2}
\end{align}
Lemma~\ref{lem:submodular-standard} implies that: $ \sum _{o\in O}F(t\indicator _S)\geq f(O)-F(t\indicator_S)$.
Plugging the latter into the right hand side of~\eqref{eq:approximate-2} gives the first part of the lemma.

Focusing on the second part of the lemma:
\begin{align}
\frac{\partial _+Q(r)}{\partial r}\bigg| _{r=t} &= \sum _{S\in \cB} \Pr[A(t)=S]\cdot \frac{\partial_+ V_{S,t}(r)}{\partial r}\bigg| _{r=t}\nonumber \\
& \geq \sum _{S\in \cB}\Pr[A(t)=S]\cdot \left( (\beta-\eta)f(O) - \beta F(t\indicator _{S})\right) \label{inequality2} \\
& = (\beta-\eta)f(O) - \beta \mathbb{E}[F(t\indicator _{A(t)})] = (\beta-\eta)f(O) - \beta Q(t). \nonumber
\end{align}
Inequality~\eqref{inequality2} follows from the first part of the lemma.
\end{proof}

We are now ready to prove Theorem~\ref{thm:approximate_poission}.
{We note that similarly to the proof of Theorem~\ref{thm:average_poisson} the fact that Lemma~\ref{lem:value-approximate-swap} lower bounds only the right derivative of $Q$ requires some care in the proof of the theorem.}
\begin{proof}[Proof of Theorem \ref{thm:approximate_poission}]
We use the following claim, a slight variant of Theorem 1.2.1 ~\cite{lakshmikantham1969differential}).
\begin{restatable}{claim}{differentEQapp}\label{claim:differentialEQ-approximate}
Let $ q\colon [\varepsilon,1]\rightarrow \mathbb{R}$ be a function satisfying:
$(1)$ $q(t)$ is continuous for every $t\in (\varepsilon,1)$;
$(2)$ $q(t)$ is left continuous for $t=1$;
$(3)$ $q(t)$ is right continuous for $ t=\varepsilon$;
$(4)$ $q(\varepsilon)\geq 0$; and
$(5)$ there exist $\beta\in (0,1] $ and $ \eta \geq 0$ such that for every $ t\in [\varepsilon,1)$ the right derivative of $q(t)$ is defined and satisfies: $ \frac{\partial _+ q(r)}{\partial r}\big|_{r=t} \geq (\beta-\eta)f(O)-\beta q(t)$.
Then: $ q(1)\geq \left( 1-\eta/\beta\right)\left( 1-e^{-\beta (1-\varepsilon)}\right)f(O)$.
\end{restatable}

We note that $Q$ satisfies all the requirements of Claim~\ref{claim:differentialEQ-approximate}:
$(1)$ requirements $(1)$, $(2)$, and $(3)$ (the continuity requirements) follow directly from the definition of $Q$\eqref{def:Q};
$(2)$ requirement $(4)$ is true since $Q$ is non-negative by definition; and
$(3)$ requirement $(5)$ follows from Lemma \ref{lem:value-approximate-swap}.
Thus, Claim \ref{claim:differentialEQ-approximate} gives:
\begin{align*}
 Q(1) &\geq \left( 1-\nicefrac{\eta}{\beta}\right) \left( 1-e^{-\beta (1-\varepsilon)}\right)f(O)\\
 &\geq (1-\varepsilon) (1-\nicefrac{\eta}{\beta})(1-e^{-\beta})f(O).
 \end{align*}
The expected value of the output of \gp equals $Q(1)$.

The expected number of swap calls equals the expected number of events in the Poisson process:
\begin{align}
  \int _{t=\varepsilon}^1 \lambda (t)dt = \int _{t=\varepsilon}^1 \frac{k}{t}dt = k\cdot \ln{t}\bigg| _{t=\varepsilon}^1 = k\ln{(1/\varepsilon)} .  \nonumber
\end{align}
This concludes the proof.
\end{proof}

\section{Applications}\label{sec:Applications}
\subsection{Submodular Welfare and Partition Matroids}\label{sec:Applications2}


\label{sec:simple}

We consider the \sw problem mentioned in Section~\ref{sec:intro}, and more generally the problem of maximizing a monotone submodular function subject to a partition matroid independence constraint.
We prove that for these problems, \gp admits a fast implementation in two settings.
The first is when the algorithm has value oracle access to the multilinear extension $F$ and the second is when it has value oracle access to the original set function $f$.

Formally, we are given a monotone submodular function $f\colon U\rightarrow \mathbb{R}_+$ along with a value oracle to its multilinear extension $F$ or the original set function $f$ (depending on the setting), and a partition matroid $\cM=(U,\cI)$ with partition $ U_1,\ldots,U_k$ of $U$.
The goal is to find $S\subseteq U$, where $ |S\cap U_j|\leq 1$ for every $ j\in [k]$, while maximizing $ f(S)$.
We denote the former setting of the problem (where oracle to $F$ is given) by \simpFprob, and the latter setting of the problem (where oracle to $f$ is given) by \simpfprob. 

\subsubsection{Oracle for $F$}
\label{sec:simple_F}
In order to utilize \gp for \simpFprob we provide a fairly straightforward above-average swap procedure: select a part $U_r$ uniformly at random, find an item $i^*\in U_r$ that maximizes $\nabla_{i^*}F(t\indicator _A)$, and return $i^*$ together with the single item in $A\cap U_r$.
This is summarized in Algorithm~\ref{alg:simpleF}.

\begin{algorithm}[t]
\caption{\simpleF $(t, A)$}
\SetKwInOut{Input}{input}
\SetKwInOut{Output}{output}
\SetAlgoNlRelativeSize{0}
\label{alg:simpleF}

     Denote $A= \{a_1,\dots, a_k\}$ where $a_j \in U_j$ for every $j\in [k]$. 

    Let $ r\sim Unif(1,\ldots,k)$.


           $i^*\leftarrow \argmax \{\nabla_i F(t \indicator_{A})\colon i\in U_r\}$.\label{simpleF:max}

       Return  $(a_r,i^*)$.

\end{algorithm}

\begin{lemma}
\label{lem:simpleF_average}
Algorithm~\ref{alg:simpleF} is a right continuous above-average swap procedure.
\end{lemma}
\begin{proof}

First, we observe that Algorithm~\ref{alg:simpleF} always returns a pair of elements $ (a_r,i^*)$, where $a_r\in A$ and $i^*\in U$, satisfying: $A-a_r+i^*\in \cB$, as both belong to the part $U_r$.
Therefore, it always returns a valid swap.

Second, we prove that Algorithm~\ref{alg:simpleF} is above-average.
Let $ O=\{ o_1,\ldots,o_k\}\in \cI$ be an optimal base, such that $ o_j\in U_j$ for every $ j\in [k]$. That is $f(O)=\max \{ f(S)\colon S\in \cI\}$.
Moreover, for every $j\in [k]$ let $i_j\triangleq \argmax \{ \nabla_{i} F(t \indicator_{A})\colon i\in U_j\}$ be the element that Algorithm~\ref{alg:simpleF} returns if $r=j$.
Clearly, by definition $\nabla_{i_j}F(t \indicator_{A}) \geq\nabla_{o_j}F(t \indicator_{A})$ for every $j\in [k]$.
Hence,
\begin{align}
\E _{r\sim Unif(1,\ldots,k)}\left[ \nabla_{i^*} F(t \indicator_{A} ) - \nabla_{a_r} F(t \indicator_{A})\right]\,&=  \frac{1}{k}  \sum_{j=1}^{k} \left( \nabla_{i_j} F(t\indicator_{A} ) -\nabla_{a_j} F(t \indicator_{A} )\right)\nonumber \\
&\geq  \frac{1}{k} \cdot \sum_{j=1}^{k} \left( \nabla_{o_j} F(t \indicator_{A} ) -\nabla_{a_j} F(t \indicator_{A} )\right) .\nonumber
\end{align}
Thus, Algorithm~\ref{alg:simpleF} is above-average. 

Finally, we note that the output of Algorithm~\ref{alg:simpleF} can be easily made right continuous by an appropriate choice of $\argmax$ in Line~\ref{simpleF:max}.

 \end{proof}

Once we have a swap procedure, we can proceed to the proof of Theorem~\ref{thrm:PartitionF}. 
\begin{proof}[Proof of Theorem~\ref{thrm:PartitionF}]
By Theorem~\ref{thm:average_poisson} it holds that \gp together with Algorithm~\ref{alg:simpleF} returns  a solution to the problem whose expected value is at least  $(1-\eps)(1-\nicefrac{1}{e})f(O)$ while performing in expectation $k \ln(\nicefrac{1}{\eps})$ calls to Algorithm~\ref{alg:simpleF}.
As every call to Algorithm~\ref{alg:simpleF} requires $ 2|U_r|$ evaluations of $F$, we can conclude that the expected running time is: $ k\ln{(\nicefrac{1}{\varepsilon})}\cdot \mathbb{E}_{r\sim Unif(1,\ldots,k)}\left[2|U_r|\right] = O(n\ln{(\nicefrac{1}{\varepsilon})})$ (since $ \mathbb{E}_{r\sim Unif(1,\ldots,k)}[|U_r|] = \nicefrac{n}{k}$).
\end{proof}

\subsubsection{Oracle for $f$}
\label{sec:partition_orc_to_f}



We now turn to the case where our oracle access is to  $f$, rather than $F$. 
In this setting, both $F(x)$ and its derivatives must be estimated via repeated sampling, which can substantially increase the running time. 
To control this overhead, we introduce a modified version of Algorithm~\ref{alg:simpleF} that incorporates a \mab algorithm, together with a pre-processing step that guaranties an approximation ratio arbitrarily close to $(1 - \nicefrac{1}{e})$ within our objective running time.

In order to present our results we need some additional technical definitions. 
We use $I=(f,U, (U_1,\ldots,U_k))$ to denote a $\simpfprob$ instance. The matroid of the instance is the simple partition matroid $(U, \cI)$ where $\cI =\{S\subseteq U\,|\, \forall j\in [k]: \abs{S\cap U_j}\leq 1\}$. The {\em  marginals to opt ratio} of the instance $I$ is 
\begin{equation}
\marsum(I)\,=\,\frac{\max_{S\in \cI} \sum_{i\in S} (f(\{i\})-f(\emptyset))}{f(O)}  \,=\, \frac{1}{f(O)}\cdot\sum_{j=1}^{k}\max_{i\in U_j} (f(\{i\})-f(\emptyset)), \label{eq:ratio}
\end{equation}
where $O=\argmax\{ f(O) \,|\,O\in \cI\}$ is an optimal solution. 
That is, $\marsum(I)$ is the ratio between 
the maximum sum of marginals an independent set can attain and the optimum. 

The marginals to opt ratio of the instance governs the quality of the swap procedure that can be achieved within an expected running time of 
$O\!\left(\frac{n}{k}\cdot \frac{1}{\delta^2} \ln\!\frac{1}{\delta}\right)$.
\begin{lemma}
\label{lem:bandit_swap_general}
There is $(1,\eta)$-approximate swap algorithm for \simpfprob, where $\eta = \delta\cdot \marsum(I)$,  that uses $O\left(\frac{n}{k}\cdot \frac{1}{\delta^2}\cdot \ln\left(\frac{1}{\delta}\right)\right)$ oracle queries in expectation.
\end{lemma}
We give the proof of Lemma~\ref{lem:bandit_swap_general} towards the end of this section. 
We can use the algorithm from Lemma~\ref{lem:bandit_swap_general} as the swap algorithm in \gp. By Theorem~\ref{thm:approximate_poission} this leads to the following result. 

\begin{corollary}
\label{lem:BanditSwapPoisson}
There exists an algorithm for monotone submodular maximization with a  partition matroid which uses $O\left( \frac{n}{\delta^2} \cdot \ln\left(\frac{1}{\delta}\right)\right)$ value queries for $f$ in expectation and returns a solution $S$ such that 
$$
\E[f(S)]\,\geq\,\left(1-\delta\right) \cdot \left(1-\delta \cdot \marsum(I)\right)\cdot (1-\nicefrac{1}{e}) \cdot f(O),
$$
where $I$ is the instance and $O$ is an optimal solution. 
\end{corollary}

Our goal is to invoke Corollary~\ref{lem:BanditSwapPoisson} to achieve a $(1-\eps)(1-\nicefrac{1}{e})$-approximation for the problem. 
For this  to work, the instance must have a sufficiently small marginals to opt ratio, $\marsum(I)$. 
We achieve this condition through a pre-processing stage which can be implemented in linear time. 

Given a $\simpfprob$ instance $I=(f,U, (U_1,\ldots, U_k))$ and an independent set $P\in \cI$ we define the {\em residual instance of $I$ and $P$}  as $I_P=(f_P, U_P, (U_j)_{j\in J_P})$  where $J_P=\{j\in [k]\,|\, U_j\cap P=\emptyset\}$, $U_P=\bigcup_{j\in J_P} U_j$ and $f_P:2^{U_P} \rightarrow \mathbb{R}_{\geq 0}$ is defined by $f_P(S)= f(P\cup S) -f(P)$. 
Conceptually, the residual instance represents the subproblem that remains when $P$ is enforced as part of the solution.
In particular, the universe of $I_P$ consists of all parts of the partition which do not intersect with $P$. It can be easily observed that if $S$ is a solution for $I_P$, then $S\cup P$ is a solution for $I$ with value $f(S\cup P)=f(S)+f_P(S)$. The pre-processing returns a set $S$ for which the marginals to opt ratio is small with high probability. 

\begin{lemma}[pre-processing]\label{lem:simpleprep}
There exists a randomized algorithm, which given a \simpfprob\ instance $I=(f,U, (U_1,\ldots, U_k))$ and error parameter $0<\delta<\nicefrac{1}{2}$, performs in expectation $ O(\nicefrac{n}{\delta})$ value queries and returns $ P\in \cI$ satisfying:
\begin{enumerate}
    \item $ \Pr\left[\marsum(I_P) \leq c\cdot \frac{f(O)}{f(O_P)}\right]\geq 1-\delta$,  
    \item $\mathbb{E}\left[(1-\nicefrac{1}{e})\cdot f(O_P\cup P) + \frac{1}{e}\cdot f(P) \right]\geq (1-\nicefrac{1}{e}-\delta)\cdot f(O)$. 
\end{enumerate}
Here, $ O$ is an optimal solution for $I$, and $O_P$ is an optimal solution for the residual instance $I_P$ of $I$ and $P$. Furthermore, $c>1$ is an absolute constant. 
\end{lemma}
We give the proof of Lemma~\ref{lem:simpleprep} in Appendix~\ref{sec:preprocessing}. Theorem~\ref{thrm:partition} is derived from applying the algorithm in Corollary~\ref{lem:BanditSwapPoisson} in conjunction with the pre-processing method outlined in Lemma~\ref{lem:simpleprep}.

\begin{proof}[Proof of Theorem~\ref{thrm:partition}]

The algorithm we define is the following:
\begin{enumerate}
 \item Execute the pre-processing algorithm  of Lemma~\ref{lem:simpleprep} on the original instance $I$ with the error parameter $\delta_1=\eps/8$  to find $P$. 
\item Use \gp  together with the swap algorithm from Lemma~\ref{lem:bandit_swap_general} (Corollary~\ref{lem:BanditSwapPoisson})  on $I_P$, the residual instance  of $I$ and $P$, with the error parameter $\delta_2 =\nicefrac{\varepsilon}{(8c)} $
to obtain $A$.
\item  Return $A\cup P$.
\end{enumerate}

First, by the definition of residual instances,  it is easy to note that $ A\cup P\in \cI$.  
Second, 
following Lemmas \ref{lem:BanditSwapPoisson} and \ref{lem:simpleprep} the number of oracle queries used by the algorithm in expectation is at most $O(n\log^2{(1/\varepsilon)}/\varepsilon^2)$.

Third and finally, let us lower bound the expected value of the output: $ \mathbb{E}[f(A\cup P)]$.
For every $R\in \cI$ let $O_R$ be an optimal solution for $I_R$,  the residual instance of $I$ and $R$. 
Also, let $\cP$ be the set of all $P\in \cI$ such that $\marsum(I_P) \leq c\cdot \frac{f(O)}{f(O_P)}$. 
By Corollary~\ref{lem:BanditSwapPoisson} , for every $R\in \cP$ it holds that,
$$
\begin{aligned}
 \E[ f_P(A) \,|\, P=R] \,&\geq \left(1-\delta_2\right) \cdot \left(1-\delta_2\cdot \marsum(I_R)\right)\cdot (1-\nicefrac{1}{e}) \cdot f_R(O_R)\\
&\geq \, \left(1-\frac{\eps}{8\cdot c}\right) \cdot  \left(1-\frac{\eps }{8\cdot c} \cdot \frac{c\cdot f(O)}{f_R(O_R)}\right)\cdot (1-\nicefrac{1}{e}) \cdot f_R(O_R)\\
&\geq \left(1-\frac{\eps}{8}\right)\cdot \left(1-\nicefrac{1}{e}\right)\cdot f_R(O_R)-\frac{\eps}{8}\cdot f(O). 
\end{aligned}
$$
The second inequality holds as $R\in \cP$ and $\delta_2= \frac{\eps}{8c}$. Therefore,
\begin{equation}
\label{eq:simple_exp}
\begin{aligned}
\E[f_P(A)]  \, &=\, \sum_{R\in \cI} \Pr[P=R]\cdot  \E[ f_P(A) \,|\, P=R]\\
&\geq \, \sum_{R\in \cP} \Pr[P=R]\cdot  \E[ f_P(A) \,|\, P=R] \\
&\geq\, \sum_{R\in \cP} \Pr[P=R]\cdot\left( \left(1-\frac{\eps}{8}\right)\cdot \left(1-\nicefrac{1}{e}\right)\cdot f_R(O_R)-\frac{\eps}{8}\cdot f(O)\right)\\
&=\, \left( 1-\frac{\eps}{8}\right) \cdot \E\left[ (1-\nicefrac{1}{e}) \cdot f_P(O_P) \right]  - \left(1-\frac{\eps}{8}\right) \left(1-\frac{1}{e}\right)\cdot \sum_{R\in \cI\setminus \cP} \Pr[P=R] \cdot f_R(O_R)-\frac{\eps}{8}\cdot f(O).
\end{aligned}
\end{equation}
Furthermore, 
\begin{equation}
\label{eq:simple_bad}
\sum_{R\in \cI\setminus \cP} \Pr[P=R] \cdot f_R(O_R) \leq f(O) \cdot \Pr[P\notin \cP] \leq f(O)\cdot \delta_1 = \frac{\eps}{8}\cdot f(O), 
\end{equation}
where the last inequality follows from Lemma~\ref{lem:simpleprep}. 
Plugging \eqref{eq:simple_bad} into \eqref{eq:simple_exp} we attain 
\begin{equation}
\label{eq:simp_expA}
\E[f_P(A)] \geq  \left( 1-\frac{\eps}{8}\right) \cdot \E\left[ (1-\nicefrac{1}{e}) \cdot f_P(O_P) \right]  -   \frac{\eps}{ 8} \cdot f(O) -   \frac{\eps}{ 8} \cdot f(O).
\end{equation}

Using \eqref{eq:simp_expA} and Lemma~\ref{lem:simpleprep}
 we have
 $$
 \begin{aligned}
 \E[f(A\cup P)] \,&= \,\E[f(P)+f_P(A)]\\
 &=\,\E\left[ f(P) + \left(1-\frac{\eps}{8}\right) (1-\nicefrac{1}{e})\cdot f_P(O_P)\right] -\frac{\eps}{4}\cdot f(O)\\
 &\geq\left( 1-\frac{\eps}{8}\right)\cdot  \E\left[ \frac{1}{e}\cdot f(P) + \left(1-\frac{1}{e}\right)\cdot f(P\cup O_P) \right] -\frac{\eps}{4}\cdot f(O)\\
 &\geq \left( 1-\frac{\eps}{4}\right) \cdot \left(1-\frac{1}{e} -\delta_1\right) \cdot f(O) -\frac{\eps}{4}\cdot f(O) \\
 &\geq (1-\nicefrac{1}{e} -\eps)\cdot f(O),
 \end{aligned}
 $$
 which concludes the proof. 
\end{proof}

\subsubsection*{Proof of Lemma~\ref{lem:bandit_swap_general}}

In order to present the swap algorithm for simple partition matroids and oracle for $f$ we first need to define the well-known \mab problem (see, e.g.,~\cite{EMM06}) where we focus on  identifying best arm.

\begin{definition}
\label{def:mab}
In the \mab problem we are given $n$ distributions $ \{ \cD_i\} _{i=1}^n$ with support $ [0,1]$.
Denote by $ \mu _i\triangleq \mathbb{E}_{Z\sim \cD _i} [Z]$ for every $ 1\leq i\leq n$.
The goal is to find $i$ for which $ \mu _i$ is maximized.
\end{definition}
It is important to note that the distributions $ \{ \cD_i\}_{i=1}^n$ are not explicitly given to the algorithm.
The only method by which the algorithm can interact with a distribution $\cD_i$ is to obtain samples from it.
All samples are independent, regardless of whether they are obtained from the same distribution or from different distributions.

We utilize the following result  of \cite{EMM06} for \mab in the PAC model.
\begin{theorem}[\cite{EMM06}]\label{thrm:mab}
There exists a randomized algorithm that given an instance of the \mab problem and $\delta >0$, samples at most $O(n\log{(1/\delta)}/\delta^2)$ times from the distributions $ \{ \cD_i\} _{i=1}^n$ and returns an $ i$ satisfying: $ \mathbb{E}[\mu _i]\geq \mu _{i^*} - \delta$.
Here, $ i^*=\argmax _{1\leq i\leq n}\mu _i$.
\end{theorem}

In our swap algorithm we use the fact that 
$$\nabla_i F(t\cdot \indicator_A)\,=\,F(t\cdot \indicator_{A+i}) - F(t\cdot \indicator_{A-i})  \,=\,\E_{R\sim t\cdot \indicator_A}\left[f(R\cup\{i\}) -f(R\setminus\{i\})\right],$$ 
and furthermore, we can sample the values out of the distribution of  $f(R\cup\{i\}) -f(R\setminus\{i\}) $ by sampling a random set $R$ and initiating two value queries for $f$.
This allows us to use the \mab\ algorithm  of Theorem~\ref{thrm:mab} to find $i\in U_r$ which  (approximately)  maximizes $\nabla_i F(t\cdot \indicator_A)$.
For every $j\in [k]$ let $\alpha_j= \max_{i\in U_j} \left( f(\{i\})-f(\emptyset)\right)$. 
In order to meet the requirements of the \mab\ algorithm we need to scale the values of $f(R\cup\{i\}) -f(R\setminus\{i\}) $ by a factor of $\frac{1}{\alpha_j}$ to attain samples within the range $[0,1]$.
We give the pseudocode of the swap procedure in Algorithm~\ref{alg:BanditSwap}.
\begin{algorithm}[t]
\caption{\BanditSwap $(t, A)$}
\SetKwInOut{Input}{input}
\SetKwInOut{Output}{output}
\SetKwInOut{Configuration}{Config}
\Configuration{Accuracy parameter $\delta>0$ }

\SetAlgoNlRelativeSize{0}
\label{alg:BanditSwap}

     Let $A= \{a_1,\dots, a_k\}$ where $a_j \in U_j$ for all $j\in [k]$. 

      Select $r\in [k]$ uniformly at random.

    Let $ \cD_i$ be the distribution of $ \left(f(R\cup \{ i\})-f(R\setminus \{ i\})\right)/\alpha_r$ where $ R\sim t\indicator_A$, $ \forall i\in U_r$ and $\alpha_r = \max_{i\in U_r} \left(f(\{i\})-f(\emptyset)\right)$.\label{BanditSwap:distribution}

    let $i^*$ be the output of the algorithm of Theorem \ref{thrm:mab} on $ \{ \cD_i\}_{i\in U_r}$ and $ \delta$. \label{BanditSwap:mab}
    
    Return  $a_r$ and $i^*$.    
 	
\end{algorithm}

\begin{proof}[Proof of Lemma~\ref{lem:bandit_swap_general}]
We prove that Algorithm~\ref{alg:BanditSwap} satisfies the conditions of the lemma.

We first observe that Algorithm~\ref{alg:BanditSwap} is indeed a swap algorithm as both $i^*\in U_r$ and $a_r \in U_r$. This implies that $A-a_r +i^*$ remains a basis of the matroid.  

The number of initiated queries to $f$ can be upper bounded by a simple argument. Conditioned on $r=j$, by Theorem~\ref{thrm:mab}, the number of queries to the distributions $\cD_i$ initiated in Line~\ref{BanditSwap:mab} is $O\left(\abs{U_j} \cdot \frac{1}{\delta} \cdot \ln\left(\frac{1}{\delta}\right)\right)$. 
Each query to  one of the distributions $\cD_i$ can be implemented by sampling $R\sim t\cdot \indicator_A$ and two queries for $f$. 
Also, the value $\alpha_j$ can be computed once using $O(\abs{U_j})$ queries to $f$. 
Overall,  the  number of queries to $f$, conditioned on $r=j$, is  $O\left(\abs{U_j} \cdot \frac{1}{\delta} \cdot \ln\left(\frac{1}{\delta}\right)\right)$. 
Since $r$ is uniformly random this implies that the (unconditioned) expectation of the number of queries to $f$ is $O\left(\frac{n}{k} \cdot \frac{1}{\delta} \cdot \ln\left(\frac{1}{\delta}\right)\right)$.

We are left to show the algorithm is indeed a $(1,\eta)$-approximate swap algorithm for the instance. 
Let $O=\{o_1,\ldots, o_k\}$ be an optimal solution for the instance such that $o_j\in U_j$ for every $j\in[k]$. 
We first condition on $r=j$ for some $j\in [k]$, and define $\mu_i = \E_{Z\sim \cD_i}[Z] = \frac{\nabla_i F(t\cdot \indicator_A)}{\alpha_j}$ for every $i\in U_j$. By Theorem~\ref{thrm:mab} we have
$$
\frac{1}{\alpha_j}\cdot \E[\nabla_{i^*} F(t\cdot \indicator_A)\,|\,r=j] = \E[\mu_{i^*} \,|\,r=j] \geq \max_{i\in U_j} \mu_i - \delta  \geq \mu_{o_j} -\delta.
$$
Therefore, 
$$
\E[\nabla_{i^*} F(t\cdot \indicator_A)\,|\,r=j]   \,\geq \,\alpha_j\cdot \mu_{o_j} -\delta\cdot \alpha_j \,=\, \nabla_{o_j} F(t\cdot\indicator_A) -\delta \cdot\alpha_j. 
$$
Since $r\in [k]$ is uniformly random this implies that 
$$
\E[\nabla_{i^*} F(t\cdot \indicator_A)] \,=\,\frac{1}{k}\cdot \sum_{j=1}^{k} \E[\nabla_{i^*} F(t\cdot \indicator_A)\,|\,r=j] \,\geq \, \frac{1}{k} \cdot \left(\sum_{j=1}^{k }\nabla_{o_j} F(t\cdot\indicator_A) -{\delta} \cdot \sum_{j=1}^{k}\alpha_j\right).
$$

Similarly, since $r\in[k]$ is uniformly random,  we have $\E\left[\nabla_{a_r} F(t\cdot \indicator_A) \right] =\frac{1}{k}\cdot \sum_{j=1}^{k}  \nabla_{a_j} F(t\cdot \indicator_A) $.
Therefore,
$$
\begin{aligned}
 \E[\nabla_{i^*} F(t\cdot \indicator_A)-\nabla_{a_r} F(t\cdot \indicator_A) ]\,&\geq\, \frac{1}{k} \left(\cdot \sum_{j=1}^{k }\nabla_{o_j} F(t\cdot\indicator_A) -\delta \cdot \sum_{j=1}^{k}\alpha_j-\sum_{j=1}^{k}  \nabla_{a_j} F(t\cdot \indicator_A) \right)
    \\
& =\frac{1}{k}\cdot  \left( \sum_{j=1}^{k }\nabla_{o_j} F(t\cdot\indicator_A)   -\sum_{j=1}^{k}  \nabla_{a_j} F(t\cdot \indicator_A)\right)  -\frac{\eta}{k}\cdot f(O),
\end{aligned}$$ 
where the last equality follows from the definition of $\eta$. Overall, we showed that Algorithm~\ref{alg:BanditSwap} is a $(1,\eta)$-approximate swap algorithm.
\end{proof}

\subsection{The Generalized Assignment and Separable Assignment Problems}\label{sec:GAPSAP}
We consider an application of our techniques to both the \gaplong (\gap) and the \saplong (\sap) and prove Theorem~\ref{thm:gap_main} and Theorem~\ref{thm:sap_main}.
{In order to prove the theorems, we follow the treatment in~\cite{CCPV11} and consider the more general \saplong, whose algorithm is based on reducing the problem to maximizing a monotone submodular function subject to a partition matroid independence constraint.}
The main challenge is that the ground set of the submodular function and the matroid is exponential in size.
Nonetheless, we use \gp for the problem.
The main ingredient in the implementation is to show that the there is a suitable swap procedure that can be implemented fast. 

\paragraph{\saplong.}
An instance of \sap consists of $n$ items and $m$ bins.
Each bin $i$ has an associated collection of feasible sets $\F_i\subseteq 2^{[n]}$ which is down-closed: $A \in \F_i$ and  $B\subseteq A$ imply that $B \in \F_i$.
Each item $j$ and bin $i$ have a value $v_{ij}$. 
The goal is to choose disjoint feasible sets $S_i \in \F_i$ of items while maximizing: $\sum_{i=1}^m \sum_{j\in S_i}v_{ij}$.
Observe that \sap captures \gap when each $\F_i$ is the collection of all subsets of items whose size with respect to bin $i$ is at most one: $\F_i = \{ S\subseteq [n]\colon \sum _{j\in S}s_{ij}\leq 1\}$.

We assume there is an efficient algorithm \ApproxKnapsack that for any bin $i$, given non-negative weights $w_j$ on items for each $j\in [n]$, returns in time $p(n,\alpha)$ an $\alpha$-approximation to the maximum weight set $T\in \F_i$ where the weight of a set $T$ is defined as: $w(T)=\sum_{j\in T} w_j$.
For \gap the above problem is exactly the knapsack problem with $n$ items and there exists an efficient $(1-\varepsilon)$-approximation that runs in time $\tilde{O}(n+\nicefrac{1}{\varepsilon^2})$~\cite{chen2024nearly}. 

We reduce \sap to maximizing a monotone submodular function subject to a partition matroid independence constraint (as in~\cite{CCPV11}).
This is achieved by utilizing \gp and presenting an $(\alpha,0)$-approximate swap procedure that can be implemented fast. 

\def \X {\mathcal{X}}
\paragraph{Reduction to Monotone Submodular Maximization subject to a Partition Matroid.}
Set the ground set $\X=\{(i,S):i\in [m], S\in \F_i\}$.
We define a function $f\colon 2^{\X}\rightarrow \mathbb{R}_+$ as follows: 
$$ f(\mathcal{S})\triangleq  \sum_{i=1}^m \max \{v_{ij}\colon (i,S)\in \mathcal{S}, j\in S\} ~~~~~\forall \mathcal{S}\subseteq \X.$$

Here $\mathcal{S}\subseteq \X$ denotes a collection of valid assignments of items to each bin, where each item can be assigned to multiple bins.
The function $f$ on such an assignment picks the value of the item with respect to the \emph{best} bin that it is assigned to.
It is easy to verify that $f$ is monotone and submodular.
The reduction is completed by maximizing this function subject to a partition matroid independence constraint.
More precisely, the matroid is $\cM=(\X,\cI)$ where $\mathcal{S}\in \cI$ if for every $i$ it contains at most a single $(i,S_i)$ that corresponds to a valid assignment of items to bin $i$.
Thus, $\cM$ is a partition matroid where the partition $ \X_1,\ldots,\X_m$ of $\X$ is defined as $ \X_i = \{ (i,S)\in \X\}$.

Given an instance of \sap, we use the above reduction and apply \gp.
While the ground set is exponential in size, observe that the rank of the matroid $\cM$ is only $m$ (the number of bins). Thus, the expected total number of iterations taken by \gp is $m\ln{(\nicefrac{1}{\epsilon})}$.
We now show how to implement an $(\alpha,0)$-approximate swap procedure fast using \ApproxKnapsack. 

\paragraph{The Swap Procedure.}
Let $\mathcal{A}\subseteq \mathcal{X}$ denote a basis of the matroid, i.e, $\mathcal{A}=\{(i,S_i):i\in [m]\}$ where each bin is assigned a single set $S_i\in \F_i$ of items (although each item can be assigned to multiple bins).
Given a time $t$ and a basis $\mathcal{A}$, the swap procedure picks a bin $i$ (or equivalently $({i},S_{i})$) uniformly at random and finds $(i,T_i)$ such that $T_i$ is an $(\alpha-\varepsilon)$-approximate optimum solution to the following:
\begin{equation}\label{eqn:sep_swap}
    \max \{ \nabla F(t\A)^\top {\bf{1}}_{(i,T)}\colon T\in \F_i\}.
\end{equation}
Now the swap procedure returns the pair of elements $(i,S_i)$ and $(i,T_i)$ ($(i,S_i)$ is to be swapped with $(i,T_i)$).
Since $(i,S_i)$ is a uniform random element in the base $\mathcal{A}$ and $(i,T_i)$ is an $(\alpha-\epsilon)$-approximation to~\eqref{eqn:sep_swap}, it is easy to verify that the swap procedure is indeed $(\alpha-\varepsilon,0)$-approximate. 

We now show how to solve the problem $\eqref{eqn:sep_swap}$ using the given \ApproxKnapsack algorithm.
For any $T\in \F_i\setminus \{S_i\}$, we claim that:
$$  \nabla F(t\A)^\top {\bf{1}}_{(i,T)}= \sum_{j\in T} w_j$$
for weights $w\colon [n]\rightarrow \mathbb{R}_+$ we define shortly.
Given these weights, we use  \ApproxKnapsack  to compute $\max\{w(T)\colon T\in \F_i\}$.
We also compute $ \nabla F(t\A)^\top {\bf{1}}_{(i,S_i)}$ separately and return the better of the two solutions as $T_i$.  



Now we define the weights.
Set $I_{kj}=1$ for any bin $k$ and item $j$ if $j\in S_k$ {and $0$ otherwise}.
For every $j\in [n]$, let $\sigma_j\in S_m$ denote the permutation such that $v_{\sigma_j(1),j}\geq v_{\sigma_j(2),j}\geq \ldots \geq v_{\sigma_j(m),m}$ and let $\sigma_j(l_j)={i}$. Then for any $T\in \F_i\setminus \{S_i\}$, a simple calculation shows that:
\begin{align*}
    \nabla F(t\A)^\top {\bf{1}}_{(i,T)}&= F(t\A+{\bf{1}}_{(i,T)}) - F(t\A)\\
   & =\sum_{j\in T} \left(I_{i,j} v_{i,j} \left(1 -t\right)^{\sum_{i'<l_j} I_{\sigma_j(i'),j}} -\sum_{k=l_j}^m
    v_{\sigma_j(k),j}t I_{\sigma_j(k),j} \left(1 -t\right)^{\sum_{i'<k} I_{\sigma_j(i'),j}} \right) .
\end{align*}

Thus, we define $w_j$ as the {$j$\textsuperscript{{th}}} summand in the equation above. For the special case of $(i,S_i)$ observe that: 
\begin{align*}
    \nabla F(t\A)^\top {\bf{1}}_{(i,S_i)}&= \frac{1}{1-t}\left(F(t\A+{\bf{1}}_{(i,S_i)}) - F(t\A)\right)\\
   & =\frac{1}{1-t} \sum_{j\in S_i} w_j ,
\end{align*}
and thus $ \nabla F(t\A)^\top {\bf{1}}_{(i,S_i)}$ can also be computed efficiently after computing the weights. 
A naive implementation to compute the weights is to first pre-compute the permutations $\sigma_j$ for each $j$ in time $O(mn\log m)$ before the beginning of the iterations. The algorithm also maintains arrays storing the binary values $I_{kj}$ for each $k\in [m]$ and $j\in [n]$. In every iteration, for job $j\in [n]$, we make one pass through the permutation to compute $w_j$ in linear time $O(m)$. Updating the values $I$ only takes $O(n)$ time. Thus the total time spent per iteration is $O(mn)$. This gives a total running time of $\tilde{O}\left(\frac{m}{\varepsilon} \left(p(n,\varepsilon)+ mn\right)\right)$.

We now show how to reduce the per iteration time complexity to $\tilde{O}(\frac{n}{\varepsilon})$ using standard rounding of weights. As an initialization step of the algorithm, we round the values $v_{kj}$ to the nearest power of $(1+\epsilon)$. 
Ignoring all $v_{kj}\leq n\cdot m \cdot \max_{k,j} v_{kj}$, we obtain that the total number of distinct values of $v_{kj}$ is at most $O
\left(\log{(mn)}/\varepsilon\right)$. Thus, it is enough to store \emph{integer} values $I'_{cj}$, number of machines $k$ such that $v_{kj}=c$ and $j\in S_k$. The weights can be computed using these $\tilde{O}\left(\frac{n}{\varepsilon}\right)$ integer values by making one pass through these values. Putting everything together, we get the guarantee claimed in Theorem~\ref{thm:sap_main}.  

\subsection{Generalized Partition Matroid}

Finally, we consider generalizations of the submodular maximization problems under a partition matroid, discussed in Section~\ref{sec:simple}, to general partition matroids.
Formally, we are given a monotone submodular function $f\colon U\rightarrow \mathbb{R}_+$ along with a value oracle to its multilinear extension $F$ or the original set function $f$ (depending on the setting), and a generalized partition matroid $\cM=(U,\cI)$ with partition $ U_1,\ldots,U_k$ of $U$, where each part $U_j$ has an upper bound $\ell_j\in\mathbb{N}$ on the number of items which can be selected from this part.
The goal is to find $S\subseteq U$, where $ |S\cap U_j|\leq \ell_j$ for every $ j\in [k]$, while maximizing $ f(S)$.
We refer to the problem with an oracle to $F$ as \partFprob, and to the problem with an oracle to $f$ as \partfprob.

\subsubsection{Oracle for $F$}
\label{sec:generalizeF}

Let $(U,\cI)$, where $\cI = \{ S\subseteq U\,|\, \forall j\in [k]:~\abs{S\cap U_j}\leq \ell_j\}$ be the matroid of the instance, and  let $r=\sum_{j=1}^{k}\ell_j$ be the rank of the matroid. We also denote a \partFprob\ instance by the a tuple \((f, U, \{U_j\}_{j \in [k]}, \{\ell_j\}_{j \in [k]})\). 

Our goal is to provide a fast approximation algorithm for \partFprob\ based on the \gp, by introducing a fast swap algorithm for \partFprob\ instances.
The swap algorithm we  present builds upon the same core ideas as \simpleF\ (Algorithm~\ref{alg:simpleF}).
However, in order to maintain a near-linear number of queries we  also incorporate ideas taken from the Stochastic-Greedy of \cite{MBKV15} (the same algorithm appears under the name ``Random Sampling Algorithm'' in \cite{BFS17}). 
Specifically, the swap algorithm randomly selects a part $j\in [k]$, and subsequently  samples a subset of items $X\subseteq U_j$ of size $\frac{\abs{U_j}}{\ell_j}\cdot \log \left(\frac{1}{\delta}\right)$. The swap algorithm then returns a random item from $A\cap U_j$ and the item $i^*\in X$ which maximizes $\nabla_{i^*} F(t\cdot \indicator_A)$. As we show, this (nearly) suffices for an approximate swap algorithm.

The approach depicted above faces difficulties in case the item $i^*$ is already in $A$.
We overcome these difficulties by running \gp on a reduced instance which has multiple copies of each item. 
We then sample $X$ from a set of copies of items in $U_j$ which do not intersect with $A$. Conceptually, this is similar to maintaining a multiset $A$, however, this requires to carefully define the 
multilinear extension as in Observations~\ref{obv:G_by_F}.

\paragraph{The Reduced Instance}
Let $I=(f,U, \{U_j\}_{j\in [k]}, \{\ell_j\}_{j\in [k]})$  be a \partFprob\ instance, and let $r$ be the rank of the matroid of the instance. 
The {\em reduced} instance of $I$ is the \partFprob\ instance $I'=(g,\tU, \{\tU_j\}_{j\in [k]}, \{\ell_j\}_{j\in [k]})$ where $\tU= U\times [ r+1]$, $\tU_j=U_j\times [r+1]$, and $g:2^{\tU}\rightarrow \mathbb{R}_{\geq 0}$ is defined by 
$$
g(Q) = f\left(\{i\in U\,\middle|\, (\{i\}\times [r+1])\cap Q\neq \emptyset\}\right)
$$
for every $Q\subseteq \tU$.
For every $i\in U$, 
we refer to the items $(i,1),(i,2),\ldots, (i,r+1)\in \tU$ as the {\em copies} of $i$. 
In particular, $g(Q)=f(S)$, where $S\subseteq U$ is the set of items which have a copy (or more) in $Q$.  
It can be easily verified that $g$ is also monotone non-negative submodular function.

Let $F$ and $G$ be the multilinear extensions of $f$ and $g$ respectively.
The next observation shows that 
 $G(y)$ can be evaluated using a single oracle query for $F$. 
\begin{observation}
    \label{obv:G_by_F}
Let $y\in [0,1]^{\tU}$, and define $x\in [0,1]^U$ by $x_i= 1- \prod_{s=1}^{r+1}(1-y_{(i,s)})$ for every $i\in U$. Then,
$G(y) = F(x)$.
\end{observation}

Let $O=\{o_1,\ldots, o_r\}$ be an optimal solution for the original instance $I$. It can be easily observed if we pick a copy of each item in $O$ into $O'$ we can obtain an optimal solution for $I'$. That is, for every set of indices $i_1,\ldots,i_r\in [r+1]$ it holds that $O'=\{ (o_1,i_1),\ldots, (o_r,i_r)\}$ is an optimal solution for $I'$. Furthermore, the optimal solution values for both $I$ and $I'$ are the same.  We use $n=\abs{U}$ to denote the universe size of the original instance $I$.

\begin{algorithm}[t]
\caption{\genF $(t, A)$}
\SetKwInOut{Input}{input}
\SetKwInOut{Output}{output}
\SetAlgoNlRelativeSize{0}
\label{alg:paritionF}
\SetKwInOut{Configuration}{Config}
\Configuration{Accuracy parameter $\delta>0$ }

      For every $i\in U$ select a copy $(i,s) \in (\{i\}\times [r+1])\setminus A$ and add $(i,s)$ to a set $C$. \label{genF:candidates}

      Select a random part $j\in [k]$ such that $\Pr(j=q)\,=\, \frac{\ell_q}{r}$ for all $q\in [k]$. 

      Sample a subset $X\subseteq C\cap \tU_j$ of size $\min\left\{\frac{\abs{U_j}}{\ell_j}\cdot \ln \left(\frac{1}{\delta}\right) , \abs{U_j} \right\}$ uniformly at random. \label{genF:X}

      For every item $(i,s)\in X$ compute $w_{(i,s)} = \nabla_{(i,s)} G(t\cdot \indicator_{A}) = G(t\cdot \indicator_{A} + \indicator_{\{(i,s)\}}) -G(t\cdot \indicator_{A})$.\label{genF:compute}

       $(i^*,s^*)\leftarrow \argmax_{(i,s)\in X} w_{(i,s)} $.\label{genF:argmax}

       Sample an item $(a_i,a_s)$ from $A\cap \tU_j$ uniformly at random.

       Return  $(a_i,a_s)$ and $(i^*,s^*)$.

\end{algorithm}

\paragraph{The Swap Algorithm} 
In order to attain an approximate solution for $I$ we run  \gp on the reduced instance $I'$.
The pseudocode of our swap algorithm for $I'$ is given 
in Algorithm~\ref{alg:paritionF}. 
The algorithm utilizes an accuracy parameter $\delta>0$ that controls the trade-off between the quality of the algorithm and its running time. 
The algorithm first arbitrary selects a copy of each item in $U$ which is not in $A$.
Let $C\subseteq \tU$ be the set of selected copies. 
Subsequently, the algorithm picks a part $j\in[k]$ at random, and then randomly selets set $X\subseteq C\cap \tU_j$ of items from that part of size $\frac{\abs{U_j}}{\ell_j}\cdot \ln \left(\frac{1}{\delta}\right)$. 
Finally, the algorithm picks a random item from $(a_i,a_s)\in A\cap \tU_j$, and returns $(a_i,a_s)$ together with the item $(i^*,s^*)$ in $X$ which maximizes $\nabla_{(i^*,s^*)} G(t\cdot \indicator_A)$.

We first show  that the algorithm is indeed a swap procedure, and upper bound its running time.  We then extend this result to prove that it is an approximate swap procedure.
\begin{lemma}
\label{lem:genF_basics}
Algorithm~\ref{alg:paritionF} is a swap algorithm which uses in expectation $O\left(\frac{n}{r}\cdot \ln\left(\frac{1}{\delta}\right)\right)$ oracle queries for~$F$.
\end{lemma}
\begin{proof}
The algorithm returns two items $(a_i,a_s)$ and $(i^*,s^*)$ which belong to the same part $\tU_j$. Also, by the selection of $C$ in Line~\ref{genF:candidates} it also holds that   $(i^*,s^*)\notin A$. Therefore $A-(a_i,s_i)+ (i^*,s^*)$ is also a basis of the matroid of the instance $I'$. 

Conditioned on $j=q$ for some $q\in[k]$, the number of queries to $F$ initiated by the algorithm is $O\left( \frac{\abs{U_q}}{\ell_q}\cdot \ln \left(\frac{1}{\delta}\right)\right)$. This is true since the value of $G(x)$ can be evaluated using a single query to $F$ (Observation~\ref{obv:G_by_F}), and due to the selected size of the set $X$ in Line~\ref{genF:X}. Thus, the number of oracle queries initiated by the algorithm in expectation is
$
O\left(\E\left[ \frac{\abs{U_j}}{\ell_j}\right] \cdot  \ln \left(\frac{1}{\delta}\right)\right)$. Additionally,
$$
\E\left[ \frac{\abs{U_j}}{\ell_j}\right] \,=\,\sum_{q=1}^{k} \Pr[j=q] \cdot \frac{\abs{U_q}}{\ell_q} \,=\,\sum_{q=1}^{k} \frac{\ell_q}{r} \cdot  \frac{\abs{U_q}}{\ell_q} \,=\,\frac{n}{r}.  
$$
Therefore, the number of oracle queries the algorithm uses, in expectation, is $O\left(\frac{n}{r}\cdot \ln \left(\frac{1}{\delta}\right)\right)$. 
\end{proof}

Next, we show that Algorithm~\ref{alg:paritionF} is a $(1-\delta,1)$-approximate swap algorithm. To do so, we utilize the following auxiliary lemma, which was implicitly proven in \cite[Lemma 2]{MBKV15} (and also in \cite{BFS17}).
\begin{lemma}
\label{lem:sampling}
Let $E$ be a set of $n$ elements, where each elements has a non-negative weight $w_e\geq 0$. Additionally, let $M\subseteq E$ be a subset of $m\in \mathbb{N}$ elements form $E$ and $\delta>0$. Finally, let $X\subseteq E$ be a uniformly random set of $E$ of size $\min\left\{\ceil{\frac{n}{m}\cdot \ln\left(\frac{1}{\delta}\right)},\,n \right\}$. Then,
$$
\E\left[\max_{e\in X} w_e\right] \geq (1-\delta)\cdot \frac{\sum_{e\in M} w_e}{m}. 
$$
\end{lemma}

\begin{lemma}
Algorithm~\ref{alg:paritionF} is a $(1-\delta,0)$-approximate swap algorithm.
\end{lemma}
\begin{proof}
Consider an execution of Algorithm~\ref{alg:paritionF}.
Let  $O'\subseteq C$ be an optimal solution for $I'$ which is fully contained in the set $C$ computed in Line~\ref{genF:candidates}. Also, for every $(i,s)\in C$ define $w_{(i,s)}= \nabla_{(i,s)} G(t\cdot \indicator_A)$, and note that the notation is consistent with Line~\ref{genF:compute}.

For every part  $q\in [k]$  let $O'_q= O'\cap \tU_q$ be the items of $O'$ from part $q$. Since $O'$ is a basis it follows that $\abs{O'_q} = \ell_q$. 
We first condition on $j=q$. 
By Lemma~\ref{lem:sampling} we have
\begin{equation*}
\begin{aligned}
\E\left[\nabla_{(i^*,s^*)} G(t\cdot \indicator_A)\,|\,j=q\right] \,&=\, \E\left[\max_{(i,s)\in X} w_{(i,s)}\,|\,j=q\right] \\
&\geq \, (1-\delta)\cdot \frac{\sum_{o\in O'_q} w_{o}}{\ell_q}\,\\
&=\, (1-\delta)\cdot \frac{\sum_{o\in O'_q} \nabla_{o} G(t\cdot \indicator_A)}{\ell_q}.
\end{aligned}
\end{equation*}
Therefore,
\begin{equation}
\label{eq:gen_samp_gain}
\begin{aligned}
\E\left[\nabla_{(i^*,s^*)} G(t\cdot \indicator_A)\right]  \,&=\, \sum_{q=1}^{k}  \Pr[j=q ]\cdot  \E\left[\nabla_{(i^*,s^*)} G(t\cdot \indicator_A)\,|\,j=q\right]\,\\
&\geq\, \sum_{q=1}^{k}  \frac{\ell_q}{r}\cdot (1-\delta)\cdot \frac{\sum_{o\in O'_q} \nabla_{o} G(t\cdot \indicator_A)}{\ell_q}\\
&\geq \frac{1-\delta}{r} \sum_{o\in O'}  \nabla_{o} G(t\cdot \indicator_A)
\end{aligned}
\end{equation}

Similarly,
\begin{equation}
\label{eq:gen_drop}
\begin{aligned}
\E\left[\nabla_{(a_i,a_s) } G(t\cdot \indicator_{A})\right] \,&=\,
\sum_{q=1}^{k} \Pr[ j=q] \cdot \E\left[\nabla_{(a_i,a_s) } G(t\cdot \indicator_{A}) \,|\,j=q \right] \,\\
&=\, \sum_{q=1}^{k} \frac{\ell_j}{r}\cdot\frac{1}{\ell_j} \cdot \sum_{(i,s)\in A\cap \tU_q} \nabla_{(i,s) } G(t\cdot \indicator_{A})\\
&=\, \frac{1}{r}\sum_{(i,s)\in A}\nabla_{(i,s) } G(t\cdot \indicator_{A}).
\end{aligned}
\end{equation}

By \eqref{eq:gen_samp_gain} and \eqref{eq:gen_drop} we have
$$
\E\left[\nabla_{(i^*,s^*)} G(t\cdot \indicator_A)-\nabla_{(a_i,a_s) } G(t\cdot \indicator_A)\right] \, \geq \, \frac{1-\delta}{r} \sum_{o\in O'}  \nabla_{o} G(t\cdot \indicator_A) - \frac{1}{r}\sum_{(i,s)\in A}\nabla_{(i,s) } G(t\cdot \indicator_{A}).
$$
Therefore, Algorithm~\ref{alg:paritionF} is a $(1-\delta,0)$-approximate swap algorithm.

\end{proof}

We are now ready to prove Theorem~\ref{thrm:GeneralizedPartitionF}.
\begin{proof}[Proof of Theorem~\ref{thrm:GeneralizedPartitionF}]
The algorithm works as follows.
\begin{enumerate}
\item Execute \gp\ on the reduced instance $I'$, \genF\ (Algorithm~\ref{alg:paritionF}) with $\delta =\eps/2$ as the swap algorithm, and error parameter $\eps/2$ (for \gp). Let $Q$ be the returned solution.
\item Return $S=\{i\in U\,|\, \exists s\in [r+1]:\, (i,s)\in Q\}$, the set of all items for which there is at least one copy in $Q$. 
\end{enumerate}
It can be easily observed that the algorithm returns a solution for the original \partFprob\ instance $I$. The execution of \gp\ uses $O\left( r\cdot \ln\left( \frac{2}{\eps}\right)\right)$  swap operations in expectation, and by Lemma~\ref{lem:genF_basics} each call for \genF\ requires $O\left(\frac{n}{r}\cdot \ln\left(\frac{2}{\eps}\right)\right)$ oracle queries to $F$. Thus, the overall number of oracle queries in expectation is $O\left( n\cdot \log^2\left( \frac{1}{\eps}\right)\right)$.

By Theorem~\ref{thm:approximate_poission} it holds that $$
\E[g(Q)]\geq \left(1-\frac{\eps}{2}\right)\cdot \left(1-e^{-\left(1-\frac{\eps}{2}\right)}\right)\cdot g(O')\,\geq\, (1-\eps)\cdot (1-e^{-1})\cdot f(O),$$ where $O'$  is an optimal solution for the reduced instance $I'$, and $O$ is an optimal solution for $I$. The second inequality holds $1-e^{-\left(1-\frac{\eps}{2}\right)} \geq \left(1-\frac{\eps}{2}\right)\cdot (1-e^{-1})$ and $f(O)=g(O')$. 
By the definition of $g$ we have $f(S)=g(Q)$. Therefore,
$$
\E[f(S)] \,=\,E[g(Q)]\,\geq  (1-\eps)\cdot (1-e^{-1})\cdot f(O),
$$
which completes the proof.

\end{proof}

\subsubsection{Oracle for $f$.}

Finally, we are left to prove Theorem~\ref{thrm:GeneralizedPartition} which gives our 
application for  \gp  with \partfprob.
The result integrates concepts from both Section~\ref{sec:partition_orc_to_f}  and Section~\ref{sec:generalizeF}.
Our algorithm for \partfprob runs \gp on the reduced instance (as defined in Section~\ref{sec:generalizeF}), and uses a variant of Algorithm~\ref{alg:paritionF} which replaces Lines~\ref{genF:compute} and~\ref{genF:argmax} with a call for a \mab algorithm (Definition~\ref{def:mab} and Algorithm~\ref{thrm:mab}). This modification is required as we are only given an oracle for $f$, hence we can only estimate the multilinear extension via sampling. The error incurred by the \mab algorithm depends on the marginals to opt ratio, as  defined in \eqref{eq:ratio}. Finally, we reduce this ratio via a preprocessing algorithm.

We follow the same notations for instances and reduced instance as in Section~\ref{sec:generalizeF}.  Let $I=(f,U,\{U_j\}_{j\in [k]}, \{\ell_j\}_{j\in [k]})$ be a \partfprob instance, and let $\cI=\{S\subseteq U\,|\,\forall j\in [k]:~\abs{S\cap U_j} \leq \ell_j\}$ be the independent sets of the instance's matroid. We use $r=\sum_{j=1}^{k}\ell_j$ to denote the rank of the matroid $(U,\cI)$, and $n=\abs{U}$ to denote the size of the universe. We define the marginals to opt ratio of $I$  by 
$$\marsum(I)\,=\,\frac{\max_{S\in \cI} \sum_{i\in S} (f(\{i\})-f(\emptyset))}{f(O)},
$$
where $O$ is an optimal solution for the instance. Note that this definition is consistent with \eqref{eq:ratio} for regular partition matroid. 

We design a swap algorithm for the reduced instance $I'$ of $I$ whose quality depends on $\marsum(I)$.
\begin{lemma}
\label{lem:swap_partfprob}
There is $(1-\delta,\eta)$-approximate swap algorithm for the reduced instance $I'$ of a \partfprob\ instance $I$, where $\eta = \delta\cdot \marsum(I)$,  that uses $O\left(\frac{n}{r}\cdot \frac{1}{\delta^2}\cdot \ln^4\left(\frac{1}{\delta}\right)\right)$ oracle queries to $f$ in expectation.
\end{lemma}
The proof of Lemma~\ref{lem:swap_partfprob} is deferred to the end of this section. 
Theorem~\ref{thm:approximate_poission} can be used together with Lemma~\ref{lem:swap_partfprob} to attain an approximate solution for the reduced instance $I'$ of $I$.  By converting the solution back to a solution for $I$ we attain the following. 
\begin{corollary}
\label{cor:GenBanditSwap}
There exists and algorithm for monotone submodular maximization with a generalized partition matroid which uses $O\left( \frac{n}{\delta^2} \cdot \ln^5\left(\frac{1}{\delta}\right)\right)$ value queries for $f$ in expectation and returns a solution $S$ such that 
$$
\E[f(S)]\,\geq\,\left(1-\delta\right) \cdot \left(1-\frac{\delta \cdot \marsum(I)}{1-\delta}\right)\cdot (1-\nicefrac{1}{e}) \cdot f(O),
$$
where $I$ is the instance and $O$ is an optimal solution. 
\end{corollary}
We use Corollary~\ref{cor:GenBanditSwap} on a {\em residual} instance  $I_P$ which is obtained by the pre-processing algorithm. The pre-processing algorithm gurantees, with high probability, that value $\delta\cdot \marsum(I_P)$ in Corollary~\ref{cor:GenBanditSwap} does not significantly harms the approximation ratio.

Given a $\partfprob$ instance $I=(f,U, (U_j)_{j\in [k]}, (\ell_j)_{j\in [k]})$ and an independent set $P\in \cI$ we define the {\em residual instance of $I$ and $P$}  as $I_P=(f_P, U_P, (U_{P,j})_{j\in J_P}, (\ell_{P,j})_{j\in J_P})$  where $J_P=\{j\in [k]\,|\, \abs{U_j\cap P}<\ell_j\}$, 
$U_{P,j}=U_j\setminus P$ for all $j\in J_P$, $\ell_{P,j} = \ell_{j}-\abs{P\cap U_j}$ for all $j\in [J_P]$,
$U_P=\bigcup_{j\in J_P} U_{P,j}$,  and $f_P:2^{U_P} \rightarrow \mathbb{R}_{\geq 0}$ is defined by $f_P(S)= f(P\cup S) -f(P)$. 
As in the case of residual instance with a regular partition matroid constraint, the residual instance represents the subproblem that remains when $P$ is enforced as part of the solution, and the definitions coincide if $\ell_j=1$ for all $j\in [k]$. 
It can be easily observed that if $S$ is a solution for $I_P$, then $S\cup P$ is a solution for $I$ with value $f(S\cup P)=f(S)+f_P(S)$.

The pre-processing algorithm for generalized partition matroid provides the same guarantees as the pre-processing for partition matroid (Lemma~\ref{lem:simpleprep}).
\begin{lemma}[pre-processing]\label{lem:generalized_prep}
There exists a randomized algorithm, which given a \partfprob\ instance $I=(f,U, (U_j)_{j\in[k]}, (\ell_j)_{j\in [k]})$ and error parameter $0<\delta<\nicefrac{1}{2}$, performs in expectation
$ O\left(n\cdot \log(n)\cdot \log(1/\delta)\right)$ value queries and returns $ P\in \cI$ satisfying:
\begin{enumerate}
    \item $ \Pr\left[\marsum(I_P) \leq c\cdot \frac{f(O)}{f(O_P)}\right]\geq 1-\delta$,  
    \item $\mathbb{E}\left[(1-\nicefrac{1}{e})\cdot f(O_P\cup P) + \frac{1}{e}\cdot f(P) \right]\geq (1-\nicefrac{1}{e}-\delta)\cdot f(O)$. 
\end{enumerate}
Here, $ O$ is an optimal solution for $I$, and $O_P$ is an optimal solution for the residual instance $I_P$ of $I$ and $P$. Furthermore, $c>1$ is an absolute constant. 
\end{lemma}
 We give the proof of Lemma~\ref{lem:generalized_prep} in Appendix~\ref{sec:preprocessing}. Here, we give the proof of Theorem~\ref{thrm:GeneralizedPartition} that is essentially identical to the proof of Theorem~\ref{thrm:partition}, with the only difference is that Corollary~\ref{cor:GenBanditSwap} and Lemma~\ref{lem:generalized_prep} are used instead of Corollary~\ref{lem:BanditSwapPoisson} and Lemma~\ref{lem:simpleprep}.  We give it here for completeness.

\begin{proof}[Proof of Theorem~\ref{thrm:GeneralizedPartition}]

The algorithm we define is the following:
\begin{enumerate}
 \item Execute the prep-processing algorithm  of Lemma~\ref{lem:generalized_prep} on the original instance $I$ with the error parameter $\delta_1=\eps/8$  to find $P$. 
\item Use \gp  together with the swap algorithm from Lemma~\ref{lem:swap_partfprob} on the reduced instance of $I_P$ (Corollary~\ref{cor:GenBanditSwap}), with the error parapmeter $\delta_2 =\nicefrac{\varepsilon}{(16c)} $
to obtain $A$. Here, $c$ is the absolute constant from Lemma~\ref{lem:generalized_prep}.
\item  Return $A\cup P$.
\end{enumerate}

First, by the definition of residual instances,  it is easy to note that $ A\cup P\in \cI$.  
Second, 
following Corollary \ref{cor:GenBanditSwap} and  Lemma \ref{lem:generalized_prep} the number of oracle queries used by the algorithm in expectation is at most $O(n\log^5{(1/\varepsilon)}/\varepsilon^2)$.

Third and finally, let us lower bound the expected value of the output: $ \mathbb{E}[f(A\cup P)]$.
For every $R\in \cI$ let $O_R$ be an optimal solution for $I_R$,  the residual instance of $I$ and $R$. 
Also, let $\cP$ be the set of all $P\in \cI$ such that $\marsum(I_P) \leq c\cdot \frac{f(O)}{f(O_P)}$. 
By Corollary~\ref{cor:GenBanditSwap} , for every $R\in \cP$ it holds that,
$$
\begin{aligned}
 \E[ f_P(A) \,|\, P=R] \,&\geq \left(1-\delta_2\right) \cdot \left(1-\frac{\delta_2\cdot \marsum(I_R)}{1-\delta_2}\right)\cdot (1-\nicefrac{1}{e}) \cdot f_R(O_R)\\
&\geq \, \left(1-\frac{\eps}{16\cdot c}\right) \cdot  \left(1-2\cdot \frac{\eps }{16\cdot c} \cdot \frac{c\cdot f(O)}{f_R(O_R)}\right)\cdot (1-\nicefrac{1}{e}) \cdot f_R(O_R)\\
&\geq \left(1-\frac{\eps}{8}\right)\cdot \left(1-\nicefrac{1}{e}\right)\cdot f_R(O_R)-\frac{\eps}{8}\cdot f(O). 
\end{aligned}
$$
The second inequality holds as $R\in \cP$ and $\delta_2= \frac{\eps}{16c}$, we further assume that $\delta_2\leq \frac{1}{2}$. Therefore,
\begin{equation}
\label{eq:generalized_exp}
\begin{aligned}
\E[f_P(A)]  \, &=\, \sum_{R\in \cI} \Pr[P=R]\cdot  \E[ f_P(A) \,|\, P=R]\\
&\geq \, \sum_{R\in \cP} \Pr[P=R]\cdot  \E[ f_P(A) \,|\, P=R] \\
&\geq\, \sum_{R\in \cP} \Pr[P=R]\cdot\left( \left(1-\frac{\eps}{8}\right)\cdot \left(1-\nicefrac{1}{e}\right)\cdot f_R(O_R)-\frac{\eps}{8}\cdot f(O)\right)\\
&=\, \left( 1-\frac{\eps}{8}\right) \cdot \E\left[ (1-\nicefrac{1}{e}) \cdot f_P(O_P) \right]  - \left(1-\frac{\eps}{8}\right) \left(1-\frac{1}{e}\right)\cdot \sum_{R\in \cI\setminus \cP} \Pr[P=R] \cdot f_R(O_R)-\frac{\eps}{8}\cdot f(O).
\end{aligned}
\end{equation}
Furthermore, 
\begin{equation}
\label{eq:generalized_bad}
\sum_{R\in \cI\setminus \cP} \Pr[P=R] \cdot f_R(O_R) \leq f(O) \cdot \Pr[P\notin \cP] \leq f(O)\cdot \delta_1 = \frac{\eps}{8}\cdot f(O), 
\end{equation}
where the last inequality follows from Lemma~\ref{lem:generalized_prep}. 
Plugging \eqref{eq:generalized_bad} into \eqref{eq:generalized_exp} we attain 
\begin{equation}
\label{eq:generalized_expA}
\E[f_P(A)] \geq  \left( 1-\frac{\eps}{8}\right) \cdot \E\left[ (1-\nicefrac{1}{e}) \cdot f_P(O_P) \right]  -   \frac{\eps}{ 8} \cdot f(O) -   \frac{\eps}{ 8} \cdot f(O).
\end{equation}

Using \eqref{eq:generalized_expA} and Lemma~\ref{lem:generalized_prep}
 we have
 $$
 \begin{aligned}
 \E[f(A\cup P)] \,&= \,\E[f(P)+f_P(A)]\\
 &=\,\E\left[ f(P) + \left(1-\frac{\eps}{8}\right) (1-\nicefrac{1}{e})\cdot f_P(O_P)\right] -\frac{\eps}{4}\cdot f(O)\\
 &\geq\left( 1-\frac{\eps}{8}\right)\cdot  \E\left[ \frac{1}{e}\cdot f(P) + \left(1-\frac{1}{e}\right)\cdot f(P\cup O_P) \right] -\frac{\eps}{4}\cdot f(O)\\
 &\geq \left( 1-\frac{\eps}{4}\right) \cdot \left(1-\frac{1}{e} -\delta_1\right) \cdot f(O) -\frac{\eps}{4}\cdot f(O) \\
 &\geq (1-\nicefrac{1}{e} -\eps)\cdot f(O),
 \end{aligned}
 $$
 which concludes the proof. 
\end{proof}

\subsubsection*{Proof of Lemma~\ref{lem:swap_partfprob}}

Our swap algorithm for \partfprob is given in Algorithm~\ref{alg:genf}.The algorithm first selects a copy of each item from $U$ into a set $C$ of candidates, such that the selected copy is not in $A$. It then selects a part $j$ at random, and samples a set $X$ of  $\approx \frac{\abs{U_j}}{\delta} \cdot \ln\left(\frac{1}{\delta}\right)$ items from $C\cap \tU_j$. Then, the \mab algorithm is used to find an item  $(i^*,s^*)\in X$ which (approximately) maximizes $\nabla_{(i^*,s^*)}G(t\cdot \indicator_A)$. Finally, the algorithm returns a random item from $A\cap \tU_j$ together with $(i^*,s^*)$. 

We begin by showing the algorithm is indeed a swap algorithm an analyze its running time. A later lemma would give its approximation guarantees.
\begin{lemma}
\label{lem:genf_is_swap}
Algorithm~\ref{alg:genf} is a swap algorithm for the reduced instance $I'$ of a \partfprob instance $I$, which uses $O\left(\frac{n}{r}\cdot \frac{1}{\delta^2}\cdot \ln^4\left(\frac{1}{\delta}\right)\right)$ value oracle queries for $f$. 
\end{lemma}
\begin{proof}
We first observe the algorithm returns two items from part $\tU_j$. Furthermore, as $(i^*,s^*)\in X\subseteq C$ and $C\cap A=\emptyset$, it follows that $(i^*,s^*)\notin A$ always hold. Therefore $A-(a_i,a_s)+(i^*,s^*)$ is a base of the matroid of the reduced instance $I'$. 

As for the running time,  considition on $X$, by Theorem~\ref{thrm:mab} the \mab\ algorithm requires $O\left( \abs{X}\cdot \frac{\log\left(\frac{1}{\delta'} \right)}{\delta'^2}\right)=  O\left( \abs{X}\cdot \frac{\log^3\left(\frac{1}{\delta} \right)}{\delta^2}\right)$ samples from the distribution, were each sample requires $O(1)$ oracle queries for $f$ (which can be easily used to attain value queries for $g$ by definition). Also, $$\E\left[X\right] = \sum_{s=1}^{k} \Pr[j=s] \cdot \frac{\abs{U_j}}{\ell_s} \cdot \ln\left(\frac{1}{\delta}\right)
= \sum_{s=1}^{k} \frac{\ell_s}{r}\ \cdot \frac{\abs{U_j}}{\ell_s} \cdot \ln\left(\frac{1}{\delta}\right)=\frac{n}{r}\cdot \ln\left(\frac{1}{\delta}\right).$$
Therefore, in expectation, the number of queries to $f$ use by Algorithm~\ref{alg:genf} is
$$
O\left(\frac{n}{r}\cdot \frac{\log^4\left(\frac{1}{\delta}\right)}{\delta^4}\right)
$$
\end{proof}

The next lemma deals with the quality of Algorithm~\ref{alg:genf} as a swap algorithm.
\begin{lemma}
\label{lem:genf_approximate}

Algorithm~\ref{alg:genf} is a $(1-\delta,\eta)$-approximate swap algorithm, where $\eta = \delta\cdot \marsum(I)$.
\end{lemma}
\begin{proof}
Consider an execution of Algorithm~\ref{alg:genf} and let $O'\subseteq C$ be an optimal solution for $I'$ which is fully contained in the set $C$ computed in Line~\ref{genf:candidates}.
Also, for every $(i,s)\in C$ define $w_{(i,s)}= \nabla_{(i,s)} G(t\cdot \indicator_A)$. Also,  let  $\cD'_{(i,s)}$ be the distribution  of $g(R\cup \{i\}) -g(R\setminus \{i\})$. It can be easily verified that the mean of $\cD'_{(i,s)}$ is $w_{i,s}$. 
Let $sqin [k]$ and $Y\subseteq C\cap \tU_q$, by \Cref{thrm:mab}
conditioned on $j=q$ and $X=Y$ we can lower bound the expected value of $w_(i^*,s^*)$:
$$
\E[w_{(i^*,s^*)}\,|\,j=q, X=Y] \geq \max_{(i,s)\in Y} w_{(i,s)} -\delta'\cdot \alpha_Y,
$$
where $\alpha_Y= \max_{(i,s)\in Y} \left(f(\{i\}) -f(\emptyset)\right)$. 

Therefore, 
\begin{equation}
\label{eq:genf_basics}
\E[w_{(i^*,s^*)}] \geq \E\left[\max_{(i,s)\in X} w_{(i,s)} \right] -\frac{\delta}{\ln\left(\frac{1}{\delta}\right)} \cdot \E[\alpha_X]. 
\end{equation}
We bound each of the terms in the last expression separately
\begin{claim}
\label{claim:EMAX}
$\E\left[\max_{(i,s)\in X} w_{(i,s)} \right]\geq \frac{1-\delta}{r} \sum_{o\in O'} \nabla_o G(t\cdot \indicator_A).$
\end{claim}
\begin{claimproof}
For every part  $q\in [k]$  let $O'_q= O'\cap \tU_q$ be the items of $O'$ from part $q$. Since $O'$ is a basis it follows that $\abs{O'_q} = \ell_q$. 
We first condition on $j=q$. 
By Lemma~\ref{lem:sampling} we have
\begin{equation*}
\begin{aligned}
 \E\left[\max_{(i,s)\in X} w_{(i,s)}\,|\,j=q\right] 
&\geq \, (1-\delta)\cdot \frac{\sum_{o\in O'_q} w_{o}}{\ell_q}\,
&=\, (1-\delta)\cdot \frac{\sum_{o\in O'_q} \nabla_{o} G(t\cdot \indicator_A)}{\ell_q}.
\end{aligned}
\end{equation*}
Therefore,
\begin{equation*}
\begin{aligned}
\E\left[\max_{(i,s)\in X} w_{(i,s)}\right]  \,&=\, \sum_{q=1}^{k}  \Pr[j=q ]\cdot  \E\left[\max_{(i,s)\in X} w_{(i,s)}\,|\,j=q\right]\,\\
&\geq\, \sum_{q=1}^{k}  \frac{\ell_q}{r}\cdot (1-\delta)\cdot \frac{\sum_{o\in O'_q} \nabla_{o} G(t\cdot \indicator_A)}{\ell_q}\\
&\geq \frac{1-\delta}{r} \sum_{o\in O'}  \nabla_{o} G(t\cdot \indicator_A)
\end{aligned}
\end{equation*}
\end{claimproof}

\begin{claim}
\label{claim:alpha}
$\E[\alpha_X]\leq \frac{\ln\left(\frac{1}{\delta}\right)}{r}\cdot \marsum(I)$
\end{claim}
\begin{claimproof}
Let $q\in [k]$. We first condition on $q=j$. Assume $\abs{\tU_q} = \ell$, and assume w.l.o.g that $U_q=\{1,\ldots, \ell\}$. Furthermore, for ever y$i\in \tU_q$ define $M_i = f(\{i\})-f(\emptyset)\geq 0$, and assume that  $M_1\geq M_2\geq \ldots \geq M_{\ell}$.  Now, let $Z\subseteq U_q$ be a random subset of $U_q$ of size $\frac{\abs{U_j}}{\ell_j}\cdot \ln\left(\frac{1}{\delta}\right)$. It can be easily observed  that the maximum value of $M_i$ for $i\in Z$ has the same distibution of $\alpha_X$ condition on $j=q$. Therefore,
\begin{equation}
\label{eq:alpha_cond}
\E[\alpha_X\,|\,j=q] = \E\left[ \max_{i\in Z} M_i\right].
\end{equation}

Now,
$$
\max_{i\in Z} M_i \leq M_{\ell_q} + \max_{i\in Z\cap[\ell_q] } (M_i-M_{\ell_q})
\leq M_{\ell_q} + \sum_{i\in Z\cap \ell_q} (M_i-M_{\ell_q}),
$$
where the first inequality holds as the items are sorted and the last inequality  replaces maximization by summation.

Therefore,
\begin{equation}
\label{eq:EmaxM}
\begin{aligned}
\E[\max_{i\in Z}M_i] &\leq M_{\ell_q} +\sum_{i\in [\ell_q]} \Pr(i\in Z) (M_i-M_{\ell_q}) \\
&=  M_{\ell_j} +\sum_{i\in [\ell_q]} \frac{1}{\ell_q} \cdot \ln\left(\frac{1}{\delta}\right) (M_i-M_{\ell_q}) \\
&\leq \frac{\ln\left(\frac{1}{\delta}\right)}{\ell_q} \sum_{i\in[\ell_j]} M_i\\
&= \frac{\ln\left(\frac{1}{\delta}\right)}{\ell_j} \max_{W\subseteq U_q:~|W|\leq \ell_q} \sum_{i\in U_q} \left(f(\{i\})-f(\emptyset)\right). 
\end{aligned}
\end{equation}

By \eqref{eq:alpha_cond} and \eqref{eq:EmaxM} we hae
$$
\begin{aligned}
\E[\alpha_X] &=\sum_{q=1}^{k} \Pr(j=q) \E[\alpha_X\,|\,j=q] \,\\
&\leq\sum_{q=1}^{k} \frac{\ell_j}{r} \cdot  \frac{\ln\left(\frac{1}{\delta}\right)}{\ell_j} \max_{W\subseteq U_j:~|W|\leq \ell_j} \sum_{i\in U_j} \left(f(\{i\})-f(\emptyset)\right)\\
&\leq \frac{1}{r}\cdot \ln\left(\frac{1}{\delta}\right) \marsum(I). 
\end{aligned}
$$
\end{claimproof}

By \eqref{eq:genf_basics} and Claims~\ref{claim:alpha} and \ref{claim:EMAX} we have
\begin{equation}
\label{eq:swapin}
\begin{aligned}
\E\left[\nabla_{(i^*,s^*)} G(t\cdot \indicator_A)\right]& =
\E[w_{(i^*,s^*)}] \\
&\geq \frac{1-\delta}{r}\sum_{o\in O'} \nabla_o G(t\cdot \indicator_{A}) -\frac{\delta}{\ln\left(\frac{1}{\delta}\right)} \cdot \frac{\ln\left(\frac{1}{\delta}\right)}{r} \cdot \marsum(I)
\\
&=  \frac{1-\delta}{r}\sum_{o\in O'} \nabla_o G(t\cdot \indicator_{A}) -\frac{\eta}{r}
\end{aligned}
\end{equation}

Additionally,
\begin{equation}
\label{eq:genf_drop}
\begin{aligned}
\E\left[\nabla_{(a_i,a_s) } G(t\cdot \indicator_{A})\right] \,&=\,
\sum_{q=1}^{k} \Pr[ j=q] \cdot \E\left[\nabla_{(a_i,a_s) } G(t\cdot \indicator_{A}) \,|\,j=q \right] \,\\
&=\, \sum_{q=1}^{k} \frac{\ell_j}{r}\cdot\frac{1}{\ell_j} \cdot \sum_{(i,s)\in A\cap \tU_q} \nabla_{(i,s) } G(t\cdot \indicator_{A})\\
&=\, \frac{1}{r}\sum_{(i,s)\in A}\nabla_{(i,s) } G(t\cdot \indicator_{A}).
\end{aligned}
\end{equation}

By \eqref{eq:swapin} and \eqref{eq:genf_drop} we have
$$
\begin{aligned}
&\E\left[\nabla_{(i^*,s^*)} G(t\cdot \indicator_A)-\nabla_{(a_i,a_s) } G(t\cdot \indicator_A)\right] \, \\
&\geq \frac{1-\delta}{r}\sum_{o\in O'} \nabla_o G(t\cdot \indicator_{A}) -\frac{\eta}{r} - \frac{1}{r}\sum_{(i,s)\in A}\nabla_{(i,s) } G(t\cdot \indicator_{A}).
\end{aligned}
$$
Therefore, Algorithm~\ref{alg:paritionF} is a $(1-\delta,\eta)$-approximate swap algorithm. 

\end{proof}

Lemma~\ref{lem:swap_partfprob} is an immediate consequence of Lemmas~\ref{lem:genf_is_swap} and~\ref{lem:genf_approximate}. 

\begin{algorithm}[t]
\caption{\genf $(t, A)$}
\SetKwInOut{Input}{input}
\SetKwInOut{Output}{output}
\SetAlgoNlRelativeSize{0}
\label{alg:genf}
\SetKwInOut{Configuration}{Config}
\Configuration{Accuracy parameter $\delta>0$ }

      For every $i\in U$ select a copy $(i,s) \in (\{i\}\times [r+1])\setminus A$ and add $(i,s)$ to a set $C$. \label{genf:candidates}

      Select a random part $j\in [k]$ such that $\Pr(j=q)\,=\, \frac{\ell_q}{r}$ for all $q\in [k]$.

      Sample a subset $X\subseteq C\cap \tU_j$ of size $\min\left\{\frac{\abs{U_j}}{\ell_j}\cdot \ln \left(\frac{1}{\delta}\right) , \abs{U_j} \right\}$ uniformly at random. \label{genf:X}

      Let $\cD_{(i,s)}$ be the distribution of $\frac{g(R\cup \{i\}) -g(R\setminus \{i\})}{\alpha_X}$ where $\alpha_X=\max_{(i,s)\in X} (f(\{i\})-f(\emptyset))$.
      \label{genf:dist}

      Let $(i^*,s^*)$ be the output of the algorithm of Theorem~\ref{thrm:mab} on $\left(\cD_{(i,s)}\right)_{(i,s)\in X}$ and error parameter $\delta'=\frac{\delta}{\ln\left(\frac{1}{\delta}\right)}$.\label{genf:mab}

       Sample an item $(a_i,a_s)$ from $A\cap \tU_j$ uniformly at random.

       Return  $(a_i,a_s)$ and $(i^*,s^*)$.

\end{algorithm}


\section{Frank-Wolfe Interpretation of Continuous Greedy}\label{app:FW}
An alternative intuitive way to understand the origin of \gp is by relating the continuous greedy algorithm to a Frank-Wolfe style continuous local search algorithm.

Consider the following continuous process defined for an arbitrarily small $ \varepsilon >0$: $(1)$ $\by(\varepsilon)$ is initialized to be an arbitrary base; and
$(2)$ $\frac{\partial \by(t)}{\partial t}\leftarrow (\indicator _{Z(t)}-\by(t))/t$ where: 
$$ Z(t)\triangleq \argmax \left\{ \sum _{i\in Z}\nabla _i F(t\by(t))\colon Z\in \cB\right\},$$
for every $ \varepsilon \leq t\leq 1$.
Note that $\by(t)\in \cP _{\cB}$ for all $\varepsilon \leq t\leq 1$.
Intuitively, the above algorithm can be seen as a Frank-Wolfe style continuous local search algorithm.
The reason is that up to scaling, it takes a convex combination of the current position of the algorithm, i.e., $\by(t)$, and another feasible point in the polytope, i.e., $\indicator_{Z(t)}$. 
The only difference between $ \{ t\by(t)\}_{\varepsilon \leq t \leq 1}$ and Frank-Wolfe is that $Z(t)$ is defined with respect to  $\nabla F(t\by (t))$ and not $\nabla F(\by(t))$, as is typically defined in Frank-Wolfe.

We observe that both continuous trajectories, $\{ \bx(t)\} _{0\leq t\leq 1}$ for continuous greedy and $ \{ t\by(t)\}_{\varepsilon \leq t \leq 1}$ of Frank-Wolfe, are in fact the same.
The obvious reason is that one can define $ \by(t)\triangleq \bx(t)/t$ and observe that:
\begin{align}
   \frac{\partial \by (t)}{\partial t}  = \frac{\partial \left( \bx(t)/t\right)}{\partial t} = \frac{1}{t}\frac{\partial \bx(t)}{\partial t} - \frac{\bx(t)}{t^2} = \frac{\indicator _{Z(t)} - \by(t)}{t}.\label{our-approach3}
\end{align}
Thus, we can conclude that up to time scaling, continuous greedy ($\bx(t)$) and Frank-Wolfe ($\by(t)$) are, in fact, the same.

Hence, as before, \gp can be obtained by rounding on the fly of the discretized trajectory $\{ \by(t)\}_{\varepsilon \leq t \leq 1}$ as the step size $\delta$ approaches $0$.
It should be noted that here $\by(t)$ is rounded on the fly and no scaling is required since $ \by(t)\in \cP_{\cB}$.
This establishes an alternative way to understand the origin of \gp.

\section*{Acknowledgments}

Mohit Singh is supported by NSF AF:2504994, 2106444.
Amit Ganz Rozenman, Ariel Kulik and Roy Schwartz receive funding from the European Union’s Horizon 2020 research and innovation program under
grant agreement no. 852870-ERC-SUBMODULAR.
The authors also thank Jan Vondr\'{a}k for discussions on the continuous interpretation of the algorithm.


\bibliographystyle{alpha}
\bibliography{refs}

\appendix

\section{Simulating the Poisson Process}\label{app:PoissonSimulate}
We note that for every time $ \varepsilon \leq t$, the density of the time of the next event after time $t$, i.e., $ \tau(t)$, is known.
This is summarized in the following observation.

\begin{observation}\label{obs:density}
For every $ t \geq \varepsilon$ the density $ f_{\tau(t)}$ of $\tau(t)$, the time of the next event after $t$ equals:
\[ f_{\tau(t)}(r) = \begin{cases} \frac{k\cdot t^k}{r^{k+1}} & r\geq t \\
 0 & r < t\end{cases}.\]
\end{observation}

One can directly simulate the Poisson process as follows.
Set the time of the zero event: $ \tau_0=\varepsilon$.
Then, for every $ i\geq 1$ in an increasing order of $i$: directly sample $ \tau_i$ by using the density of $ \tau(\tau_{i-1})$ as given by Observation \ref{obs:density}.

\section{Proof from Section~\ref{sec:poisson}}

\submodularstandard*
\label{app:submodular-standard}
\begin{proof}
\begin{align}
\sum_{i \in T} \nabla_i F(t \indicator_S) & = \sum _{i\in T\setminus S}\left( F(t\indicator _S \vee \indicator _{\{ i\} }) - F(t\indicator _S)\right) + \frac{1}{1-t}\sum _{i\in T\cap S }\left( F(t\indicator _S \vee \indicator _{\{ i\} }) - F(t\indicator _S) \right) \label{eq16b} \\
& \geq \sum _{i\in T\setminus S}\left( F(t\indicator _S \vee \indicator _{\{ i\} }) - F(t\indicator _S)\right) + \sum _{i\in T\cap S }\left( F(t\indicator _S \vee \indicator _{\{ i\} }) - F(t\indicator _S)\right)\label{eq16c} \\
& \geq F(t\indicator _S \vee \indicator _{T})  - F(t\indicator _S) \label{eq17-submodularity} \\
& \geq f(T) - F(t\indicator _S). \label{eq18-monotonicity}    
\end{align}
We note that equality \eqref{eq16b} follows from the multilinearity of $F$, whereas inequality \eqref{eq16c} holds due to the monotonicity of $f$ and property \eqref{multilinear-monotone} of the multilinear extension
implying that  $F(t\indicator _S \vee \indicator _{\{ i\} }) \geq F(t\indicator _S) $.
Moreover, inequality \eqref{eq17-submodularity} follows from the submodularity of $f$ and property \eqref{multilinear-submodular} of the multilinear extension:
\[ \sum _{i\in T} \left( F(t\indicator _S \vee \indicator _{\{ i\} })-F(t\indicator _S)\right) \geq F(t\indicator _S \vee \indicator _T).\]
Finally, inequality \eqref{eq18-monotonicity} holds due to the monotonicity of $f$ and property \eqref{multilinear-monotone} of the multilinear extension.
\end{proof}

\label{app:differentialEQ}

\differntialEq*
\begin{proof}
First we prove a slightly weaker version of Claim \ref{claim:differentialEQ}, where the inequality in requirement $(5)$ is a strong inequality.
This is summarized in the following claim.
\begin{claim}\label{claim:differentialEQ-weak}
Let $ q\colon [\varepsilon,1]\rightarrow \mathbb{R}$ be a function satisfying:
$(1)$ $q(t)$ is continuous for every $t\in (\varepsilon,1)$;
$(2)$ $q(t)$ is left continuous for $t=1$;
$(3)$ $q(t)$ is right continuous for $ t=\varepsilon$;
$(4)$ $q(\varepsilon)\geq 0$; and
$(5)$ for every $ t\in [\varepsilon,1)$ the right derivative of $q(t)$ is defined and satisfies: $ \frac{\partial _+ q(r)}{\partial r}\big|_{r=t} > f(O)-q(t)$.
Then: $ q(1)\geq \left( 1-e^{-(1-\varepsilon)}\right)f(O)$.
\end{claim}
Let us now show that the above weaker claim proves Claim \ref{claim:differentialEQ}.
For every $\eta >0$ define $ q_{\eta}\colon [\varepsilon,1]\rightarrow \mathbb{R}$ as follows: $q_{\eta}(t)\triangleq q(t)+\eta (t-\varepsilon)$.
Note that $ q_{\eta}$ satisfies conditions $(1)$, $(2)$, $(3)$, and $(4)$ of Claim \ref{claim:differentialEQ-weak} (since $q$ satisfies conditions $(1)$, $(2)$, $(3)$, and $(4)$ of Claim \ref{claim:differentialEQ}).
Focusing on requirement $(5)$, for every $ t\in [\varepsilon,1)$:
\[ \frac{\partial _+ q_{\eta}(r)}{\partial r}\bigg|_{r=t} = \frac{\partial _+ q(r)}{\partial r}\bigg|_{r=t}+\eta \geq f(O)-q(t)+\eta=f(O)-q_{\eta}(t)+\eta (1+t-\varepsilon) > f(O)-q_{\eta}(t),\]
where the first (weak) inequality follows since $q$ satisfies condition $(5)$ of Claim \ref{claim:differentialEQ} and the last (strong) inequality follows since $1+t-\varepsilon >0$.

Thus, for every $ \eta >0$, applying Claim \ref{claim:differentialEQ-weak} with $q_{\eta} $ implies that: $ q_{\eta}(1)\geq \left( 1-e^{-(1-\varepsilon)}\right)f(O)$.
Hence, for every $\eta >0$: $ q(1)\geq \left( 1-e^{-(1-\varepsilon)}\right)f(O)-\eta (1-\varepsilon)$.
Since the above holds for every $\eta >0$, we can conclude that $  q(1)\geq \left( 1-e^{-(1-\varepsilon)}\right)f(O)$.

All that remains is to prove Claim \ref{claim:differentialEQ-weak}.
\begin{proof}[Proof of Claim \ref{claim:differentialEQ-weak}]

Let \[X\triangleq \left\{ t:\varepsilon \leq t \leq 1 \text{ and } q(t)\geq \left( 1-e^{-(t-\varepsilon)}\right)f(O) \right\} .\]
We note that $X\neq \emptyset$ since condition $(4)$ implies $q(\varepsilon)\geq 0$ and $\left( 1-e^{-(\varepsilon-\varepsilon)}\right)f(O)=0$ and thus $ \varepsilon \in X$.
Thus, since $ X\neq \emptyset$ we can define $ s\triangleq \sup \{ x:x\in X\}$.
If $s=1$ we are done since condition $(2)$ implies that:
\[ q(1)=\lim _{\delta \rightarrow 0^+}q(1-\delta) \geq \left( 1-e^{-(1-\varepsilon)}\right)f(O),\]
where the inequality follows from the assumption that $s=1$.

Thus, let us assume that $ s<1$. Therefore,
\begin{align}
\frac{\partial _+ q(r)}{\partial r}\bigg|_{r=s} & =\lim _{\delta \rightarrow 0^+} \frac{1}{\delta}\cdot \left( q(s+\delta) - q(s)\right)  \nonumber \\
& \leq \lim _{\delta \rightarrow 0^+}\frac{1}{\delta}\cdot \left( 1-e^{-(s+\delta-\varepsilon)} - 1+e^{-(s-\varepsilon)}\right) f(O) \label{diffEQ1} \\
& = e^{-(s-\varepsilon)}f(O)\cdot \lim _{\delta \rightarrow 0^+}\frac{1}{\delta}\left( 1-e^{-\delta}\right) \nonumber\\
& = e^{-(s-\varepsilon)}f(O).\nonumber
\end{align}
Inequality \eqref{diffEQ1} follows from the following two:
$(1)$ requirements $(1)$ and $(3)$ of $q$, together with the definition of $s$, imply that $ q(s)=\left( 1-e^{-(s-\varepsilon)}\right) f(O)$; and
$(2)$ the definition of $s$ implies that for every $t\in (s,1]$ we have: $ q(t)<\left( 1-e^{-(t-\varepsilon)}\right) f(O)$.

Requirement $(5)$ of $q$, together with the above observation that $ q(s)=\left( 1-e^{-(s-\varepsilon)}\right) f(O)$, yield:
\[ \frac{\partial _+ q(r)}{\partial r}\bigg|_{r=s} > f(O)-q(s) = f(O)-\left( 1-e^{-(s-\varepsilon)}\right) f(O) = e^{-(s-\varepsilon)}f(O).\]
This is a contradiction.
Hence, it cannot be the case that $s<1$.
This completes the proof of Claim \ref{claim:differentialEQ-weak}.
%
%
%
\end{proof}
This concludes the proof of Claim~\ref{claim:differentialEQ}.
\end{proof}

\differentEQapp*
\label{app:differentialEQ-approximate}
\begin{proof}
First we prove a slightly weaker version of Claim \ref{claim:differentialEQ-approximate}, where the inequality in requirement $(5)$ is a strong inequality.
This is summarized in the following claim.
\begin{claim}\label{claim:differentialEQ-weak-approximate}
Let $ q\colon [\varepsilon,1]\rightarrow \mathbb{R}$ be a function satisfying:
$(1)$ $q(t)$ is continuous for every $t\in (\varepsilon,1)$;
$(2)$ $q(t)$ is left continuous for $t=1$;
$(3)$ $q(t)$ is right continuous for $ t=\varepsilon$;
$(4)$ $q(\varepsilon)\geq 0$; and
$(5)$ there exist $ \beta\in (0,1]$ and $\eta \geq 0$ such that for every $ t\in [\varepsilon,1)$ the right derivative of $q(t)$ is defined and satisfies: $ \frac{\partial _+ q(r)}{\partial r}\big|_{r=t} > (\beta-\eta)f(O)-\beta q(t)$.
Then: $ q(1)\geq (1-\nicefrac{\eta}{\beta})\left( 1-e^{-\beta (1-\varepsilon)}\right)f(O)$.
\end{claim}
Let us now show that the above weaker claim proves Claim \ref{claim:differentialEQ-approximate}.
For every $\eta >0$ define $ q_{\eta}\colon [\varepsilon,1]\rightarrow \mathbb{R}$ as follows: $q_{\eta}(t)\triangleq q(t)+\eta (t-\varepsilon)$.
Note that $ q_{\eta}$ satisfies conditions $(1)$, $(2)$, $(3)$, and $(4)$ of Claim \ref{claim:differentialEQ-weak-approximate} (since $q$ satisfies conditions $(1)$, $(2)$, $(3)$, and $(4)$ of Claim \ref{claim:differentialEQ-approximate}).
Focusing on requirement $(5)$, for every $ t\in [\varepsilon,1)$:
\[ \frac{\partial _+ q_{\eta}(r)}{\partial r}\bigg|_{r=t} = \frac{\partial _+ q(r)}{\partial r}\bigg|_{r=t}+\eta \geq f(O)-q(t)+\eta=f(O)-q_{\eta}(t)+\eta (1+t-\varepsilon) > f(O)-q_{\eta}(t),\]
where the first (weak) inequality follows since $q$ satisfies condition $(5)$ of Claim \ref{claim:differentialEQ-approximate} and the last (strong) inequality follows since $1+t-\varepsilon >0$.

Thus, for every $ \eta >0$, applying Claim \ref{claim:differentialEQ-weak-approximate} with $q_{\eta} $ implies that: $ q_{\eta}(1)\geq (1-\eta/\beta)\left( 1-e^{-\beta (1-\varepsilon)}\right)f(O)$.
Hence, for every $\eta >0$: $ q(1)\geq (1-\eta/\beta)\left( 1-e^{-\beta (1-\varepsilon)}\right)f(O)-\eta (1-\varepsilon)$.
Since the above holds for every $\eta >0$, we can conclude that $  q(1)\geq (1-\eta/\beta)\left( 1-e^{-\beta (1-\varepsilon)}\right)f(O)$.

All that remains is to prove Claim \ref{claim:differentialEQ-weak-approximate}.
\begin{proof}[Proof of Claim \ref{claim:differentialEQ-weak-approximate}]
Let \[X\triangleq \left\{ t:\varepsilon \leq t \leq 1 \text{ and } q(t)\geq (1-\nicefrac{\eta}{\beta})\left( 1-e^{-\beta (t-\varepsilon)}\right)f(O) \right\} .\]
We note that $X\neq \emptyset$ since condition $(4)$ implies $q(\varepsilon)\geq 0$ and $(1-\eta/\beta)\left( 1-e^{-\beta (\varepsilon-\varepsilon)}\right)f(O)=0$ and thus $ \varepsilon \in X$.
Thus, since $ X\neq \emptyset$ we can define $ s\triangleq \sup \{ x:x\in X\}$.
If $s=1$ we are done since condition $(2)$ implies that:
\[ q(1)=\lim _{\delta \rightarrow 0^+}q(1-\delta) \geq (1-\nicefrac{\eta}{\beta})\left( 1-e^{-\beta (1-\varepsilon)}\right)f(O),\]
where the inequality follows from the assumption that $s=1$.

Thus, let us assume that $ s<1$. Therefore,
\begin{align}
\frac{\partial _+ q(r)}{\partial r}\bigg|_{r=s} & =\lim _{\delta \rightarrow 0^+} \frac{1}{\delta}\cdot \left( q(s+\delta) - q(s)\right)  \nonumber \\
& \leq \lim _{\delta \rightarrow 0^+}\frac{1}{\delta}\cdot \left( (1-\nicefrac{\eta}{\beta})\left(1-e^{-\beta (s+\delta-\varepsilon)}\right) - (1-\nicefrac{\eta}{\beta}) \left(1-e^{-\beta (s-\varepsilon)}\right)\right) f(O) \label{diffEQ2987} \\
& = (1-\nicefrac{\eta}{\beta})\cdot e^{-\beta (s-\varepsilon)}\cdot f(O)\cdot \lim _{\delta \rightarrow 0^+}\frac{1}{\delta}\left( 1-e^{-\beta\delta}\right) \nonumber\\
& = (\beta-\eta)\cdot e^{-\beta (s-\varepsilon)}\cdot f(O).\nonumber
\end{align}
Inequality \eqref{diffEQ2987} follows from the following two:
$(1)$ requirements $(1)$ and $(3)$ of $q$, together with the definition of $s$, imply that $ q(s)=(1-\nicefrac{\eta}{\beta})\left( 1-e^{-\beta (s-\varepsilon)}\right) f(O)$; and
$(2)$ the definition of $s$ implies that for every $t\in (s,1]$ we have: $ q(t)<(1-\nicefrac{\eta}{\beta})\left( 1-e^{-\beta (t-\varepsilon)}\right) f(O)$.

Requirement $(5)$ of $q$, together with the above observation that $ q(s)=(1-\nicefrac{\eta}{\beta})\left( 1-e^{-\beta (s-\varepsilon)}\right) f(O)$, yield:
\[ \frac{\partial _+ q(r)}{\partial r}\bigg|_{r=s} > (\beta-\eta)f(O)-\beta q(s) = (\beta-\eta)f(O)-(\beta-\eta)\left( 1-e^{-\beta (s-\varepsilon)}\right) f(O) = (\beta-\eta) \cdot e^{-\beta (s-\varepsilon)}\cdot f(O).\]
This is a contradiction.
Hence, it cannot be the case that $s<1$.
This completes the proof of Claim \ref{claim:differentialEQ-weak-approximate}.
%
%
%
\end{proof}
This concludes the proof of Claim~\ref{claim:differentialEQ-approximate}.  
\end{proof}

\section{Preprocessing}\label{sec:preprocessing}

In this section we provide the proof for Lemmas~\ref{lem:simpleprep} and~\ref{lem:generalized_prep}.
In this section we consider the problem of maximizing a monotone submodular function $f:2^{U}\rightarrow \mathbb{R}_{\geq 0}$ with a generalized partition matroid, where $U$ is partitioned into $U_1,\ldots,U_k$ and up to $\ell_j\in \mathbb{N}$ items can be selected form $U_j$. When we deal with the case of a regular partition  matroid, we view it as a special case of the generalized one in which $\ell_j=1$ for all $j\in [k]$. We also define $r=\sum_{j=1}^{k}\ell_j$ as the rank of the matroid.

Our pre-processing algorithms utilize a variant of the residual random greedy \cite{BFNS14}.  
The pseudo-code  of the algorithm  is given in Algorithm \ref{alg:residual}. 
In particular, our  pre-processing algorithms always return one of the sets $P_0, \ldots, P_r$ returned by Algorithm \ref{alg:residual}. A central challenge of the pre-processing algorithms is to pick the right set to be returned, without causing a significant  overhead to the running time.

\begin{algorithm}
	\caption{Residual Random Greedy$(f,(U_j)_{j\in [k]}, (\ell_j)_{j\in [k]}, \delta )$}
	\SetKwInOut{Input}{input}
	\SetAlgoNlRelativeSize{0}
	\label{alg:residual}
	
	\Input{A monotone submodular function $f:2^{U}\rightarrow \mathbb{R}_+$, partition $U_1,\ldots, U_k$ of $U$} 
	
	Initialize $P_0 \leftarrow \emptyset$, and $i\leftarrow 0$. 
	
	\For{$i = 1$ \textnormal{\textbf{to}} $r$}
		{

			Sample $j(i)\in [k]$ such that $\Pr(j(i)=q) = \frac{\ell_q - \abs{P_{i-1}\cap U_j}}{r-(i-1)} $. \label{residual:select}

            Sample a set $X_i \subseteq U_{j(i)} \setminus P_{i-1}$ of size $\min\left\{ \abs{U_j\setminus P_{i-1}}, \frac{\abs{U_{j(i)}\setminus P_{i-1}}}{\ell_{j(i)}-\abs{P_{i-1}\cap U_{j(i)}}}\cdot \log\left(\frac{1}{\delta}\right) \right\}$. 
            
			Define $e_i = \argmax_{e\in X_i} \left(f(P_{i-1}+e) - f(P_{i-1}) \right)$ \label{residual:argmax}
			
			Update $P_{i} \leftarrow P_{i-1} + e_i$.
	}
	\Return{$P_0,\ldots, P_r$}
	
\end{algorithm}

We note that the number of values queries used by  Algorithm~\ref{alg:residual} is nearly linear. 
\begin{lemma}
\label{lem:residual_runtime}
	Algorithm~\ref{alg:residual} uses $O\left(n\cdot \log(n)\cdot \log\left(\frac{1}{\delta}\right)\right)$ values queries for a general partition matroid and $O(n)$ queries for a regular partition matroid. 
\end{lemma}
\begin{proof}
Observe that the algorithm only uses oracle queries in Step~\ref{residual:argmax} which is invoked in every iteration. 
Furthermore, a specific part $q$ is selected exactly $\ell_q$ times ($j(i)=q$ for exactly $\ell_q$ values of $i$), and each time $q$ is selected the value $\ell_q-\abs{P_{i-1}\cap U_q}$ is decreasing by $1$.
The number of queries needed in these iterations is at most 
$$
\sum_{s=1}^{\ell_q} \frac{\abs{U_q}}{s}\cdot \log\left(\frac{1}{\delta}\right) \leq O\left(|U_q| \cdot \log(\abs{U_q}) \cdot \log(1/\delta)\right),
$$
where the last inequality holds as the $1+\frac{1}{2}+\frac{1}{3}+\ldots +\frac{1}{\ell_q} =O(\ln (\ell_q))$.
Therefore, the total number of queries used by Algorithm~\ref{alg:residual}
is 
$$O\left(\sum_{q=1}^{k} |U_q| \cdot \log(\abs{U_q}) \cdot \log(1/\delta)\right)=O\left(n\cdot \log(n)\cdot \log1/\delta\right).$$
In case the matroid is a regular partition matroid, the running time is actually linear. This holds as in this case, each part $U_j$ can only be selected once in Line~\ref{residual:select}, and the size of the set $X_i$ in this case is exactly $U_j$. Therefore, the number of total queries is at most 

$$ \sum_{i=1}^k 2|U_i|\leq 2n$$
as claimed.
\end{proof}

To state the following lemmas, we first establish notation for the probability space induced by the algorithm.
Define 
$$
\cF_i = \sigma(j(1),\ldots, j(i),X_1,\ldots, X_i)
$$
as the $\sigma$-algebra of the random variables $j(1),\ldots, j(i)$ and $X_1,\ldots, X_i$ for all $i\in \{0,\ldots, r\}$.  Intuitively, $\cF_i$ are all the events whose value is determined by the end of the $i$-th  iteration of the algorithm. This implies that a random variable  $Z$ is $\cF_i$-measurable if its value is known by the end of the $i$-th iteration. For example $P_i$  and $j(i)$ are $\cF_i$-measurable, but $S_{i+1}$ and $j(i+1)$ are not $\cF_i$-measurable.   We remind the reader that the conditional expectation $\E[X\,|\,\cF_i]$, of a random variable $X$, is a random variable which is $\cF_i$-measurable. 

For every $i\in \{0,1,\ldots,r\}$ define the value 
\begin{equation}
	\label{eq:V_def}
	V_i=\max_{S\subseteq U\setminus P_i: |S\cap U_j| \leq \ell_j - \abs{P_i\cap U_j} \forall j} \sum_{e\in S} \left(f(P_{i}+e) - f(P_{i}) \right).
\end{equation}
In particular, if we return the set $P_i$ as the output of the pre-processing algorithm, then the events in Lemmas~\ref{lem:simpleprep}  and~\ref{lem:generalized_prep} is  $V_i\leq c\cdot f(O)$.  It can be easily observed that the sequence $V_0,V_1,\ldots, V_r$ is decreasing and $V_r=0$. That is,
$$
V_0\geq V_1\geq \ldots \geq V_r=0.
$$
Furthermore, as a random variable, $V_i$ is $\cF_i$-measurable-  its value is  determined by the end of the $i$-th iteration.

A central lemma in the analysis of Algorithm~\ref{alg:residual} is that the gain in value of $f(P_i)$ in the $i$-th iteration,  $f(P_i)-f(P_{i-1})$, increases in expectation in proportion to $V_i$. 
\begin{lemma}
	\label{lem:residual_P_gain}
For every  $i\in [r]$ it holds that $\E\left[ f(P_i)-f(P_{i-1})\,|\,\cF_{i-1}\right]\geq  (1-\delta)\cdot \frac{V_{i-1}}{r-(i-1)}$.
\end{lemma}
\begin{proof}

Let $S^*$ denote the set that is the maximizer in \eqref{eq:V_def} defining $V_{i-1}$. Let $j$ denote the part picked by Algorithm~\ref{alg:residual} in $i^{th}$ iteration and $S^*_{j}=S^*\cap U_{j}$. By a simple concentration inequality, $Pr[X_{j}\cap S^*_{j}=\emptyset]\leq \delta$ and by symmetry for each $e\in S_j$ occurs in $X_j$ with probability at least $\frac{1-\delta}{|S^*_j|}$ where we have $|S^*_j|=l_j-|P_{i-1}\cap U_j|$. Since, the element $e$ picked by the algorithm is greedy among all the elements of  sampled set, we have 

    $$
    \E[ f(P_{i})-f(P_i) \,|\,\cF_{i-1}, j] \geq (1-\delta)\cdot \frac{ \sum_{e\in S_j^*} \left(P_{i-1}+e\right)-f(P_{i-1})  }{\ell_{j} - \abs{P_{i-1}\cap U_{j}}},
    $$

Now taking expectation of picking part $j$ in iteration $i$, we obtain 

	\begin{equation*}
		\begin{aligned}
			\E\left[ f(P_i)-f(P_{i-1})\,|\,\cF_{i-1}\right]   
			 &\geq \sum_{j=1}^k \frac{\ell_{j} - \abs{P_{i-1}\cap U_{j}}}{r-(i-1)} \cdot (1-\delta)\cdot \frac{ \sum_{e\in S_j^*} \left(P_{i-1}+e\right)-f(P_{i-1})  }{\ell_{j} - \abs{P_{i-1}\cap U_{j}}}\\
             &= 
     \frac{1}{r-(i-1)} \cdot \sum_{e\in S}  \left(  f(P_{i-1}+e)-f(P_{i-1}) \right) \\
			& = \frac{V_{i-1}}{r-(i-1)}.\\
		\end{aligned}
	\end{equation*}
\end{proof}

Let $\OPT= \bigcup_{j=1}^{k} \bigcup_{s=1}^{\ell_j} \{o_{j,s}\}$ be an optimal solution for the input instance $f$ 
where $o_{j,s}\in U_j$ for every $j\in \{1,\ldots, k\}$.  
We require that the order between $o_{j,1},\ldots, o_{j,\ell_{j}}$ for all $j\in[k]$ will be selected at random
We define a residual solution $\OPT_i = \OPT/P_i$ for every iteration $i\in \{0,1,\ldots, k\}$.  That is,  
$$\OPT_i = \bigcup_{j=1}^{k} \bigcup_{s=1}^{\ell_j-\abs{U_j\cap P_i}}\{o_{j,s}\}.$$

Assuming we fix the order of items with $\OPT$, the residual solution $\OPT_i$ is a random variable whose value is determined by the end of the $i$-th iteration.
By definition it holds that, 
$$
\OPT=\OPT_0\supseteq \OPT_1\supseteq \OPT_2\supseteq \ldots \supseteq \OPT_{r}
$$
and $\OPT_i \cup P_i\in \cI$ for all $i\in [r]$. 
	\begin{lemma}
	\label{lem:residual_opt_loss}
	For all $i\in [r]$ it holds that 
	$$
	\E[ f(\OPT_{i-1}) - f(\OPT_i) \,|\,\cF_{i-1}] \leq  \frac{f(\OPT)}{r-(i-1)}.
	$$
\end{lemma}
\begin{proof}
Fix an iteration $i\in [r]$ and define $s_j = \ell_k -
\abs{U_j\cap P_{i-1}}$ for all $j\in [k]$. Observe those are random variables, and that $j(i)=q$ is $\frac{s_j}{r-(i-1)}$. Therefore,
$$
	\begin{aligned}
		\E&[ f(\OPT_{i-1}) - f(\OPT_i) \,|\,\cF_{i-1}, \OPT_{i-1}] \\
		&= \E\left[ \sum_{q\in [k]} \indicator_{j(i)=q} \cdot \left( f(\OPT_{i-1} ) - f(\OPT_{i-1} -o_{j',s_{j'}})\right)\,|\,\cF_{i-1}, O_{i-1}\right]\\
		&=  \frac{1}{r-(i-1)}\cdot \sum_{o\in O_{i-1} }  \left( f(\OPT_{i-1} ) - f(\OPT_{i-1} -o)\right)\\
		&\leq \frac{1}{r-(i-1)} \cdot f(\OPT_{i-1}) \\
		&\leq    \frac{f(\OPT)}{r-(i-1)}. \\
	\end{aligned}
	$$
    The second inequality is correct as the second expression $o_{j(i),s_{j(i)}}$ is a random item from $O_{i-1}$.
	The first inequality follows from the submodularity of $f$, and  the second  inequality from monotonicity. 
\end{proof}

Given an arbitrary  threshold $T>0$ we can define a (random) index  $I_T$ by
\begin{equation}
	\label{eq:I_T_def}
	I_T = \min\{i\in \{0,1,\ldots, n\} ~|~ V_i\leq T\}.\end{equation}
That is, $I_T$ is  the first index for which $V_i$ is below (or equals to) the threshold $T$. If $T$ is known by context we simply use $I$.  We note that the events $\{ I_T\leq i\}=\{V_i \leq T \}$ and  $\{i<I_T\} = \{V_i > T\}$
are $\cF_i$-measurable. 
Using 
Lemmas \ref{lem:residual_P_gain} and \ref{lem:residual_opt_loss} we can show  that the set $P_{I_T}$ satisfies the conditions in Lemmas~\ref{lem:simpleprep} and~\ref{lem:generalized_prep}, subject to the requirement that $ (e-1)\cdot f(\OPT)\leq T\leq c\cdot f(\OPT)$.
\begin{lemma}
	\label{lem:good_P}
	Let  $\frac{e-1}{1-\delta}\cdot f(\OPT)\leq T\leq c\cdot f(\OPT)$  be a threshold, and let $P=P_{I_T}$. Then,
	\begin{itemize}
		\item $ \max_{S\subseteq U\setminus P_i: |S\cap U_j| \leq \ell_j - \abs{P_i\cap U_j} \forall j} \sum_{e\in S} \left(f(P_{i}+e) - f(P_{i}) \right)\leq c\cdot f(\OPT)$ with probability of $1$. 
		\item $\mathbb{E}\left[(1-\nicefrac{1}{e})\cdot f(\OPT_{I_T})) + \frac{1}{e}\cdot f(P)\right]\geq (1-\nicefrac{1}{e})\cdot f(\OPT)$. 
	\end{itemize}
\end{lemma}
\begin{proof}
The first property trivially holds from the definition of $I_T$. 

The second property of the lemma follows from Lemmas \ref{lem:residual_P_gain} and \ref{lem:residual_opt_loss}. 
\begin{equation}
	\label{eq:resdiaul_proper_stop}
\begin{aligned}
	\E[f(P_{I_T})-f(\emptyset)] &= \E\left[ \sum_{i=1}^{k} \indicator_{i-1 < I_T} \left(f(P_i) -f(P_{i-1})\right) \right]\\
	&= \E\left[ \sum_{i=1}^{k} \indicator_{i-1< I_T} \cdot  \E \left[f(P_i) -f(P_{i-1})\,|\, \cF_{i-1}\right] \right]\\
		&\geq  \E\left[ \sum_{i=1}^{k} \indicator_{i\leq I_T}  \cdot(1-\delta)\cdot  \frac{V_{i-1}}{k-(i-1)} \right]\\
		&\geq  \E\left[ \sum_{i=1}^{k} \indicator_{i\leq I_T} \cdot \frac{(e-1)\cdot f(\OPT)}{k-(i-1)} \right]\\
	&\geq  \E\left[(e-1) \cdot \sum_{i=1}^{k} \indicator_{i\leq I_T} \cdot \E\left[ f(\OPT_{i-1})-f(\OPT_i)\,|\, \cF_{i-1}\right] \right]\\
	&=  (e-1)\cdot \E\left[ \sum_{i=1}^{k} \indicator_{i\leq  I_T} \cdot \left( f(\OPT_{i-1})-f(\OPT_i) \right) \right] \\
	&=(e-1)\cdot \E\left[ f(\OPT) - f(\OPT_{I_T}) \right].
\end{aligned}
\end{equation}
The first inequality follows from Lemma~\ref{lem:residual_P_gain}. The second inequality holds as for every $i\leq I_T$ we have $V_{i -1}> T \geq \frac{e-1}{1-\delta}\cdot f(\OPT)$ by \eqref{eq:I_T_def}. The third inequality follows from Lemma~\ref{lem:residual_opt_loss}. 

Rearranging \eqref{eq:resdiaul_proper_stop} we get,
$$
\E\left[ (e-1)\cdot f(\OPT_{I_T})  +f(P_{I_T})\right]  \geq (e-1)\cdot f(\OPT) +f(\emptyset).
$$
Dividing both sides by $e$ we obtain,
$$
\E\left[ (1-\nicefrac{1}{e})\cdot f(\OPT_{I_T})  +\frac{1}{e}\cdot f(P)\right]   
\geq (1-\nicefrac{1}{e})\cdot f(\OPT) + \frac{1}{e}\cdot f(\emptyset) \geq (1-\nicefrac{1}{e})\cdot f(\OPT).
$$
Therefore the second property of the lemma holds as well.
\end{proof}

\begin{proof}[Proof of Lemma~\ref{lem:generalized_prep}]
Lemma~\ref{lem:good_P} suggests a method by which we can return a set which satisfies the conditions of Lemma~\ref{lem:generalized_prep}. We find a  $\gamma$-approximation for $f(\OPT)$ and use it to set a threshold $T$ which satisfies the condition of Lemma~\ref{lem:good_P}.  
This can be efficiently done in time $O(n\log r)$ for $\gamma=\nicefrac{1}{3}$ using an algorithm from \cite[Appendix B]{BFS17}. 
Then, we use a binary search over the array $(V_0,V_1,\ldots, V_r)$  to find $I_T$, the first entry in the array whose value is at most $T$. This requires looking up for $O(\log r)$ values in the array, and hence only $O(\log r)$ entries of the array need to be evaluated. As an explicit evaluation of  each entry using the formula in \eqref{eq:V_def} takes  up to $2n$ value queries, the total number of value queries required for the the binary search is $O(n\log k)$.  We give the pseudo-code of the procedure in Algorithm~\ref{alg:binary_prep}.

\begin{algorithm}
	\caption{Binary Search Pre-processing $(f,(U_j)_{j\in [k]}, (\ell_j)_{j\in [k]},\delta)$}
	\SetKwInOut{Input}{input}
	\SetAlgoNlRelativeSize{0}
	\label{alg:binary_prep}
	
	\Input{A monotone submodular function $f:2^{U}\rightarrow \mathbb{R}_+$, partition $U_1,\ldots, U_k$ of $U$ and bounds $\ell_1,\ldots, \ell_k$ }

    Find a value  $\app$ such that $\gamma \cdot f(\OPT) \leq \app\leq  f(\OPT)$  for $\gamma=\nicefrac{1}{3}$
	
	Let $P_0,P_1,\ldots,  P_k$ be the output of the residual Random Greedy  on $(f,(U_j)_{j\in [k]}, (\ell_j)_{j\in [k]})$ with $\delta$
	
	Define $T = \frac{e-1}{(1-\delta)\gamma} \cdot \app $ 
	
	Find $I_T$ using binary search on the array $(V_0,V_1,\ldots, V_r)$ as defined in \eqref{eq:V_def} 
	
	Return $P_{I_T}$
\end{algorithm}

The value of $T$ used by Algorithm~\ref{alg:binary_prep} satisfies the condition $\frac{e-1}{1-\delta}\cdot f(\OPT)\leq T \leq \frac{e-1}{\gamma}  \cdot f(\OPT)$. Hence, by Lemma~\ref{lem:good_P} and the above discussion this completes the proof of the Lemma~\ref{lem:generalized_prep}.
\end{proof}

Now, we give the proof  Lemma~\ref{lem:simpleprep} for partition matroids. 
Algorithm~\ref{alg:binary_prep} nearly meets the requirements of Lemma~\ref{lem:simpleprep}. However, its running time is not linear in $n$. Notably,  the running time of the binary search operation which finds $I_T$ serves as a bottleneck for the algorithm's running time. Recall that for the partition matroid, we have $r=k$ since we pick exactly one element from each part.  
We speed up the processing by searching for the minimal index $\tI\in \checkpoints(k,\eps)$  such that $V_{\tI} \leq T$, where $\checkpoints(k,\eps)$ is a carefully selected set of {\em check points}.
The set $\checkpoints(k,\eps)$ is sufficiently sparse so  calculating $V_i$ for every $i\in \checkpoints(k,\eps)$ takes a linear number of queries (in expectation) and is also sufficiently dense so  restricting our attention to sets $P_i$ with  $i\in\checkpoints(k,\eps)$ still allows us to satisfy the conditions of a pre-processing algorithm.  

The following lemma summerizes the properties of the set $\checkpoints$. 
\begin{lemma}
\label{lem:checkpoints}
For every $k\in \mathbb{N}$ and $0<\eps<\nicefrac{1}{2}$ there is a set $\checkpoints(k,\eps) \subseteq \{0,1,\ldots, k\}$  which contains $k$ and satisfies the following:
\begin{enumerate}
\item 
\label{check:runtime}
$\sum_{i\in \checkpoints(k,\eps)}\left( k-i\right) \leq d\cdot \frac{k}{\eps}$, and 
\item 
\label{check:quality}
For all $i\in \{0,1,\ldots,k\}$ it holds that $\sum_{i'=i+1}^{s(i)}\frac{1}{k-(i'-1)} \leq \eps $, where $s(i)=\min\{i'\in \checkpoints(k,\eps) ~|~i'\geq i\}$.
\end{enumerate}
Here $d>0$ is a global constant which does not depend on $k$ and $\eps$. 
\end{lemma}

Broadly speaking, the set $\checkpoints(k,\eps)$ contains the the values  $k\cdot \left(1- (1-\nicefrac{\eps}{4})^\ell\right)$   for every $\ell\in \mathbb{N}$, up to the ceiling operation. It can be easily verified that this construction satisfies the conditions of the lemma. 
We give the full details  in Appendix~\ref{sec:checkpoints}. 

Fix an arbitrary threshold $T>0$ and error parameter $\eps>0$. Define $\tI=\min\{i\in \checkpoints(k,\eps)~|~V_i\leq T\}$ as the first index in $\checkpoints(k,\eps)$ for which $V_i \leq T$. It trivially holds that $I_T=I\leq \tI$ and $\tI= s(I)$, where $s(I)$ is as defined in Lemma~\ref{lem:checkpoints}.  We use the following lemma, which is a consequence of Lemma~\ref{lem:residual_opt_loss} and Property~\ref{check:quality} of  Lemma~\ref{lem:checkpoints},  to show a variant of Lemma~\ref{lem:good_P} which considers  $P=P_{\tI}$ instead of $P_I$.
\begin{lemma}
\label{lem:correct_to_cp_diff}
$\E\left[ f(\OPT_{\tI}) -f(\OPT_I) \right] \geq - \eps \cdot f(\OPT).$
\end{lemma}
\begin{proof}
As in Lemma~\ref{lem:checkpoints} we use $s(i)=\min\{i'\in \checkpoints(k,\eps) ~|~i'\geq i\}$. 
We first split the expectation into a sum of expections, and bound each separately. 
\begin{equation}
\label{eq:opt_CP_first}
\begin{aligned}
\E&\left[ f(\OPT_{\tI}) -f(\OPT_I) \right] = \E\left[\sum_{i=0}^k \indicator_{I=i} \cdot \left(f(\OPT_{\tI}) -f(\OPT_I)\right) \right] \\
&=\E\left[\sum_{i=0}^k \indicator_{I=i} \cdot \left(f(\OPT_{s(i)}) -f(\OPT_i)\right) \right] \\
&=\sum_{i=0}^k \E\left[ \indicator_{I=i} \cdot \left(f(\OPT_{s(i)}) -f(\OPT_i)\right) \right].
\end{aligned}
\end{equation}
The  second equality holds as $\tI = s(I)$.  
Furthermore, for every $i\in \{0,1,\ldots, k\}$ we have,
\begin{equation}
\label{eq:opt_CP_second}
\begin{aligned}
\E&\left[ \indicator_{I=i} \cdot \left(f(\OPT_{s(i)}) -f(\OPT_i)\right) \right] = 
\E\left[  \sum_{i'=i+1}^{s(i)} \indicator_{I=i} \cdot \left(f(\OPT_{i'}) -f(\OPT_{i'-1})\right) \right] \\
& =\E\left[  \sum_{i'=i+1}^{s(i)} \E\left[\indicator_{I=i} \cdot \left(f(\OPT_{i'}) -f(\OPT_{i'-1})\right) |\,\cF_{i'-1}~\right]\,\right]\\
& =\E\left[  \sum_{i'=i+1}^{s(i)} \indicator_{I=i} \cdot \E\left[ f(\OPT_{i'}) -f(\OPT_{i'-1})\ |\,\cF_{i'-1}~\right]\,\right]\\
&\geq \E\left[  \sum_{i'=i+1}^{s(i)} \indicator_{I=i} \cdot (-1)\cdot  \frac{ f(\OPT)}{k-(i'-1)} \,\right] \\
&= -\E[\indicator_{I=i}]\cdot \sum_{i'=i+1}^{s(i)} \frac{1}{k-(i'-1)}\\
&\geq  -\Pr[{I=i}]\cdot \eps \cdot f(\OPT),
\end{aligned}
\end{equation}
The second equality is the tower property, and the third equality holds as $\indicator_{I=i}$ is $\cF_{i'-1}$-measurable  as $i'\geq i+1$.  The first inequality follows from Lemma~\ref{lem:residual_opt_loss}, and the last inequality holds due to Property~\ref{check:quality} of Lemma~\ref{lem:checkpoints}. 

Incorporating \eqref{eq:opt_CP_second} into \eqref{eq:opt_CP_first} we get, 
\begin{equation*}
\begin{aligned}
\E&\left[ f(\OPT_{\tI}) -f(\OPT_I) \right]\\ 
&=\sum_{i=0}^k \E\left[ \indicator_{I=i} \cdot \left(f(\OPT_{s(i)}) -f(\OPT_i)\right) \right] \\
&\geq -\sum_{i=0}^k \Pr[I=i] \cdot \eps \cdot f(\OPT) \\
&= - \eps \cdot f(\OPT).
\end{aligned}
\end{equation*}
\end{proof}

We can proceed to show a variant of Lemma~\ref{lem:good_P} which considers $P=P_{\tI}$.
\begin{lemma}
	\label{lem:good_tP}
	Let  $(e-1)\cdot f(\OPT)\leq T\leq c\cdot f(\OPT)$  be a threshold, let $\eps>0$ be an accuracy parameter, and let $P=P_{\tI}$. Then,
	\begin{itemize}
		\item $ \sum _{j:U_j\cap P=\emptyset} \max_{e\in U_j}\left( f(P+e) -f(P)\right)\leq c\cdot f(\OPT)$ with probability of $1$. 
		\item $\mathbb{E}\left[(1-\nicefrac{1}{e})f(\OPT/P) + \frac{1}{e}\cdot f(P)\right]\geq (1-\nicefrac{1}{e} -\eps)\cdot f(\OPT)$. 
	\end{itemize}
\end{lemma}
\begin{proof}
As in the proof of Lemma~\ref{lem:good_P}, the first property trivially follows from the definition of  $\tI$. 
$$
 \sum _{j:U_j\cap P=\emptyset} \max_{e\in U_j}\left( f(P+e) -f(P)\right) = 
  \sum _{j\in[k]:~U_j\cap P_{\tI} =\emptyset } \max_{e\in U_j}\left( f(P_{\tI}+e) -f(P_{\tI})\right) = V_{\tI} \leq T\leq c\cdot f(\OPT).
$$

Recall the shorthand $I=I_T$.
To show the second property we use the bound in Lemma~\ref{lem:correct_to_cp_diff}. Specifically, we have,
$$ 
\begin{aligned}
 \mathbb{E}\left[(1-\nicefrac{1}{e})f(\OPT/P) + \frac{1}{e}\cdot f(P)\right] 
 &= \mathbb{E}\left[(1-\nicefrac{1}{e})\cdot f(\OPT_{I}) + \frac{1}{e}\cdot f(P_I)\right] \\
 &~~~+(1-\nicefrac{1}{e}) \cdot \E[f(\OPT_{\tI})-f(\OPT_I)] +\frac{1}{e}\cdot \E[ f(P_{\tI}) -f(P_{I})] \\
 &\geq (1-\nicefrac{1}{e}) \cdot f(\OPT) -\eps \cdot f(\OPT)  \\
 &= (1-\nicefrac{1}{e} -\eps)\cdot f(\OPT).
\end{aligned} 
$$
The first equality uses the fact that $\OPT/P= \OPT/P_{\tI} = \OPT_{\tI}$. The  inequality follows from  the Lemmas~\ref{lem:good_P} and \ref{lem:correct_to_cp_diff} and 
 as well as $f(P_{\tI})\geq f(P_I)$ since $\tI\geq I$. 
\end{proof}

In particular, Lemma~\ref{lem:good_tP} indicates that an algorithm that returns $P_{\tI}$ may serve as a good pre-processing algorithm.  Algorithm~\ref{alg:cp_prep} implements this approach. Since we aim for a linear time, we cannot use the $O(n\log k)$ constant-approximation algorithm of \cite{BFS17} to approximate the value of $f(\OPT)$ as in  Algorithm~\ref{alg:binary_prep}.  Instead, we assume the algorithm receives an approximation $\app$ for the value of $f(\OPT)$. Later, in Algorithm~\ref{alg:fast_prep}, we overcome  this issue. 

\begin{algorithm}
	\caption{Check Point  Pre-processing $(f,U_1,\ldots, U_k,\app,\eps)$}
	\SetKwInOut{Input}{input}
	\SetAlgoNlRelativeSize{0}
	\label{alg:cp_prep}
	
	\Input{A monotone submodular function $f:2^{U}\rightarrow \mathbb{R}_+$, partition $U_1,\ldots, U_k$ of $U$, a value $\app$ such that $\gamma\cdot f(\OPT)\leq  \app \leq \OPT$ and $0<\eps<\nicefrac{1}{2}$} 
	
	Let $P_0,P_1,\ldots,  P_k$ be the output of the residual Random Greedy  on $f$ and $U_1,\ldots, U_k$
	
    Compute $V_i$ (as in \eqref{eq:V_def}) for every $i\in \checkpoints(k,\eps)$ and find $\tI$ for $T=\frac{e-1}{\gamma}$.  
	
	Return $P_{\tI}$

\end{algorithm}

Property~\ref{check:runtime} of Lemma~\ref{lem:checkpoints} ensures the expected number of value queries used by Algorithm~\ref{alg:cp_prep} is linear in $n$. 
\begin{lemma}
\label{lem:cp_time}
Algorithm~\ref{alg:cp_prep} uses $O\left(\nicefrac{n}{\eps}\right)$ oracle  value queries in expectation.
\end{lemma}
\begin{proof}
For every index $i\in \checkpoints(k,\eps)$ let $C_i$ be the expected number of queries needed to compute $V_i$. Then,
$$
C_i =  2\cdot \sum_{j=1}^{k} \Pr[U_j\cap P_i =\emptyset] \cdot |U_j| =2\cdot \sum_{j=1}^{k}  \left(1-\frac{i}{k}\right) \cdot |U_j| = 2\cdot n\cdot   \left(1-\frac{i}{k}\right)= 2\cdot \frac{n}{k}\cdot (k-i). 
$$
Therefore, the expected number of value queries used by Algorithm~\ref{alg:cp_prep} is
$$
\sum_{i\in \checkpoints(k,\eps)} C_i =\sum_{i\in \checkpoints(k,\eps)} 2\cdot \frac{n}{k}\cdot   \left(k-i\right)  \leq 2\cdot \frac{n}{k}\cdot \frac{d\cdot k}{\eps} = O\left(\nicefrac{n}{\eps} \right),
$$
where the inequality follows from Property~\ref{check:runtime} of Lemma~\ref{lem:checkpoints}, the $d$ is the constant define in the lemma. 
\end{proof}

The following lemma is immediate consequence of Lemmas~\ref{lem:good_tP} and~\ref{lem:cp_time}. 
\begin{lemma}\label{lem:cp_prep}
For every partition matroid $ \cM=(U,\cI)$ with partition $ U=U_1\cup \ldots\cup U_k$, a monotone submodular function $ f\colon 2^U\rightarrow \mathbb{R}_+$, an error parameter $0<\eps<\nicefrac{1}{2}$ and value $\app$ such that $ \gamma\cdot f(\OPT) \leq \app\leq f(\OPT)$, Algorithm~\ref{alg:cp_prep} performs in expectation $ O(\nicefrac{n}{\eps})$ value queries, returns $ P\in \cI$ satisfying:
\begin{enumerate}
 
    \item $  \sum _{j:U_j\cap P=\emptyset}\alpha _{P,j} \leq c\cdot f(O)$ with probability $1$,  
    \item $\mathbb{E}\left[(1-\nicefrac{1}{e})f((O/P)\cup P) + \frac{1}{e}\cdot f(P) \right]\geq (1-\nicefrac{1}{e}-\eps)f(O)$.
\end{enumerate}
Here $ O=\argmax \{ f(S):S\in \cI\}$ is an optimal base, $ \alpha _{P,j}=\max _{i\in U_j}f_P(\{ i\})$ for every $ j$ such that $ P\cap U_j=\emptyset$, and $c$ and $\gamma$ are absolute constants.
\end{lemma}

We are left to overcome the fact that Algorithm~\ref{alg:cp_prep} needs to receive $\app$ as part of its input.
We use the Random Residual Greedy (Algorithm~\ref{alg:residual}) to attain this value in linear time. 
Specifically, we utilize the following results from \cite{BFNS14}.
\begin{lemma}[\cite{BFNS14}]
\label{lem:residual_approximation}
The set $P_k$ returned by Algorithm~\ref{alg:residual} satisfies $\E[f(P_k)] \geq \frac{1}{4} \cdot f(\OPT)$. 
\end{lemma}
In fact, it was shown in \cite{BFNS14} that 
Lemma~\ref{lem:residual_approximation} holds even if $f$ is non-monotone.  
We apply the above lemma in Algorithm~\ref{alg:fast_prep}, which fulfills the existence guarantee stated in Lemma~\ref{lem:simpleprep}.
\begin{algorithm}
	\caption{Fast Pre-processing $(f,U_1,\ldots, U_k,\eps)$}
	\SetKwInOut{Input}{input}
	\SetAlgoNlRelativeSize{0}
	\label{alg:fast_prep}
	
	\Input{A monotone submodular function $f:2^{U}\rightarrow \mathbb{R}_+$, partition $U_1,\ldots, U_k$ of $U$, and $0<\eps<\nicefrac{1}{2}$} 
	
    Run Algorithm~\ref{alg:residual} for $\ell = \log_{7/6}\left(\frac{8}{ \eps}\right) $ times  on the input instance\label{fast:iterations}
    
    Let $\app$ be the maximal value of $f(P_k)$ among all executions of Algorithm~\ref{alg:residual}
	
    Run Algorihtm~\ref{alg:cp_prep} with the input instance, error parameter $\nicefrac{\eps}{2}$ and $\app$. Return its output

\end{algorithm}

 \begin{proof}[Proof of Lemma~\ref{lem:simpleprep}] 
Since Algorithm~\ref{alg:residual} uses $O(n)$  value queries (Lemma~\ref{lem:residual_runtime}), and Algorithm~\ref{alg:cp_prep} uses in expectation $O(\nicefrac{n}{\eps})$ value queries (Lemma~\ref{lem:cp_prep}), it follows that Algorithm~\ref{alg:fast_prep} uses in expectation 
 $O(\nicefrac{n}{\eps})$ value queries.  Additionally,  the algorithm always returns a set $P\in \cI$ since Algorithm~\ref{alg:cp_prep} always returns such a set. 

 We are left to show the algorithm satisfies Properties (1) and (2) of the lemma. Let $\Psi$ be the event in which $\frac{1}{8}\cdot f(\OPT) \leq \app\leq f(\OPT)$.

 Each time the Residual Random Greedy is executed in Step~\ref{fast:iterations}, it returns a solution $S$ whose value satisfies $0\leq f(S)\leq f(\OPT)$ and $\E[f(S)]\geq\nicefrac{1}{4}\cdot f(\OPT)$. Therefore, by Markov inequality 
$$
\Pr\left[f(S)< \frac{1}{8}\cdot f(\OPT)\right] 
\leq \frac{6}{7}.
$$
This implies that 
$$
\Pr[\textnormal{not } \Psi] =\Pr\left[ \app < \frac{1}{8}\cdot  f(\OPT) \right] \leq \left(\frac{6}{7}\right)^{\log_{\nicefrac{7}{6}} (\nicefrac{8}{\eps} )} = \nicefrac{\eps}{8}.
$$
Thus, $\Pr[ \Psi] \geq 1-\nicefrac{\eps}{8}$, which establishes  the first property of the lemma. 

By Lemma~\ref{lem:cp_prep} it holds that 
$$
\E\left[(1-\nicefrac{1}{e})f(\OPT/P) + \frac{1}{e}\cdot f(P) ~|~\Psi \right] \geq (1-\nicefrac{1}{e} - \nicefrac{\eps}{2} ) \cdot f(\OPT),
$$
therefore, 
$$
\begin{aligned}
\E\left[(1-\nicefrac{1}{e})f(\OPT/P) + \frac{1}{e}\cdot f(P)\right]  &\geq \Pr [\Psi] \cdot \E\left[(1-\nicefrac{1}{e})f(\OPT/P) + \frac{1}{e}\cdot f(P)~|~\Psi \right] \\
&\geq (1-\nicefrac{\eps}{8} )\cdot   (1-\nicefrac{1}{e} - \nicefrac{\eps}{2} ) \cdot f(\OPT) \\
&\geq  (1- \nicefrac{1}{e} - \eps) \cdot f(\OPT),
\end{aligned}
$$
which show the second property of the lemma.

\end{proof}

\subsection{The Construction of the Check Points }
\label{sec:checkpoints}

In this section we provide the construction of the set of checkpoints $\checkpoints(k,\eps)$ used by Algorithm~\ref{alg:cp_prep}. That is, we prove Lemma~\ref{lem:checkpoints}.
\begin{proof}[Proof of Lemma~\ref{lem:checkpoints}]
For every $\eps>0$ and $k\in \mathbb{N}$ define 
$$
\checkpoints(k,\eps) = \left\{ \ceil{k\cdot (1-\beta^\ell) }~|~\ell\in \mathbb{N}_{\geq 0 } \right\},
$$
where $\beta= 1-\nicefrac{\eps}{4}$.  We  show the set satisfies the two properties in the lemma.

We use a simple arithmetic to show the first property,
$$
\sum_{i\in \checkpoints(k,\eps)} (k-i)  \leq \sum_{\ell \in \mathbb{N}} \left(k-\ceil{k\cdot (1-\beta^\ell) }\right )  \leq \sum_{\ell \in \mathbb{N}} \left(k-k\cdot (1-\beta^\ell) \right ) \leq k \sum_{\ell =0}^{\infty} \beta^{\ell} = k\cdot \frac{1}{1-\beta }= k\cdot \frac{4}{\eps}.
$$
Thus, the set $\checkpoints(k,\eps)$ satisfies the first property in  the lemma. 

As for the second property, let $i\in \{0,1,\ldots, k\}$. If $i\in \checkpoints(k,\eps)$ then $s(i) = i$, and therefore $\sum_{i'=i+1}^{s(i)} \frac{1}{k-(i'-1) } =0\leq \eps$ and the property holds. 
Hence, we are left to handle the case that $i\notin \checkpoints(k,\eps)$. In this case, there is $\ell\in\mathbb{N}_{\geq 0}$  such that 
\begin{equation}
\label{eq:cp_strict}
\ceil{ k \cdot (1-\beta^{\ell -1} )} < i < \ceil{k \cdot (1-\beta^{\ell }) }.
\end{equation}
Denote $p(i)= \ceil{ k \cdot (1-\beta^{\ell -1} )}$ and observe that $s(i)=\ceil{k \cdot (1-\beta^{\ell })}$.
Since all the numbers in \eqref{eq:cp_strict} are integers we also have,
\begin{equation}
\label{eq:cp_weak}
p(i)+1\leq i \leq s(i)-1.
\end{equation}
Therefore,
\begin{equation}
\label{eq:cp_sum_first}
\begin{aligned}
\sum_{i'=i+1}^{s(i)}  \frac{1}{k-(i'-1)} &\leq \sum_{i'=i+1}^{s(i)}  \frac{1}{k-(s(i)-1)} \\
&= \frac{s(i)- i}{k-s(i)+1} \\
&\leq \frac{s(i) -p(i) - 1 }{k-s(i)+1}\\
&= \frac{-k+s(i)  - 1 +k-p(i)}{k-s(i)+1} \\ 
& = -1 + \frac{k-p(i)}{k-s(i)+1 },\\
\end{aligned}
\end{equation}
where the second inequality holds due to \eqref{eq:cp_weak}. 

Furthermore, it holds that 
\begin{equation}
\label{eq:cp_prev}
k-p(i) = k- \ceil{ k \cdot (1-\beta^{\ell -1} )} \leq  k- \ k \cdot (1-\beta^{\ell -1} )= k\cdot \beta^{\ell-1}
\end{equation}
and 
\begin{equation}
\label{eq:cp_succ}
k-s(i)+1  = k -\left( \ceil{k \cdot (1-\beta^{\ell }) }-1 \right)  \geq k -k \cdot (1-\beta^{\ell }) = k\cdot \beta^{\ell }. 
\end{equation}

Incorporating \eqref{eq:cp_prev} and \eqref{eq:cp_succ} into \eqref{eq:cp_sum_first} we get
$$
\sum_{i'=i+1}^{s(i)}  \frac{1}{k-(i'-1)}  \leq 
-1 + \frac{k-p(i)}{k-s(i)+1 } \leq -1 +\frac{ k\cdot \beta^{\ell-1}}{k\cdot \beta^{\ell }} = -1+ \frac{1}{\beta} =  \frac{ 1-\beta}{\beta} \leq \eps,
$$ 
where the last inequality holds as $\beta \geq \nicefrac{1}{2}$ and $1-\beta =\nicefrac{\eps}{4}$. This shows the set $\checkpoints(k,\eps)$ also satisfies the second property of the lemma, and completes the proof.

\end{proof}
\end{document}